\documentclass[aps,prl,reprint,superscriptaddress,nofootinbib,floatfix]{revtex4-2}

\usepackage{graphicx}
\usepackage{soul}
\usepackage{svg}
\usepackage{float}
\usepackage{tabularx,booktabs}
\usepackage{tcolorbox}
\usepackage[inline]{enumitem}
\usepackage{fancybox}
\usepackage{nicefrac}
\usepackage{circledsteps}

\usepackage{titlesec}
\usepackage[title]{appendix}
\renewcommand{\appendixname}{Anexo}

\usepackage{amsmath}
\usepackage{amsthm}
\usepackage{bm}
\theoremstyle{plain}
\newtheorem{theorem}{Theorem}
\newtheorem*{proposition*}{Proposition}
\theoremstyle{remark}
\newtheorem{lemma}{Lemma}

\newtheorem{corollary}{Corollary}

\newtheorem{definition}{Definition}
\newtheorem{example}{Example}

\usepackage{mathrsfs}
\usepackage{physics}
\usepackage{mathtools}
\usepackage{mathbbol}
\DeclareSymbolFontAlphabet{\mathbbol}{bbold}
\usepackage{amssymb}
\DeclareSymbolFontAlphabet{\mathbb}{AMSb}
\usepackage{dsfont}
\usepackage{lipsum}

\usepackage[colorlinks=true,urlcolor=blue,citecolor=blue,linkcolor=blue]{hyperref}

\usepackage{soul}
\usepackage{comment}

\newcommand{\step}[1]{\textbf{\textit{Step #1.}}}

\makeatletter
\newcommand\thefontsize{\f@size pt}
\makeatother

\newcommand{\qtrace}{\Tr}

\newcommand{\CFI}{\mathcal{F}}
\newcommand{\QFI}{\mathcal{J}}
\newcommand{\QFIsp}{{\mathcal{J}_{\vert \vert, \theta}}}
\newcommand{\QFIext}{{\mathcal{J}_{\mathrm{Ext}, \theta}}}
\newcommand{\QFItr}{{\mathcal{J}_{\mathrm{Tr},\theta}}}
\newcommand{\sQFI}{\mathcal{I}_\theta}
\newcommand{\CFImax}{{\mathcal{F}_\theta}_\mathrm{max}}

\def\ANU{Centre for Quantum Computation and Communication Technology,
Department of Quantum Science and Technology, Australian National University, Canberra, ACT 2601, Australia}

\def\ASTAR{Quantum Innovation Centre (Q.InC), Agency for Science Technology and Research (A\!*STAR), 2 Fusionopolis Way, Innovis \#08-03, Singapore 138634, Singapore}

\def\QIST{Center for Quantum Technology, Korea Institute of Science and Technology, Seoul 02792, Republic of Korea}
\def\KUST{Division of Quantum Information, KIST School, Korea University of Science and Technology, Seoul 02792, Republic of Korea}
\def\CQT{Centre for Quantum Technologies, National University of Singapore,
3 Science Drive 2, Singapore 117543.}

\usepackage{etoolbox}

\let\oldsec=\section
\let\oldsubsec=\subsection
\let\oldsubsubsec=\subsubsection
\makeatletter
\renewcommand\section[1]{%
  \par
  \vspace{1.5ex}%
  {\raggedright\normalfont\large\bfseries #1\par}%
  \vspace{1ex}%
}

\renewcommand\subsection[1]{%
  \par
  \vspace{1.5ex}%
  {\raggedright\normalfont\bfseries #1\par}%
  \vspace{1ex}%
}

\renewcommand\subsubsection[1]{%
  \noindent
  {\bfseries #1.}%
}
\makeatother

\begin{document}

\title{Precision Bounds for Characterising Quantum Measurements}

\author{Aritra Das}
\email{aritra.das@anu.edu.au}
\affiliation{\ANU}
%\author{Jun Suzuki}
%\affiliation{\UEC}
\author{Simon K. Yung}
\affiliation{\ANU}
\author{Lorc\'{a}n~O.~Conlon}
\affiliation{\ASTAR}
\affiliation{\CQT}
\author{\"{O}zlem Erk{\i}l{\i}\c{c}}
\affiliation{\ANU}
\author{Angus Walsh}
\affiliation{\ANU}
\author{Yong-Su Kim}
\affiliation{\QIST}
\affiliation{\KUST}
\author{Ping K. Lam}
\affiliation{\ANU}
\affiliation{\ASTAR}
\affiliation{\CQT}
\author{Syed M. Assad}
\affiliation{\ASTAR}
\author{Jie Zhao}
\email{jie.zhao@anu.edu.au}
\affiliation{\ANU}

\date{\today}

%TC:ignore
\begin{abstract}
Quantum measurements, alongside quantum states and processes, form a cornerstone of quantum information processing. However, unlike states and processes, their efficient characterisation remains relatively unexplored.
We resolve this asymmetry by introducing a comprehensive framework for
efficient detector estimation that
reveals the fundamental limits to extractable parameter information
and errors arising in detector analysis - the \emph{detector quantum Fisher information}.
Our development eliminates
the need to optimise for the best probe state,
while highlighting
aspects of detector analysis that fundamentally differ from quantum state estimation.
Through proofs,
examples and
experimental validation,
we demonstrate the relevance and robustness of our proposal
for current quantum detector technologies.
By formalising a dual perspective to state estimation,
our framework
completes and connects the
triad of efficient state, process, and detector tomography,
advancing quantum information theory with broader implications for
emerging technologies
reliant on precisely calibrated measurements.
\end{abstract}
%TC:endignore

\maketitle

Measurements hold a special place in quantum mechanics,
bridging abstract
quantum states and real-world classical observations.
This probabilistic transition
from the quantum to the classical, originally postulated by the Born rule~\cite{vonNeumann},
typically washes out key quantum features like superposition and entanglement, thereby limiting the information that can be extracted from a quantum system~\cite{Alfredo20}. The modern theory of quantum measurements has evolved well beyond the `observe-and-collapse' paradigm~\cite{Wiseman2009,NC10,LSP+11},
encompassing generalised measurements such as
weak measurements that blur the line between observation and interaction~\cite{Aharonov1988,Ritchie1991,Pryde2005,Hosten2008,LSP+11}.
Experimental practice has lent further credence to the utility of these generalised  measurements~\cite{Pryde2005,Hosten2008,LSP+11,KBR+11}, for instance by showing the advantage of entangling measurements for the precise characterisation of quantum states~\cite{Hou2018,Mansouri2022,Conlon2023,Conlon2023b,Conlon2023c}.

The measurement of a quantum state constitutes one part of the triad of quantum states, processes and detectors that forms the basis of any quantum information protocol~\cite{Feito2009}. From quantum estimation theory, the Quantum Fisher Information (QFI) provides a metric for the distinguishability of parametrised quantum states and processes, leading to bounds on estimation errors known as quantum Cram\'{e}r-Rao bounds (CRBs)~\cite{Paris2009,Meyer2021}. Strikingly, despite the fundamental and practical importance of quantum detectors, whether similar bounds exist for general quantum detectors remains an open question. In this paper we answer this question. More precisely, we ask whether  information-theoretic precision bounds exist for the accurate estimation of detector parameters.

The aforementioned disparity is particularly surprising given the
dual nature of states and measurements in quantum theory~\cite{Alfredo20}---a symmetry
suggesting their informational properties should be balanced. However, while provably optimal state estimation protocols are known, existing approaches to detector estimation overlook the efficiency of the process~\cite{Luis1999,MLEQM01}, and thus fail to inform optimal estimation strategies. Optimal detector estimation strategies are essential to extract maximal benefit from future quantum devices. For example, precise characterisations of state preparation and measurement (SPAM) errors on quantum computing platforms hinge on reliable detector reconstructions~\cite{LWH+21}. Large uncertainties in reconstructed detectors propagate into SPAM error estimates~\cite{Chen2019} that could limit platform performance and impede effective error mitigation. Similarly, photonic experiments rely on photodetectors~\cite{Lloyd2008,Lundeen2008,Piacentini2015,Schapeler2021,Endo2021} whose precise calibration must be known a priori. Although coherent state probes are one option for photodetector tomography~\cite{Lundeen2008,Feito2009,Piacentini2015}, it remains unclear to what extent quantum properties of light
such as entanglement and squeezing
could improve this precision~\cite{DAriano2004}.
These questions underscore the need
for a general framework
for efficient and high-precision reconstruction of the measurement operators.

In this work, we formulate and prove
the maximum information extractable
from probing unknown quantum measurements
and call it the \emph{detector quantum Fisher information}
(DQFI). The DQFI leads to fundamental limits,
that is, the quantum Cram\'{e}r-Rao bounds (QCRBs),
on the uncertainties in locally estimating parameters of measurement operators~\cite{Helstrom1967,Helstrom1968,Helstrom1969, Helstrom1974,Luis1999,MLEQM01},
thereby setting a performance
benchmark for past and future detector tomography experiments. In addition to providing the first fundamental bound for detector estimation, the DQFI enables us to draw the following physical insights: 1) We draw a direct analogy between the DQFI and the well-known state quantum Fisher information (SQFI) finding some similarities and some surprising differences. 2) We apply the DQFI to noisy on-off single-photon
detectors~\cite{Lundeen2008,Feito2009,Piacentini2015} and implement the first provably-optimal detector estimation experiment
using the IBM platform. 3)
We show that the proposed DQFI framework provides an alternative approach for optimal sensing of a quantum process~\cite{Meyer2021},
which requires both optimal probe states and optimal measurements (Fig.~\ref{fig:tomographytriad}(c)).
In practice, our alternative approach
could lead to more appropriate benchmarks
in settings where probe preparation capabilities
by far exceed control over the measurement device.

Overall, our work reveals surprising connections between the theories of state and detector tomography while highlighting the unique aspects of detector analysis.
These findings complete the triad of optimal state, process, and detector tomography
(Fig.~\ref{fig:tomographytriad})~\cite{Lundeen2008,Feito2009}, paving the way for a more balanced role of states and measurements in precision quantum metrology.

\begin{figure}[t]
    \centering
    \includegraphics[width=\columnwidth]{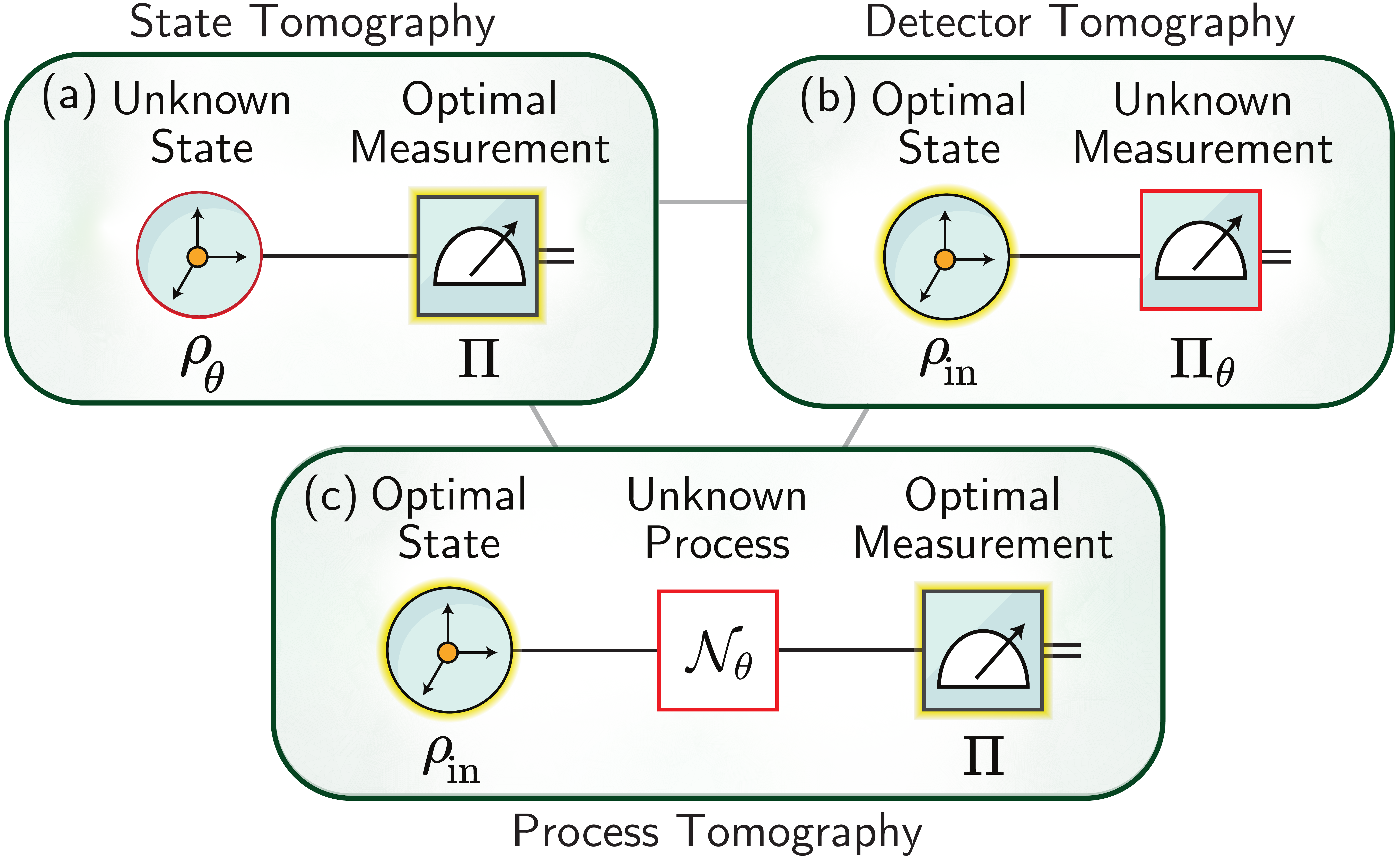}
    \caption{%
    \textbf{Triad of quantum state,
    detector and process tomography.}
     {\bf a,} Characterising an unknown quantum state
    invokes the SQFI, which specifies
    the optimal measurements.
    {\bf b,} Characterising an unknown measurement
    invokes the DQFI, which specifies
    the optimal probe states.
    {\bf c,} Characterising an unknown process
    invokes the process QFI, or maximum of
    output-state SQFI
    over all process-input states,
    simultaneously requiring the optimal
    probe states and measurements.}
    \label{fig:tomographytriad}
\end{figure}

\section{Introduction to the framework}
\label{sec:Background}

\noindent
Before introducing precision bounds
for the estimation of quantum measurements,
we briefly recap the
concept of
efficiency and how it leads to precision bounds
in quantum state
estimation~\cite{Helstrom1967,Helstrom1968,Helstrom1969,
Helstrom1974,Wootters81,BC94,Paris2009}.
The task here (Fig.~\ref{fig:tomographytriad}(a))
is to
estimate parameters~$\theta$
of an unknown state~$\rho_\theta$,
by first measuring the state
via a~$m$-outcome POVM,~$\Pi \equiv \{\pi_1, \dots, \pi_m\}$,
and then processing the
measured probabilities,~${p_\theta(j)} = \qtrace(\rho_\theta \pi_j)$,
to get the parameter estimates~\cite{Paris2009}.
The \emph{efficiency} or effectiveness
of the measurement~$\Pi$
in extracting parameter information
depends on how sensitive its measured statistics
are to changes in the parameter value.
In the local setting, where the unknown
parameters are close to some known true
values~\cite{Helstrom1968,Paris2009},
efficiency
is quantified by the
classical Fisher information (CFI) of the
distribution~\cite{Paris2004,Paris2009,Pezze17,Pezze18},
\begin{equation}
\label{eq:CFIdef}
    \CFI_\theta\left[\rho_\theta, \Pi\right] \coloneqq \sum_{j\in[m]}
    \frac{(\partial_\theta p_\theta(j))^2}{p_\theta(j)} = \mathbb{E}_{p_\theta}\left [ (\partial_\theta \log p_\theta)^2 \right ] \, ,
\end{equation}
capturing the variance
of the parameter-derivative~$\nicefrac{{\partial_\theta p_\theta}}{p_\theta}$
of the distribution
(abbreviating~$\partial_\theta
\coloneqq \frac{\partial}{\partial \theta}$
and~$[m]\coloneqq \{1, \dots, m\}$).
The minimum mean-squared error (MSE)
of parameter estimates achievable
(on average per trial)
using a measurement~$\Pi$ is given by~$1/\CFI_\theta$,
as per the classical CRB (CCRB)~\cite{Helstrom1969}.
Accordingly, optimal measurements are those
that maximise the
CFI in Eq.~\eqref{eq:CFIdef},
thus minimising the CCRB and
yielding the most
precise estimates~\cite{Paris2009,
HayashiOuyang2023,Jun2024}.
However,
maximising this
non-linear quantity is challenging in
practice,
even more so in the multi-parameter case~\cite{HayashiOuyang2023,Jun2024},
so
the state QFI (SQFI)~$\sQFI$ of state~$\rho_\theta$
was introduced as an upper bound to
the maximum CFI over all measurements~\cite{Helstrom1967,Helstrom1968,Helstrom1969}.
The SQFI is defined as~$\sQFI[\rho_\theta] \coloneqq \Tr(\rho_\theta L_\theta^2) = \mathbb{E}_{\rho_\theta} [L_\theta^2]$,
and satisfies
\begin{equation}
    \label{eq:stateestCFIQFI}
      \max_\Pi \CFI_\theta[\rho_\theta, \Pi] \leq \sQFI[\rho_\theta]
\end{equation}
for any valid~$\rho_\theta$,
with equality holding in several cases~\cite{BC94,Paris2009}.
The symmetric logarithmic derivative (SLD)
operator~${L_\theta}$ appearing in
the definition generalises the
classical logarithmic derivative~${\partial_\theta} {\log p_\theta}$
in Eq.~\eqref{eq:CFIdef}
implicitly via~\mbox{$L_\theta \rho_\theta {+} \rho_\theta L_\theta {=} 2 \partial_\theta \rho_\theta$}.
Following
Eq.~\eqref{eq:stateestCFIQFI},
the QCRB given by~$1/\sQFI$ lower-bounds
the MSE of estimates
for any measurement~$\Pi$,
thus setting a precision bound~\cite{Helstrom1969}.

\section{Results}
\label{sec:Methods}
\noindent
In the following, we introduce a general framework
for bounding estimation precision for
parameters of quantum measurements.
To start off, we consider detector models
with a single unknown parameter; these are
relevant for practical single-photon detectors
and noisy qubit detectors.
Then, we
generalise the framework to multi-parameter scenarios, addressing
general quantum detector tomography.

\subsection{Single-parameter detector estimation}
\noindent
The task here
(Fig.~\ref{fig:tomographytriad}(b))
is to characterise an
unknown $m$-outcome measurement
$\Pi_\theta \equiv \{{\pi_j}_\theta\}_{j\in[m]}$
by first probing it with a known quantum
state~$\rho_\mathrm{in}$, and then processing the
measured probabilities,~${p_\theta(j)} = \qtrace(\rho_\mathrm{in} {\pi_j}_\theta)$,
to obtain parameter estimates.
The \emph{efficiency} of the probe
state~$\rho_\mathrm{in}$ in extracting
parameter information from~$\Pi_\theta$
is determined by the CFI,~$\CFI_\theta[\rho_\mathrm{in}, \Pi_\theta]$,
which leads to the CCRB,~$1/\CFI_\theta$,
lower-bounding MSE when using this probe.
Thus, the optimal or most precise probe state,~$\rho^\mathrm{opt}$,
is the one maximising the CFI,
i.e.,
\begin{equation}
\label{eq:maxCFIdef}
\begin{aligned}
    \rho^\mathrm{opt}[\Pi_\theta] &\coloneqq \arg \max_{\rho_\mathrm{in}} \CFI_\theta[\rho_\mathrm{in}, \Pi_\theta] \, , \\
    {\CFI_\theta}_\mathrm{max}[\Pi_\theta] &\coloneqq \max_{\rho_\mathrm{in}} \CFI_\theta[\rho_\mathrm{in}, \Pi_\theta] \, .
\end{aligned}
\end{equation}
The maximisation in Eq.~\eqref{eq:maxCFIdef}
has no known analytical solution,
though numerical techniques
from channel literature may be applicable.
Specifically, by translating
the detector POVM to Kraus operators
for quantum-classical maps,
the maximisation for~${\CFI_\theta}_\mathrm{max}$
can, in principle, be converted to a maximisation over all equivalent
Kraus operators
through a
numerical semi-definite program~\cite{RDD2012}.
In practice, however, this approach is computationally demanding
due to the large number of Kraus operators for a measurement channel,
and leaves the question of the optimal probe~$\rho^\mathrm{opt}$
unanswered.
The lack of efficiently computable
precision bounds
for the minimum MSE,~$1/{\CFI_\theta}_\mathrm{max}$,
for detector estimation
reveals a gap in our understanding of
the information content of measurements
and resulting fundamental limits on
estimation errors.

We now clarify how the domain
of maximisation in Eq.~\eqref{eq:maxCFIdef}
corresponds to
various probing strategies for detector estimation.
Let~$\mathcal{H}_d$ denote the~$d$-dimensional Hilbert space
of probe states and~$\mathcal{D}(\mathcal{H}_d)$
the space of density matrices on~$\mathcal{H}_d$.
The simplest probing strategy is to pick
a single state,~$\rho_\mathrm{in}$, from~$\mathcal{D}(\mathcal{H}_d)$
to probe the detector, repeatedly and independently.
Alternatively, one may use
an ensemble of~$p$ different
probe states,~$\rho_k \in \mathcal{D}(\mathcal{H}_d)$
for~$k=1, \dots, p$,
each with probability~$q_k$~($\sum_k q_k=1$).
Which strategy is ultimately favoured?
As we show in
Lemma~\ref{lemma:singleoptimalprobe} in Methods,
from the convexity of the CFI~\cite{Pezze18},
a single quantum state is sufficient to be optimal.

Central to the CFI and the SQFI
is the logarithmic derivative quantity:
both information measures
capture the variance
of this quantity,
as in Eq.~\eqref{eq:CFIdef}.
To quantify the information
content of a measurement,
we first introduce
logarithmic-derivative operators
for each measurement operator.
For each outcome~$j\in[m]$
of $\Pi_\theta$, we
define an SLD operator~${L_j}_\theta$ via
\begin{equation}
\label{eq:SLDdetector}
    {L_j}_\theta {\pi_j}_\theta + {\pi_j}_\theta {L_j}_\theta \coloneqq 2 \,  \partial_\theta {\pi_j}_\theta \, .
\end{equation}
This implicit definition
has known solutions in terms
of the measurement operators
and their parameter derivatives
(Eqs.~\eqref{eq:Ljeachpovm},~\eqref{eq:vecunvecSLDsol}
in Methods)~\cite{Safranek2018}.
More importantly,
the SLD operators let us
express the measurement
probabilities~\mbox{$p_\theta(j\vert \rho) = \Tr(\rho {\pi_j}_\theta)$}
and their derivatives~\mbox{$\partial_\theta p_\theta(j \vert \rho) =
\Tr( \rho \,  \partial_\theta {\pi_j}_\theta)
=  \Re\left [\Tr({\pi_j}_\theta \,  \rho \, {L_j}_\theta)\right]$}
as linear operations
on the state and the measurement
(where~$\Re[\, \cdot \, ]$ denotes real part)~\cite{BC94,Paris2009}.
These two expressions, when substituted
into the CFI
in Eq.~\eqref{eq:CFIdef},
open the door to operator inequalities
that can gauge~${\CFI_\theta}_\mathrm{max}$
without
actual optimisation~\cite{Paris2009}.

\subsubsection{DQFI definition}
To upper-bound~$\CFI_\theta$
by an expression independent of~$\rho$,
we borrow from state estimation a chain
of inequalities~\cite{BC94,Paris2009}
that, when applied to
the CFI expressed in terms
of~${\pi_j}_\theta, {L_j}_\theta$, and~$\rho$,
yields
\begin{equation}
\label{eq:shortproofEq1}
    \CFI_\theta[\rho, \Pi_\theta] \leq \sum_{j\in[m]} \left \vert \frac{ \Tr({\pi_j}_\theta \, \rho \, {L_j}_\theta ) }{\sqrt{\Tr(\rho \, {\pi_j}_\theta)}} \right \vert^2 \leq \Tr \left (Q_\theta \rho \right ) \, ,
\end{equation}
where~$Q_\theta \coloneqq \sum_{j\in[m]} {L_j}_\theta {\pi_j}_\theta {L_j}_\theta$
is a positive semi-definite operator.
The first inequality above relies on~$\Re[z]^2 \leq \vert z\vert^2$
for complex number~$z$,
whereas the second is the operator
Cauchy-Schwarz inequality.
In state estimation
we obtain the measurement-independent
SQFI~$\sQFI$ at this point, but here
the trace quantity upper-bounding~$\CFI_\theta$
still depends on~$\rho$.
Our final step settles this
by maximising both sides
of Eq.~\eqref{eq:shortproofEq1}
over the probe space~$\mathcal{D}(\mathcal{H}_d)$,
resulting in~${\CFI_\theta}_\mathrm{max} \leq
\max_{\rho\in\mathcal{D}(\mathcal{H}_d)}   \Tr  ( Q_\theta \rho ) = \Vert Q_\theta \Vert^2_\mathrm{sp}$,
where~$\Vert \, \cdot \, \Vert^2_\mathrm{sp}$
denotes the largest eigenvalue
(or spectral radius)
of~$Q_\theta$.
Moreover, the probe
state maximising~$\Tr(Q_\theta \rho)$
is the eigenvector of~$Q_\theta$
corresponding to its
largest eigenvalue~\cite{footnote1}---this
state is thus always pure.

These results lay the ground for defining
the DQFI:
an upper bound to the maximum CFI
that leads to precision bounds for detector estimation.
We propose two definitions for the DQFI.

\begin{definition}[Spectral DQFI]
\label{def:QFI2}
Define the \emph{spectral} DQFI~$\QFIsp$
of a measurement~$\Pi_\theta \equiv \{{\pi_j}_\theta\}_{j\in[m]}$ as
\begin{equation}
    \QFIsp\left[\Pi_\theta\right] \coloneqq \big \Vert \sum_{j\in[m]} {L_j}_\theta {\pi_j}_\theta {L_j}_\theta \;  \big \Vert^2_\mathrm{sp} \, ,
\end{equation}
where~$\Vert X \Vert^2_\mathrm{sp}$ denotes the spectral radius
or largest eigenvalue of operator~$X$.
\end{definition}

A simpler but less tight
upper bound on~${\CFI_\theta}_\mathrm{max}$
worth consideration is~$\Tr Q_\theta$.
In fact,~$\Tr Q_\theta
= \sum_{j\in[m]} \Tr({\pi_j}_\theta {L_j}_\theta^2)$ resembles the SQFI~$\sQFI$.
Assembling the measurement operators
and their SLD counterparts into~$md\times md$ block-diagonal
operators~$\Pi_{\theta,\mathrm{bd}} \coloneqq
\oplus_{j\in[m]} {\pi_j}_\theta$
and~$L_{\theta,\mathrm{bd}}
\coloneqq \oplus_{j\in[m]} {L_j}_\theta$
reveals~$\Tr Q_\theta = \Tr(\Pi_{\theta,\mathrm{bd}} L_{\theta,\mathrm{bd}}^2) = \sQFI(\Pi_{\theta,\mathrm{bd}})$.
This approach is justified in
that~$\Pi_{\theta,\mathrm{bd}}$
is trace-constant (though not
unit-trace~\cite{footnote2})
like a state and~$L_{\theta,\mathrm{bd}}$
is indeed its SLD operator
when treated as such.
This weaker upper bound
results from treating the
measurement
as an unnormalised, higher-dimensional  state.

\begin{definition}[Trace DQFI]
\label{def:QFI1}
Define the \emph{trace} DQFI~$\QFItr$
of a measurement~$\Pi_\theta \equiv \{{\pi_j}_\theta\}_{j\in[m]}$
as
\begin{equation}
    \QFItr[\Pi_\theta] \coloneqq \sum_{j\in[m]} \Tr\left (  {L_j}_\theta {\pi_j}_\theta {L_j}_\theta \right ) \, .
\end{equation}
\end{definition}

\begin{figure}[t]
    \centering
    \includegraphics[width=\columnwidth]{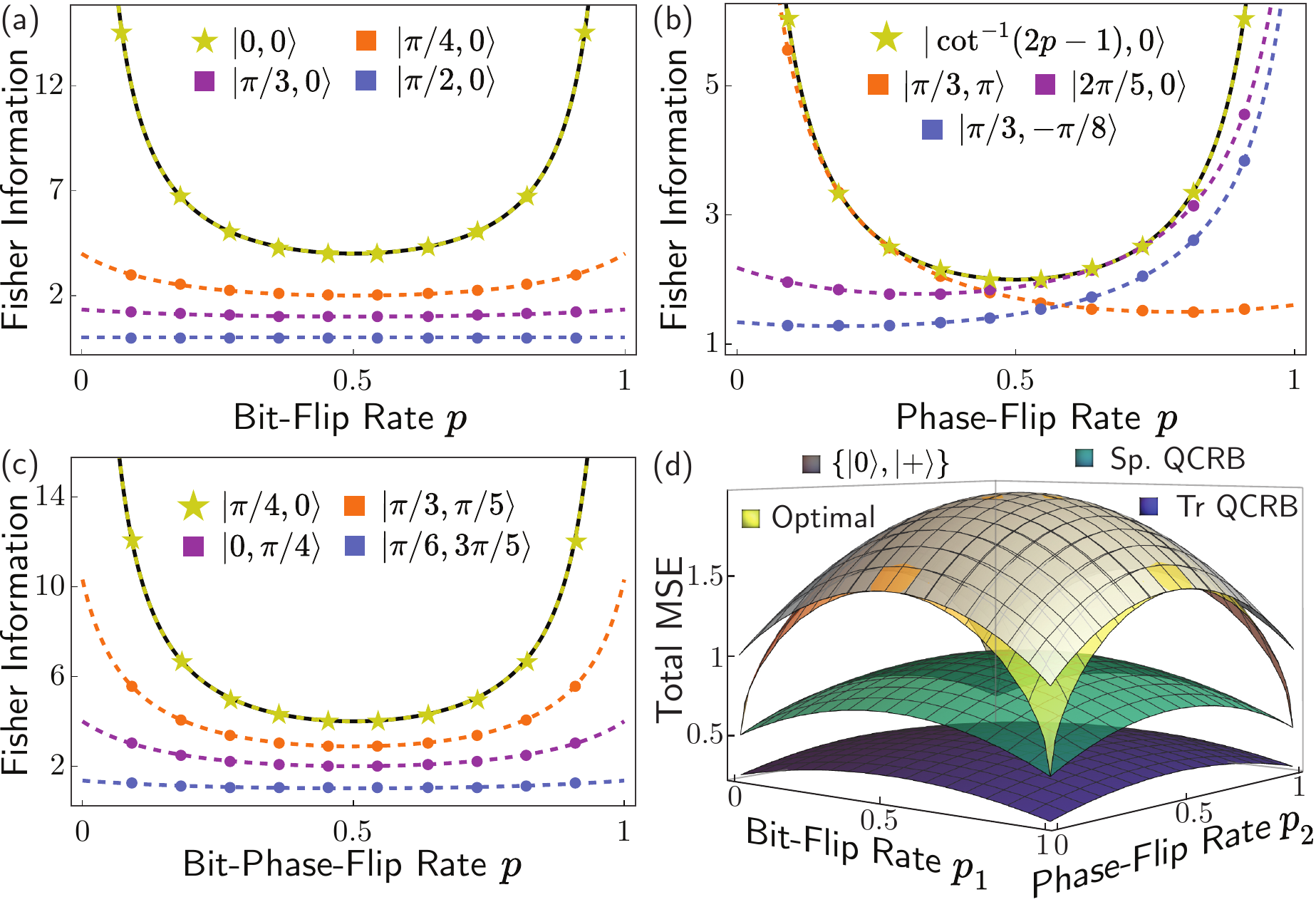}
    \caption{\textbf{Estimating qubit measurements subject to Pauli errors.}
    {\bf a-c,} Noisy~$Z$-measurement with bit-flip error, noisy~${(X+Z)}/{\sqrt{2}}$-measurement with phase-flip error, and noisy~${(X+Z)}/{\sqrt{2}}$-measurement with bit-phase-flip error, respectively.
    The Fisher information for
    probe states notated as~$\ket{\theta,\phi}\coloneqq \cos(\nicefrac{\theta}{2}) \ket{0} + \exp{i\phi} \sin(\nicefrac{\theta}{2}) \ket{1}$ is depicted as circular markers (orange, purple and blue) on dashed lines for non-optimal probes, and as
    star-shaped markers (golden) for the optimal probes.
    The DQFI~$\QFIsp$ (dashed black curve) is tight, as it
    equals the maximum CFI~${\CFI_\theta}_\mathrm{max}$.
    {\bf d,} Two-parameter estimation of a noisy~${(X+Z)}/{\sqrt{2}}$-measurement subject to independent
    bit-flip and phase-flip errors.
    Here, the two detector QCRBs, i.e.,~$\QFItr$ (blue) and~$\QFIsp$ (green) both
    overestimate the tight bound (golden)
    for total MSE.}
    \label{fig:tightQFIexamples}
\end{figure}

Having proposed two definitions
for the DQFI, we must compare the two
to understand their maximum disagreement
and cases where they are equivalent.
It is evident from the positivity of~$Q_\theta$
that~$\QFIsp\leq\QFItr$.
On the other hand,
the largest eigenvalue of~$Q_\theta$
cannot be smaller than
its average eigenvalue,
so that~$\QFItr/d \leq \QFIsp$.
This factor of~$d$ represents the maximum
disagreement and, altogether,
we find the following ordering
of the DQFIs,
\begin{equation}
\label{eq:ordering}
    \frac1{d} \,  \QFItr \leq \QFIsp
    \leq \QFItr \leq d \, \QFIsp \, .
\end{equation}
The cause of the disagreement
clarifies why the spectral DQFI
is tighter---it respects that quantum states are normalised,
so that probing along multiple directions
requires sacrificing some probability of detection,
and hence parameter information,
along each direction.
Nonetheless,
the trace DQFI serves as a simpler upper bound
that is more amenable to analytical
evaluation.

To consolidate the two DQFI definitions, we propose Theorem~\ref{th:upperbound} in the following.
\begin{theorem}[DQFI upper-bounds maximum CFI]
\label{th:upperbound}
    For estimating quantum detectors,
    the spectral DQFI~$\QFIsp$ and the trace DQFI~$\QFItr$
    upper-bound the maximum CFI over probe states,
    ${\CFI_\theta}_\mathrm{max}$, as
        ${\CFI_\theta}_\mathrm{max} \leq \QFIsp \leq \QFItr$.
\end{theorem}
\begin{proof}
    The proof follows from Eqs.~\eqref{eq:shortproofEq1}
    and~\eqref{eq:ordering} and is provided in Methods.
\end{proof}
\noindent
Theorem~\ref{th:upperbound} leads to
the trace QCRB~$1/\mathcal{J}_{\Tr,\theta}$
and the spectral QCRB~$1/\mathcal{J}_{\Vert,\theta}$,
which
lower-bound the minimum MSE of estimates
as~$1/\mathcal{J}_{\Tr,\theta}  \leq 1/\mathcal{J}_{\Vert,\theta} \leq 1/{\CFI_\theta}_\mathrm{max}$,
thereby setting precision bounds for detector estimation.

\begin{example}
\label{eg:bitflipdet}
We now present a simple example to illustrate the importance of the DQFI.
Consider a qubit $Z$-basis detector
with inherent bit-flip noise~\cite{NC10}
of unknown strength~$p$
($0\leq p \leq 1$),
corresponding to the POVM~$\Pi_p \equiv \{{\pi_1}_p, {\pi_2}_p\}$
with
\begin{equation}
    {\pi_1}_p = \begin{pmatrix} 1-p & 0 \\ 0 & p \end{pmatrix}
    \quad \& \quad
    {\pi_2}_p = \begin{pmatrix} p & 0 \\ 0 & 1-p \end{pmatrix} \, .
\end{equation}
Figure~\ref{fig:tightQFIexamples}(a) illustrates the Fisher information
in estimating $p$ for both non-optimal and optimal probe states.
We find that the spectral DQFI~${\mathcal{J}_{\vert \vert, p}}=\nicefrac{1}{p(1-p)}$
(black dashed curve),
and the trace DQFI~${\mathcal{J}_{\mathrm{Tr},p}} = \nicefrac{2}{p(1-p)}$.
The CFI for a generic pure probe state~$\rho$
(coloured dashed curves)
is~$4/(1/\langle Z \rangle^2-(1-2p)^2)$
where~$\langle Z \rangle = \Tr(\rho Z)$
is the Pauli~$Z$-expectation of~$\rho$.
%The CFI is maximised at~$\theta=0$ ($\ket0$) or~$\pi$ ($\ket1$),
The CFI is maximised at~$\langle Z \rangle =\pm 1$,
corresponding to~$\ket0$ or~$\ket1$,
attaining a maximum of~${\CFI_p}_\mathrm{max} = \nicefrac{1}{p(1-p)}$
in either case
(golden points).
As~$Q_p=\nicefrac{1}{p(1-p)} \mathds{1}_2$,
its eigenstates are also~$\ket0$ and~$\ket1$.
Thus, the maximum CFI over probe states
equals~${\mathcal{J}_{\vert \vert, p}}$
and is attained by probes~$\ket{0}$
or~$\ket{1}$, in
agreement with~${\mathcal{J}_{\vert \vert, p}}$.
Here we
have~\mbox{${\mathcal{J}_{\mathrm{Tr},p}}=2{{\mathcal{F}_p}_\mathrm{max}}=2{\mathcal{J}_{||,p}}$},
displaying
the maximum disagreement possible for~$d=2$.
\end{example}

\subsubsection{Attainability of DQFI}
In the above example,
the spectral DQFI is attainable,
i.e.,~$\QFIsp={\CFI_\theta}_\mathrm{max}$,
and the optimal probe states are among
the basis states. These two statements are
generally true for diagonal or phase-insensitive
measurements,
because inequality~\eqref{eq:shortproofEq1}
is saturated, as we state in the following theorem.
\begin{theorem}[Attainability for diagonal measurements]
\label{th:diagonaltheorem}
    For estimating a phase-insensitive
    detector represented by a diagonal POVM~$\Pi_\theta$,
    the DQFI~$\QFIsp$ is attainable, i.e.,~$\QFIsp = {\CFI_\theta}_\mathrm{max}$,
    and an optimal probe exists within the family
    of basis states.
\end{theorem}
\begin{proof}
    The proof is presented in Supp. Mat.~\ref{supp:CFIoptdiag}.
\end{proof}
\noindent
Phase-insensitive measurements
feature prominently in
detector characterisation experiments,
in part because they do not require
phase stabilisation of probes~\cite{Lundeen2008,Feito2009}.
Indeed,
avalanche photodiodes (APDs)~\cite{Lundeen2008,Feito2009},
multiplexed photon-number-resolving detectors~\cite{Lundeen2008,Feito2009,Piacentini2015},
and superconducting nanowire single-photon detectors~\cite{Schapeler2020,Schapeler2021,Endo2021}
have been experimentally
characterised
in this setting.
Given the extensive usage
of these detectors
across quantum technologies,
understanding the information-theoretic bounds for
estimating them has practical implications for benchmarking
existing approaches to detector characterisation
and designing more effective strategies. Our result in
Theorem~\ref{th:diagonaltheorem} means that
the spectral QCRB
sets the ultimate precision limit
here.

The proof of
Theorem~\ref{th:diagonaltheorem}
also covers measurements
that are simultaneously diagonalisable
independent of the parameters
(see Supp. Mat.~\ref{supp:CFIoptdiag}~\cite{footnote3}).
This category
includes
qubit detectors with phase-flip noise
(Fig.~\ref{fig:tightQFIexamples}(b))~and
bit-phase-flip noise
(Fig.~\ref{fig:tightQFIexamples}(c))~\cite{NC10},
so the spectral DQFI
is tight for all single-Pauli error
qubit detectors. More generally,
we find a necessary condition for
attainability of the spectral DQFI to be
that the SLD operators
commute,~$[{L_j}_\theta, {L_k}_\theta]=0$, akin to
`compatible' multi-parameter state
estimation~\cite{Paris2009,Pezze17,Liu2019,Conlon2023b,conlon2025role}
(see Theorem~\ref{th:attaincrit} in Methods
for the complete set of necessary and sufficient
attainability criteria).
%\st{Interestingly}
However, if the SLD operators do not commute,
the spectral DQFI could be unattainable even in the single-parameter regime,
which is not the case in state estimation.
Interestingly, in this case,
the attainable bound can be computed through a semidefinite program
(Supp. Mat.~\ref{sec:tightboundsdp})
that effectively incorporates non-Hermitian components
into the SLD operator definition in Eq.~\eqref{eq:SLDdetector}
(see Methods).
%\st{In summary, we
%provide the complete set of necessary and sufficient criteria for attainability of the spectral DQFI
%(Theorem}~\ref{th:attaincrit} \st{in Methods),
%which is typically more general than~${L_j}_\theta$ commuting.}

\begin{figure*}[t]
    \centering
    \includegraphics[width=\textwidth]{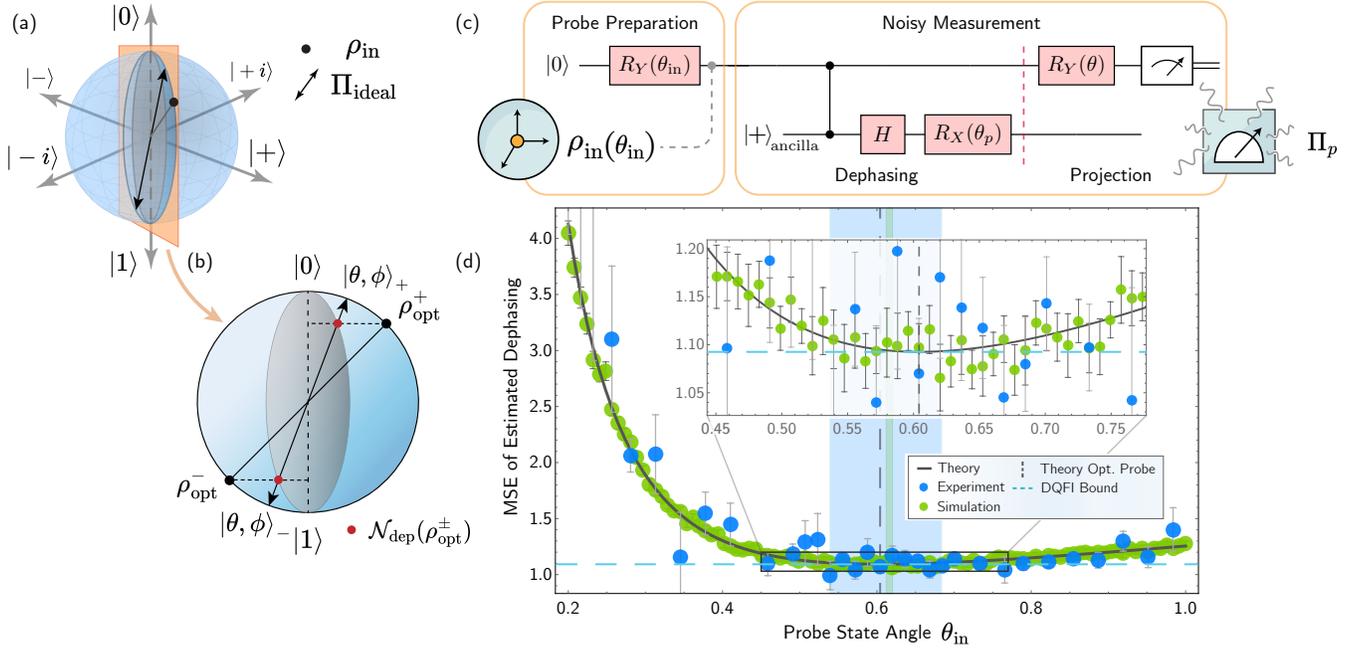}
    \caption{\textbf{DQFI applied to
    dephased projection measurements
    of qubits.}
    {\bf a,} The optimal probe state
    for estimating dephasing strength~$p$
    lies in a vertical section (orange plane)
    of the Bloch sphere containing
    the noiseless measurement~$\Pi_\mathrm{ideal}$
    corresponding to~$p=0$.
    {\bf b,} The intersection that contains the measurement projectors~$\ket{\theta,\phi}_{\pm}$
    (arrowheads) depicts the dynamics of probes
    under dephasing:
    the dephasing action~$\mathcal{N}_\mathrm{dep}$
    contracts the blue disk horizontally
    (towards the vertical dashed line)
    by a factor of~$(1-2p)$
    to form the grey elliptic region.
    The optimal probe states~$\rho_\mathrm{opt}^{\pm}$
    (black dots) are those for which
    the dephased states~$\mathcal{N}_\mathrm{dep}(\rho_\mathrm{opt}^{\pm})$
    (red dots)
    align with
    the measurement
    direction~$\ket{\theta, \phi}_\pm$.
    {\bf c,} Circuit for estimating the noisy measurement~$\Pi_p$
    comprised of environment-assisted dephasing~$\mathcal{N}_{\mathrm{dep}}$
    followed by the
    projection~$\Pi_\mathrm{ideal}$. The varying probe states,~$\rho_\mathrm{in}$, are controlled by polar angles~$\theta_\mathrm{in}$.
    (For details of the circuit implementation
    see Methods.)
    {\bf d,} The experimental MSE (blue dots), superimposed by the noiseless simulation results (green dots), agrees well with the theoretical black curve. The MSEs are plotted
    against the DQFI bound (dashed blue line).
    The minimum-MSE probe angle
    from simulation ($90\%$ confidence interval
    shown in green shading)
    and from experiment (in blue shading)
    match the theory, shown as dashed black line.
    The inset shows a zoomed-in range of probe
    angles near this optimal point. The simulation error bars
    are statistical arising
    due to finite samples
    whereas the experimental error bars
    also include the effect of platform noise
    (see Supp. Mat.~\ref{app:IBMExp}).}
    \label{fig:DephasedDetIBM}
\end{figure*}

\subsubsection{Platform demonstration}
We now consider a simple, but practically relevant and experimentally feasible example, of estimating a qubit detector affected by dephasing noise~\cite{NC10}.
Given that dephasing is particularly prominent on current quantum computing platforms~\cite{WUZ17,Krantz19}, its precise characterisation
is crucial for effective noise control and suppression.
Other noise models~\cite{NC10} are deferred to Supp. Mat.~\ref{sec:SuppNoteEgsApplication}.

\begin{example}
\label{eg:example2dephasing}
    Consider a generic dephased
    qubit projective measurement,
    corresponding to the POVM~$\Pi_p \equiv \{{\pi_1}_p, {\pi_2}_p\}$
    with elements
    %\begin{equation}
    %\begin{split}
    \begin{align}
        {\pi_1}_p &= \frac12 {\begin{pmatrix} {1+ \cos\theta} & {e^{-i\phi} (1-2p)\sin\theta} \\ {e^{i \phi} (1-2p) \sin\theta} & {1-\cos\theta}  \end{pmatrix}} \, , \\
        {\pi_2}_p &= \frac12 {\begin{pmatrix} {1- \cos\theta} & {-e^{-i\phi} (1-2p)\sin\theta} \\ {-e^{i \phi} (1-2p) \sin\theta} & {1+\cos\theta}  \end{pmatrix}} \, , \nonumber
    \end{align}
    %\end{split}
    %\end{equation}
    where~$p$ denotes the unknown dephasing strength
    ($0\leq p \leq 1/2$).
    Here, the noiseless measurement~$\Pi_\mathrm{ideal}$
    corresponding to~$p=0$ is assumed known.
    It corresponds to projections along polar and azimuthal
    angles~$\theta$ and~$\phi$
    on the Bloch sphere
    (double arrows in Fig.~\ref{fig:DephasedDetIBM}(b)), via two orthogonal
    projectors,~$\ket{\theta,\phi}_{\pm}$, given by the non-trivial eigenvectors of~${\pi_j}_0$.
    The spectral and trace DQFIs can be calculated as
    \begin{equation}
    \label{eq:DepDetQFIs}
        \QFI_{\Vert, p} = \frac{\sin^2\theta}{p(1-p)} \quad \& \quad
        \QFI_{\Tr, p} = \frac{2\sin^2\theta}{p(1-p)} \, .
    \end{equation}
    The two optimal probe states,~$\rho_\mathrm{opt}^\pm$,
    are both pure and have phase~$\phi$,
    thus lying in the same vertical plane
    as measurement projectors~$\ket{\theta,\phi}_{\pm}$,
    as shown in Fig.~\ref{fig:DephasedDetIBM}(a).
    This observation, magnified in Fig.~\ref{fig:DephasedDetIBM}(b),
    provides an intuitive explanation
    of the optimal probes
    (black dots in
    Fig.~\ref{fig:DephasedDetIBM}(b)):
    they are the pure states
    that align perfectly
    with the measurement direction upon dephasing.
    The polar angles~$\theta_{\pm}$ of
    these  optimal probes,
    given by~$\tan\theta_{\pm} = \pm \tan\theta/(1-2p)$,
    can be deduced geometrically
    from Fig.~\ref{fig:DephasedDetIBM}(b).
\end{example}

We now experimentally
implement Example~\ref{eg:example2dephasing}
on an IBM Eagle r3 quantum computer
as follows.
First we introduce dephasing noise
of fixed strength~$p{=}0.2$
by interacting the probe
qubit with an ancilla qubit (Fig.~\ref{fig:DephasedDetIBM}(c)).
Then we estimate the noise strength
using a range of probe states
including both optimal and non-optimal ones
(see Methods for details).
Critically,
the noiseless measurement is fixed at~$(\theta,\phi)=(\pi/8, 0)$,
differentiating the problem from
state or process estimation.
In Fig.~\ref{fig:DephasedDetIBM}(d),
we compare the experimental MSEs
for each probe state
(blue dots)
to the CCRB~$1/\CFI_p$ (black curve)
that gives the theoretical minimum MSE~\cite{MLEQM01}
for this probe state
and to the MSEs from
noiseless simulation (green dots).
The spectral QCRB~$1/\QFI_{\Vert,p}$
(horizontal blue dashed line)
lower-bounds the MSE for any probe state.
The theory-optimal
(vertical black dashed line),
simulated-optimal and experimentally-inferred
optimal probes agree within confidence
intervals (green and blue regions),
consistent with platform noise.

In contrast to previous efforts to characterise the platform's SPAM errors through detector tomography~\cite{Chen2019}, which do not use optimal quantum states, our work provides the first instance of provably optimal detector tomography.

\subsection{Multi-parameter detector estimation and tomography}
\noindent
In quantum state estimation
there are many
physically motivated scenarios where we will want to
estimate multiple parameters
simultaneously~\cite{Conlon2023,Conlon2023b,vidrighin2014joint,szczykulska2017reaching,vrehavcek2017multiparameter,chrostowski2017super,hou2020minimal,cimini2019quantum}.
Similarly, there exist scenarios where
multi-parameter detector estimation is favoured.
A prominent example is full detector tomography
where the maximal number of unknown parameters
are simultaneously inferred
(see below and Refs.~\cite{Lundeen2008,Feito2009} for more examples).
For state estimation, it is known that measuring
the parameters simultaneously can provide
greater sensitivity compared to measuring them
sequentially~\cite{baumgratz2016quantum}.
Analogously, multiple parameters
of a detector can be successively estimated
using the single-parameter framework~\cite{Albarelli2022},
but achieving the ultimate precision limit
requires the simultaneous estimation
of all parameters.
In this setting,
the CFI is a matrix
reflecting the probe's sensitivity to each parameter individually
as well as correlations in sensitivities to multiple parameters
(defined in Eq.~\eqref{eq:multiparaCFI}, Supp. Mat.~\ref{supp:MultiPara}).
The corresponding QCRBs
minimise a weighted-sum of the different
parameter variances and covariances~\cite{Meyer2021},
but an ensemble of probe states
is generally required (see discussion in Methods).
For simplicity, here, we set the
weight to identity and minimise total MSE
of the parameters,
though our techniques also apply
to arbitrary weight matrices
(see Supp. Mat.~\ref{supp:MultiPara}).
Below, we propose two QCRBs
that lower-bound the
attainable minimum total MSE,
which is defined as the minimum CCRB
over all probing ensembles,
\begin{equation}
\label{eq:CCRB}
\mathcal{C}^\mathrm{CCRB}_* \coloneqq \min \Tr(\CFI_\theta^{-1}) \,.
\end{equation}

We first extend the trace DQFI to the multi-parameter setting.
For estimating~$n$ real
parameters,~$\theta\coloneqq{(\theta_1, \dots, \theta_n)}$,
the trace DQFI from Def.~\ref{def:QFI1}
symmetrically extends to
the~$n\times n$ real matrix~$\QFItr$
with elements
\begin{equation}
\label{eq:multiparatraceQFI}
    ( \QFItr )_{jk} \coloneqq \frac12 \Tr \left [ \sum_{l \in [m]} \left ( L^{l}_{\theta_j} \pi_l L^{l}_{\theta_k}  +  L^l_{\theta_k} \pi_l L^l_{\theta_j} \right )\right ] \, ,
\end{equation}
where~$L^{l}_{\theta_j}$
is the SLD operator
for the~$l^\text{th}$ outcome
and the~$j^\text{th}$ parameter.
This matrix
upper-bounds the CFI matrix
for any probing strategy as
\mbox{${\QFItr} - {\CFI_\theta} \succcurlyeq 0$}
(Theorem~\ref{th:multipara}
in Supp. Mat.~\ref{supp:MultiPara}),
where~$A \succcurlyeq 0$ denotes~$A$
to be positive semi-definite.
The trace QCRB,~\mbox{${\mathcal{C}^\mathrm{QCRB}_{\Tr}} \coloneqq
{\Tr\big(\mathcal{J}_{\Tr,\theta}^{-1}\big)}$},
thus
lower-bounds $\mathcal{C}^\mathrm{CCRB}_*$ in Eq.~\eqref{eq:CCRB}.
However, due to
the limitation
discussed below Eq.~\eqref{eq:ordering},
the trace QCRB is generally unattainable,
as in
Example~\ref{eg:example2multi}.

Next we consider the more complex problem of extending the spectral DQFI to
the spectral QCRB, denoted~${\mathcal{C}}^{\mathrm{QCRB}}_{\Vert}$.
This is done
through an efficient semi-definite program
(in Supp. Mat.~\ref{supp:MultiPara}),
where we define a linear objective that is maximised,
symmetrically extending~${\Tr(Q_\theta \rho)}$
from Eq.~\eqref{eq:shortproofEq1}
to multiple parameters.
This results in a bound tighter than
the trace QCRB, see Supp. Mat.~\ref{supp:MultiPara} for details.
Notably, the spectral QCRB
does not require considering an ensemble of probes,
unlike the minimisation in Eq.~\eqref{eq:CCRB},
thus offering a practical computational advantage.

Example~\ref{eg:example2multi} below compares
the two QCRBs to~${\mathcal{C}^\mathrm{CCRB}_*}$,
showing the spectral bound to be tight for an
experimentally relevant problem~\cite{Lundeen2008,Feito2009,Piacentini2015}.

\begin{example}
\label{eg:example2multi}
Consider the tomography
of a phase-insensitive,
on-off qubit detector.
The POVM~$\Pi_\theta \equiv \{{\pi_1}_\theta, {\pi_2}_\theta\}$ representing
such a measurement may be
parametrised as
\begin{equation}
    {\pi_1}_\theta = \begin{pmatrix}
        \theta_1 & 0 \\ 0 & \theta_2 \end{pmatrix} \quad \& \quad
    {\pi_2}_\theta = \begin{pmatrix}
        1-\theta_1 & 0 \\ 0 & 1-\theta_2 \end{pmatrix} \, .
\end{equation}
The trace DQFI
for simultaneously estimating~$\theta \equiv \{\theta_1, \theta_2\}$
is
\begin{equation*}
     \QFItr = {\begin{bmatrix} {\frac{1}{\theta_1(1-\theta_1)}} & 0\\
     0 & {\frac{1}{\theta_2(1-\theta_2)}} \end{bmatrix}} \, ,
\end{equation*}
and the resulting QCRB
is~${\mathcal{C}}^{\mathrm{QCRB}}_{\Tr}
= {\sum_{j=1}^2} \theta_j(1-\theta_j)$.
The spectral QCRB
for this problem is
\begin{equation}
\label{eq:multiexampleQCRB}
\begin{split}
    {\mathcal{C}}^{\mathrm{QCRB}}_{\Vert}
    &= \bigg [ {\sum_{j=1}^2} \sqrt{\theta_j (1-\theta_j)} \bigg]^2 = \Tr(\sqrt{\QFItr^{-1}})^2  \\
    &= {\mathcal{C}}^{\mathrm{QCRB}}_{\Tr} + 2 \sqrt{\theta_1 \theta_2 (1-\theta_1) (1-\theta_2)} \, .
\end{split}
\end{equation}
The spectral QCRB
is attainable, i.e.,~\mbox{${\mathcal{C}}^{\mathrm{CCRB}}_* = {\mathcal{C}}^{\mathrm{QCRB}}_{\Vert}$}.
The optimal probing strategy requires
an ensemble of two states:
$\ket{0}$ and~$\ket{1}$
prepared with probabilities~$q$ and~$1-q$
in the proportion~$q/(1-q)=\sqrt{\theta_1(1-\theta_1)}/\sqrt{\theta_2(1-\theta_2)}$.
\end{example}

\begin{figure*}[tb]
    \centering
    \includegraphics[width=\textwidth]{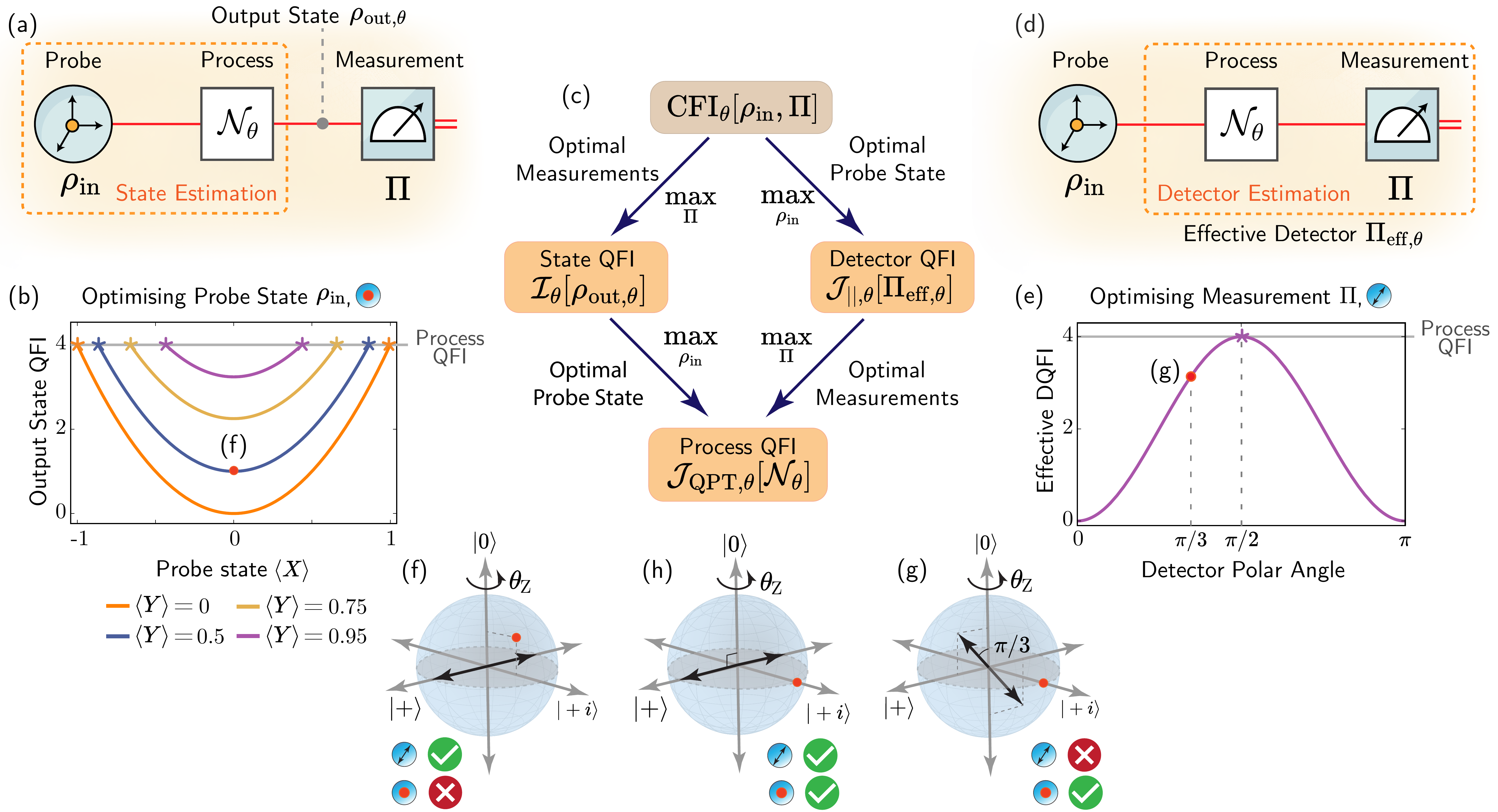}
    \caption{%
    \textbf{Optimal process estimation through the dual approaches of detector and
    state estimation. a,b,}
    Current approaches to achieve optimal process estimation
    that attains the process QFI
    (see Supp. Mat.~\ref{supp:QPT}
    for definition). The approach involves maximising
    the SQFI of the process output
    state over input probes~$\rho_\mathrm{in}$.
    {\bf d,e,} Alternatively, we show that
    optimality can be achieved by maximising the DQFI~$\QFIsp$
    of the effective measurement~$\Pi_{\mathrm{eff}, \theta}$
    over final measurements~$\Pi$. {\bf c,} The two approaches
    treat the optimisation over probes
    and measurements in reverse order
    and are equivalent.
    {\bf b,e,f-h,} An example demonstrating
    the equivalence for estimating qubit~$Z$-rotations.
    In (b), we consider generic mixed qubit states~$\rho_\mathrm{in}$
    with~$\langle X\rangle, \langle Y\rangle$ \&~$\langle Z\rangle$
    Pauli expectation, whereas in (e),
    we consider projection measurements~$\Pi$
    with arbitrary polar angle and zero azimuthal angle
    (from rotational symmetry).
    The state
    approach (b) says, for any given probe state (f),
    the optimal measurement is in the equatorial plane (EP)
    and is unbiased in direction to the state.
    The detector approach (e) says,
    for any given measurement (g),
    the optimal probe lies in the EP
    in a direction unbiased to the measurement.
    The optimal process estimation strategy (h)
    requires both to hold true,
    fixing both probe and measurement in the EP,
    while still unbiased to each other
    (star markers in (b), (e)).}
    \label{fig:processestimation}
\end{figure*}

\noindent
\subsubsection{Probe incompatibility effect}
An interesting feature of
Example~\ref{eg:example2multi} is that the optimal
multi-parameter probing strategy
is a convex mixture of the single-parameter
optimal probes, namely,~$\ket{0}$ and~$\ket{1}$
for~$\theta_1$ and~$\theta_2$, respectively.
Although this may not hold
for generic multi-parameter models,
it underscores the fact that
when the optimal single-parameter probes
for different parameters are incompatible,
termed probe incompatibility in Ref.~\cite{Albarelli2022},
there is an uncertainty trade-off in
simultaneous multi-parameter estimation. Naturally,
we can expect the optimal simultaneous strategy
to generally outperform the optimal strategy that
utilizes a fraction of detector uses
to estimate each parameter separately.
This latter strategy,
which we call sequential multi-parameter estimation,
is connected to the
total QFI metric~\cite{Albarelli2022},
and we compare our %multi-parameter
bounds
to this metric
in Supp. Mat.~\ref{sec:SupMatMultiComparison}.

\subsubsection{Application to photodetector tomography}
Example~\ref{eg:example2multi}
generalises to phase-insensitive measurements
of higher-dimensional states,
such as the APDs experimentally
tomographed in Refs.~\cite{Lundeen2008,Feito2009,Piacentini2015}.
These detectors correspond to diagonal
two-outcome measurements in the photon number basis
for the continuous-variable state
space, truncated at some large dimension~$d_\mathrm{tr}$.
The spectral
QCRB here,~\mbox{${\mathcal{C}}^{\mathrm{QCRB}}_{\Vert} = \big( \sum_{j=1}^{d_\mathrm{tr}} {\sqrt{\theta_j(1-\theta_j)}} \, \big)^2$},
is tight
and
sets the precision limit
for APD tomography.
Attaining this precision
requires preparing an ensemble
of photon-number states~$\{\ket{j}\}$,
%_{j=0}^{d_\mathrm{tr}-1}
each with probability~$q_j \propto \sqrt{\theta_j(1-\theta_j)}$,
and assigning a higher~$q_j$ to states~$\ket{j}$
that produce more uniformly distributed detection probabilities.
%OR more weightage given to
%states for which the detection probabilities
%are more uniformly distributed.
Interestingly,
the spectral QCRB in
Eq.~\eqref{eq:multiexampleQCRB}
mirrors the
Gill-Massar QCRB for state estimation~\cite{Gill2000},
with similar behaviour seen in Fig.~\ref{fig:tightQFIexamples}(d)
for the simultaneous estimation of
bit-flip and phase-flip rates
(Example~\ref{eg:bitphaseflip}
in Supp. Mat.~\ref{supp:MultiPara}).

%\section{Discussion}
\section{Discussion}
\noindent
While we have comprehensively analysed the DQFI, introducing it in both single and multi-parameter estimation, presenting physically motivated examples where it is useful, and demonstrating several novel properties of the DQFI, there remain many more avenues for exploration.

Thus far, we have considered state estimation (Fig.~\ref{fig:tomographytriad}(a)) or detector estimation (Fig.~\ref{fig:tomographytriad}(b)). For a general theory of quantum estimation, it is essential to also consider process estimation (Fig.~\ref{fig:tomographytriad}(c))~\cite{Harper2020,Chen2023}.
The standard approach,
called optimal sensing~\cite{Meyer2021},
can be viewed as a two-step
computation
(left half of Fig.~\ref{fig:processestimation}(c)):
(i) evaluating the process-output state's
SQFI~$\sQFI[\rho_{\mathrm{out}, \theta}]$
(Fig.~\ref{fig:processestimation}(a)),
and (ii) maximising this over input
probe states~$\rho_\mathrm{in}$
(Fig.~\ref{fig:processestimation}(b))~\cite{Hayashi24}.
Our DQFI framework presents an alternative route to process estimation
(right half of Fig.~\ref{fig:processestimation}(c)):
(i) evaluating the effective detector's
DQFI~$\QFIsp[\Pi_{\mathrm{eff}, \theta}]$
(Fig.~\ref{fig:processestimation}(d)),
and (ii) maximising this
over all actual measurements~$\Pi$
(Fig.~\ref{fig:processestimation}(e)).
The equivalence of the two approaches is
illustrated for estimating
a~$Z$-rotation in Figs.~\ref{fig:processestimation}(f)--(h).
Further examples involving non-unitary
processes $\mathcal{N}_{\theta}$ are detailed
in Supp. Mat.~\ref{sec:SuppNoteEgsApplication}.
In practice, high-precision process estimation
balances two optimisation levers: probe state and final measurement.
The detector approach
offers flexibility in selecting an experimentally-feasible subset of measurements
to consider,
useful in experimental
settings with limited measurement capabilities.
Complementarily, the state approach offers
flexibility in selecting feasible probe states,
thus suited to settings with limited probe preparation capabilities.

It is well established that entangled resources can enhance precision in state estimation at both the state preparation~\cite{Giovannetti2006} and measurement stages~\cite{Conlon2023,Conlon2023b,Conlon2023c,Das2024}. In Supp. Mat.~\ref{supp:nonadditive} we have started to quantify similar effects for detector estimation, showing that multi-partite
entangled probe states generally extract more information
than separable strategies, even at
the single-parameter level
(Example~\ref{eg:heisenbergscaling},
Supp. Mat.~\ref{supp:nonadditive}).
However, we also demonstrate that this advantage disappears
for phase-insensitive measurements
(Lemma~\ref{lemma:nocollectiveadvantagediagonalmeasurement},
Supp. Mat.~\ref{supp:nonadditive}).
For single parameters, entanglement
with an ancillary system does not offer any advantage either
(Corollary~\ref{corr:sepopt} in Supp. Mat.~\ref{sec:tightboundsdp}).
Further studies are needed to fully understand the role of
entangled probes in detector estimation~\cite{Escher2011,Zhou2018}.

%\subsubsection{\st{Characterising weak measurements}}
Finally, in this work we have considered only quantum detectors. However, there exist measurements more general
than detectors,
such as weak measurements~\cite{LSP+11,KBR+11}
and
partial measurements of a larger system~\cite{RRG+22}, that
can produce quantum states
in addition to classical outcomes~\cite{Wilde13}. In Supp. Mat.~\ref{supp:QPT}, we show how our formalism can provide lower and upper bounds on the information content of such measurements.

\section{Conclusion and Outlook}
\label{sec:Disc}
\noindent
In developing
a comprehensive framework
for the precise characterisation of quantum measurements,
we close a fundamental open question
in quantum information,
completing the triad of optimal state, detector, and process tomography
(Fig.~\ref{fig:tomographytriad}).
Within this triad,
we illuminate the two connecting
arms of process tomography (Fig.~\ref{fig:processestimation}),
underscoring the complementary nature
of our work and existing results.

Despite this complementarity,
our work reveals crucial differences
between state, process and detector estimation
(see Table~\ref{table:SEvsDE}, Supp. Mat.~\ref{supp:nonadditive}).
A naive extension of existing state- and channel-theoretic tools
to detectors (Def.~\ref{def:QFI1} -- trace bound
and unoptimised channel bound~\cite{Fujiwara2008}, respectively)
produces inferior bounds compared to
our newly-introduced spectral bound
(Def.~\ref{def:QFI2}),
which is needed
for attainability.
In practice, this makes
the spectral bound the
default benchmark because it closely tracks the attainable
bound in several aspects (see numerical
comparisons in  Figs.~\ref{fig:singleparacomb} \&~\ref{fig:singleparacomp}
in Supp. Mat.~\ref{sec:tightboundsdp}),
whereas the trace bound serves
as a more convenient analytical approximation
but can overestimate by a factor
of the dimensions of the system in question.
This contrasts with state estimation, where
the trace-based bound is known to be attainable
in the single-parameter setting,
and also with channel estimation, where
the optimised channel bound is known to
require ancilla systems for attainability in general~\cite{Sarovar2006,Fujiwara2008,RDD2012}.
Indeed, the attainability of the spectral DQFI
resembles that of the multi-parameter SQFI.
Similarly, the utility of entanglement only
present in multi-parameter state estimation appears
in single-parameter detector estimation.
At the same time, the utility of ancilla systems
established for single-parameter channel estimation
disappears in single-parameter detector estimation.
These unique features re-emphasize the fact
that while detector estimation
may be technically considered as a special case of
channel estimation,
generic estimation tools that do not
utilise the quantum-classical nature
of the measurement process may
fail to capture
the unique aspects of detector estimation.

Given
the unique features of detector estimation
identified here,
there is an evident need for
further investigation of
both the single- and multi-parameter
settings in greater detail.
Whereas the Fock basis probes that are generally optimal
for phase-insensitive measurements can be generated using
heralded measurements, low state generation fidelities could
limit the actual enhancement in practical setups. The exploration
of more accessible continuous-variable input probes, including
displaced squeezed probe ensembles, remains the scope of future work.
While our findings suggest that the DQFI
can grow quadratically with number of copies
(Example~\ref{eg:heisenbergscaling}),
the precise strategies including error-correcting techniques
that can preserve this Heisenberg scaling
for realistic noisy detector estimation require further exploration~\cite{Escher2011,RDD2012,Zhou2018}.
Apart from tighter information measures
for multi-parameter measurement models
that could help identify the optimal probe states
for full tomography,
future works could
probe the quantum-geometric properties
of measurements by employing other
information metrics
from the Petz family,
or address the
practical limitations of the local estimation regime
through a Bayesian approach.

Our work is however, not just of
fundamental importance---it has a wide domain of applications.
Previous detector tomography experiments
had no way of certifying their optimality:
our results enable benchmarking these past
experiments to either verify
or deny optimality, while also providing future experiments with a tool to improve their precision.
Using our framework, we demonstrate
the first provably-optimal detector estimation
of dephasing noise on an IBM quantum computing platform.
These results can be generalised to
more complex noise mechanisms.
Moreover, we provide the optimal tomography strategy
for phase-insensitive detectors with
direct applications to characterising imperfect
photonic detectors---a subject of ongoing research.
Beyond absolute benchmarks,
our results facilitate the comparison
of existing theoretical proposals
and simplify the experimental requirements
for demonstrating an entanglement advantage
in detector characterisation.
Given its versatility and practical relevance,
we anticipate rapid adoption of our technique
across quantum technology platforms
where efficient detector calibration is essential,
including quantum computing~\cite{LWH+21} and communication~\cite{Watanabe2008}.

%TC:ignore
\section{Methods}

\subsection{Probing Strategies for Detector Estimation}

%\noindent
\subsubsection{Local Unbiased Estimation}
In the local estimation setting,
where unknown parameters~$\theta$ are
close to known true values~$\theta^*$~\cite{Paris2009,Helstrom1968},
the Fisher information
framework requires locally-unbiased estimators,
meaning estimators with zero bias
at the true parameter values~\cite{Paris2004,Paris2009,Unbiased25}.
Effectively, this constrains how the probabilistic detector outcomes
are mapped to corresponding predicted values
for the unknown parameters~\cite{MLEQM01}.
In linear estimation~\cite{MLEQM01,DAriano2004},
this is done through a classical estimator function
represented by an~$m \times n$ real matrix~$\Xi_{jk}$
that, given a detector outcome~$j \in [m]$,
predicts the estimate~$\Xi_{jk}$ for~$\theta_k$ ($k\in[n]$).
The overall estimate produced for each~$\theta_k$
is then~$ \sum_j {p_j}_\theta \, \Xi_{jk}$.
The local-unbiasedness constraint on the
estimator~\cite{Unbiased25} now
translates to~$(\sum_j {p_j}_\theta \, \Xi_{jk})\vert_{\theta=\theta^*}
= \theta_k^*$.
The widely-used maximum-likelihood estimator~\cite{Lundeen2008,MLEQM01}
is a suitable example, because
it is asymptotically
unbiased, and therefore can saturate the
CCRB in that limit~\cite{MLEQM01}.

\subsubsection{Single-Parameter Detector Estimation}
We first present Lemma~\ref{lemma:singleoptimalprobe},
establishing single-probe strategies to be
sufficient for optimal single-parameter
detector estimation.

%Note that the effective CFI~$\CFI_\theta\left[\{q_k, \rho_k\}, \Pi_\theta\right]$ for this ensemble
%equals the convex sum~$\sum_{k} q_k\, \CFI_\theta[\rho_k, \Pi_\theta]$~\cite{BC94,Pezze18}; we obtain the following lemma
%(see proof in Methods).
\begin{lemma}
\label{lemma:singleoptimalprobe}
    For single-parameter detector estimation,
    a single-probe strategy is optimal
    and
    %there exists
    at least one optimal
    state~${\rho^\mathrm{opt}}\in{\mathcal{D}(\mathcal{H}_d)}$
    exists that
    attains~${\CFI_\theta}_\mathrm{max}$ in Eq.~\eqref{eq:maxCFIdef}.
    %that maximises the CFI~$\CFI_\theta$
    %of the outcome
    %probability distribution
    %and constitutes an optimal probing strategy.
\end{lemma}

\begin{proof}[Proof of Lemma~\ref{lemma:singleoptimalprobe}]
Suppose, to estimate a
single parameter~$\theta$
of a measurement~$\Pi_\theta$, we probe
with two distinct states,~$\rho_1 \neq \rho_2$,~$\rho_1, \rho_2\in\mathcal{D}(\mathcal{H}_d)$,
with probabilities~$q_1$ and~$q_2=1-q_1$.
The joint list of probabilities,
\begin{equation*}
\begin{split}
    p^\mathrm{list} = \{q_1 p_\theta(1\vert \rho_1), \,  &\dots, \,  q_1 p_\theta(m \vert \rho_1), \\
    &q_2 p_\theta(1 \vert \rho_2), \, \dots, \, q_2 p_\theta(m \vert \rho_2)\} \, ,
\end{split}
\end{equation*}
forms a valid distribution
and, as $q_j$ are (by definition) independent of~$\theta$,
we see from Eq.~\eqref{eq:CFIdef} that
the CFI of~$p^\mathrm{list}$,
\begin{equation}
%\begin{split}
\label{eq:additiveCFI}
    \CFI_\theta[p^\mathrm{list}] = \,  q_1  \, \CFI_\theta[\rho_1, \Pi_\theta] %\\
    + \,  (1-q_1) \, \CFI_\theta[\rho_2, \Pi_\theta] \, ,
%\end{split}
\end{equation}
is a convex combination of the CFIs
of each probe~$\rho_k$ individually.
Unless the two individual CFIs
are equal, the combined CFI~$\CFI_\theta[p^\mathrm{list}]$
is maximised at the extreme points~$q=0$
or~$q=1$, depending on which
of the individual CFIs is larger.
Even if the two individual CFIs
are equal,
probing with only~$\rho_1$
or only~$\rho_2$
is sufficient for optimality.
The same argument extends to ensembles
of size~$p>2$, proving that
for single-parameter detector estimation,
a `one-state' probing strategy
is optimal.
\end{proof}
\noindent
Lemma~\ref{lemma:singleoptimalprobe} says
that for single parameters,
the maximisation in Eq.~\eqref{eq:maxCFIdef}
can be performed
over~\mbox{$\rho_\mathrm{in} \in \mathcal{D}(\mathcal{H}_d)$}.

\subsubsection{Multi-parameter detector estimation}
For multi-parameter problems,
Lemma~\ref{lemma:singleoptimalprobe}
does not hold and
an ensemble of probes is generally required for optimality.
This is because multiple unknown parameters cannot be uniquely
determined without sufficiently many
linearly-independent outcome probabilities.
The multi-parameter CFI matrix~$\CFI_\theta$ for an ensemble~$\{q_k, \rho_k\}_{k\in[p]}$
comprising~$p$ probe states
is a convex sum of the component CFI matrices,
\begin{equation}
\label{eq:CFImatconvexsum}
    \CFI_\theta[\{q_k, \rho_k\}_{k\in[p]}, \Pi_\theta] = \sum_{k\in[p]} q_k \,  \CFI_\theta[\rho_k, \Pi_\theta] \, ,
\end{equation}
similar to Eq.~\eqref{eq:additiveCFI}.
Here, the difference from the single-parameter setting
%Eqs.~\eqref{eq:additiveCFI}
%and~\eqref{eq:CFImatconvexsum}
is that the CCRB~$\Tr(\CFI_\theta^{-1})$
requires matrix inversion of~$\CFI_\theta$.
Whereas
the individual CFI matrices~$\CFI_\theta[\rho_k, \Pi_\theta]$
would be singular (due to insufficient
number of
independent probabilities) and non-invertible,
the composite matrix~$\CFI_\theta[\{q_k, \rho_k\}_{k\in[p]}, \Pi_\theta]$
for the ensemble
would be invertible for sufficiently large~$p$.
In particular, for estimating~$n$
parameters from an~$m$-outcome measurement,
the number of different probe states
required is~$p \geq \lceil n/(m-1)\rceil$,
assuming independent parameters and measurement elements.
For two-outcome detectors,
such as the qubit detector
considered in Example~\ref{eg:example2multi},
this means at least~$n$ different probe states
must be included in the ensemble for estimating~$n$ parameters.
However, this minimum number of probes
merely guarantees the feasibility of
multi-parameter estimation, whereas optimality
may require even larger ensembles.

Notably,
an ensemble of~$p$ different
$d$-dimensional probe states
can be realised through
access to an ancillary system, A,
that can be measured perfectly.
This involves preparing
a single bipartite-entangled state
in~$\mathcal{D}(\mathcal{H}_d\otimes \mathcal{H}_p^A)$,
where~$\mathcal{H}_p^A$ denotes
the $p$-dimensional Hilbert space
corresponding to system A.
For example, the ensemble
comprising
states~$\ket0$ and~$\ket1$ with
probabilities~$q$ and~$1-q$ is realised
by probing~$\Pi_\theta$ with
the state~$\sqrt{q} \ket0 \ket{0_A}
+ \sqrt{1-q} \ket1 \ket{1_A}$
and measuring the ancilla perfectly
in the~$\ket{0_A},\ket{1_A}$ basis~\cite{DAriano2004}.
In this sense,
Lemma~\ref{lemma:singleoptimalprobe}
extends
to multi-parameter problems,
with the optimisation domain
becoming~$\rho_\mathrm{in} \in \mathcal{D}(\mathcal{H}_d \otimes \mathcal{H}_p^A)$.

The equivalence between bipartite-entangled states
and ensemble probes means that unless the probe and ancilla are
measured jointly, the convexity of the CFI
(Lemma~\ref{lemma:singleoptimalprobe}) rules out
any entanglement advantage in precision
in the single-parameter setting.
However,
for the most general estimation scenario,
ancilla measurements conditioned on
the detector outcome should be considered.
This
estimation strategy was
applied to detector tomography in Ref.~\cite{DAriano2004}
and extended to
imperfect but pre-characterised ancilla measurements.
In Supp. Mat.~\ref{sec:tightboundsdp},
using Theorem~\ref{th:extdqfitight},
we show that for single-parameter detector estimation,
even conditional measurements cannot enable an entanglement advantage,
and the maximum precision can be attained using a single separable pure-state probe.
At the same time,
the advantages of ancilla-assisted probe
states
for multi-parameter detector estimation
have been addressed in recent
experimental demonstrations~\cite{Altepeter2003,Brida2012},
though the extent of precision enhancement
remains to be explored.

\subsection{Single-Parameter DQFI Derivation}

%\paragraph{Detector SLD Operators:}
The Lyapunov equation~\eqref{eq:SLDdetector}
can be explicitly solved by
utilising either the eigenspectrum of~${\pi_j}_\theta = \sum_{k\in[d]} \lambda_{jk} \ketbra{\lambda_{jk}}$~\cite{Helstrom1967,Helstrom1968,Paris2009},
\begin{equation}
\label{eq:Ljeachpovm}
    {L_j}_\theta = 2 \sum_{\substack{m,n\in[d-1] \, ,\\ \lambda_{jm} + \lambda_{jn} \neq 0}} \ket{\lambda_{jm}} \frac{\bra{\lambda_{jm}} \partial_\theta {\pi_j}_\theta \ket{\lambda_{jn}}}{\lambda_{jm} + \lambda_{jn}} \bra{\lambda_{jn}} \, ,
\end{equation}
%where~$\lambda_{jk}$ and~$\ket{\lambda_{jk}}$
%represent the~$k^\text{th}$ eigenvalue
%and the~$k^\text{th}$ eigenvector
%of the~$j^\text{th}$ POVM element~${\pi_j}_\theta$,
%respectively,
or the vectorisation approach~\cite{Safranek2018},
\begin{equation}
\label{eq:vecunvecSLDsol}
    \mathrm{vec}({L_j}_\theta) = 2 \,  ({\pi_j}_\theta^* \otimes \mathds{1}_d + \mathds{1}_d \otimes {\pi_j}_\theta)^{-1} \, \mathrm{vec}(\partial_\theta {\pi_j}_\theta) \, ,
\end{equation}
where~$( \, \cdot \,  )^*$ denotes complex conjugation
and the pseudo-inverse can be used instead of the
inverse, if required.
Below, we first prove
Eq.~\eqref{eq:shortproofEq1},
and then prove Theorem~\ref{th:upperbound}.

\begin{proof}[Proof of Eq.~\eqref{eq:shortproofEq1}]
To upper-bound~$\CFI_\theta$
by an expression independent of~$\rho$,
we employ the chain of inequalities~\cite{BC94,Paris2009}
\begin{equation}
\label{eq:proofEq1}
\begin{aligned}
    &\CFI_\theta[\rho,\Pi_\theta] = \sum_{j\in[m]} \frac{\left ( \Re \left [ \Tr({\pi_j}_\theta \, \rho \, {L_j}_\theta ) \right ] \right )^2}{\Tr(\rho \, {\pi_j}_\theta)} \\
    \leq  &\sum_{j\in[m]} \left \vert \frac{ \Tr({\pi_j}_\theta \, \rho \, {L_j}_\theta ) }{\sqrt{\Tr(\rho \, {\pi_j}_\theta)}} \right \vert^2\\
    = &\sum_{j\in[m]} \left \vert \Tr \left ( \frac{ \sqrt{{\pi_j}_\theta} \, \sqrt{\rho}}{\sqrt{\Tr(\rho \, {\pi_j}_\theta)}}  \; \;  \sqrt{\rho} \, {L_j}_\theta \sqrt{{\pi_j}_\theta} \right )  \right \vert^2\\
    \leq &\sum_{j\in[m]} \Tr \left ( \frac{\sqrt{{\pi_j}_\theta} \rho \sqrt{{\pi_j}_\theta}}{\Tr(\rho \, {\pi_j}_\theta)}\right ) \Tr \left ( \sqrt{\rho} {L_j}_\theta {\pi_j}_\theta {L_j}_\theta \sqrt{\rho} \right ) \\
    = &\sum_{j\in[m]} \Tr \left ({L_j}_\theta {\pi_j}_\theta {L_j}_\theta \rho \right )  = \Tr \left (Q_\theta \rho \right ) \, .
\end{aligned}
\end{equation}
The first inequality above follows from~$\Re[z]^2 \leq \vert z\vert^2$
for any complex number~$z$,
and the second inequality is the operator
Cauchy-Schwarz inequality,~$\vert \Tr(A^\dagger B)
\vert^2 \leq \Tr(A^\dagger A) \Tr(B^\dagger B)$.
\end{proof}

\begin{proof}[Proof of Theorem~\ref{th:upperbound}]
As any state~$\rho\in\mathcal{D}(\mathcal{H}_d)$
satisfies~$\rho \preccurlyeq \mathds{1}_d$,
it follows that
$\Tr(Q_\theta \rho) \leq \Tr(Q_\theta)$
always.
Consolidating Eq.~\eqref{eq:proofEq1}
and its maximisation,~$\max_{\rho\in\mathcal{D}(\mathcal{H}_d)}   \Tr  ( Q_\theta \rho ) = \Vert Q_\theta \Vert^2_\mathrm{sp}$,
into
\begin{equation}
\label{eq:CFIattainChain}
    \CFI_\theta[\rho, \Pi_\theta] \leq \Tr(Q_\theta \rho) \leq \Vert Q_\theta \Vert_\mathrm{sp}^2 \leq \Tr(Q_\theta) \,
\end{equation}
leads to
\begin{equation}
    {\CFI_\theta}_\mathrm{max}\left [ \Pi_\theta \right ]  \leq \QFIsp[\Pi_\theta] \leq \QFItr[\Pi_\theta] \, ,
\end{equation}
as claimed in Theorem~\ref{th:upperbound}.
\end{proof}

\subsection{Attainability Criteria of the DQFI}
The attainability criteria for
the spectral DQFI~$\QFIsp$
are threefold.
The first two of these conditions
ensure that the two inequalities
in Eq.~\eqref{eq:shortproofEq1}
are saturated, implying
that~${\CFI_\theta}_\mathrm{max} = \Tr(Q_\theta \rho^\mathrm{opt})$,
for~$\rho^\text{opt}$ the CFI-optimal
probe state from Eq.~\eqref{eq:maxCFIdef}.
The third condition, unique to detector estimation,
ensures that~$\Tr(Q_\theta \rho^\mathrm{opt}) = \Vert Q \Vert_\mathrm{sp}^2 = \QFIsp$.
Taken together, these criteria require
the existence of a common eigenstate
to each SLD operator,
that is also common to~$Q_\theta$.
This common eigenstate,
if it exists, corresponds
to the optimal probe state.
We state this result as Theorem~\ref{th:attaincrit}
below.
Main-text Theorem~\ref{th:diagonaltheorem}
is recovered as a special case of Theorem~\ref{th:attaincrit},
where basis or number states comprise the common eigenstates.
\begin{theorem}
\label{th:attaincrit}
    For estimating a detector,
    the DQFI~$\QFIsp$ is attainable if and only if
    the SLD operators~$\{{L_j}_\theta\}$
    of the POVM
    and~$Q_\theta$ all
    share a common eigenvector,
    which constitutes an optimal probe state.
\end{theorem}

\begin{proof}[Proof of Theorem~\ref{th:attaincrit} (Backward Direction)]
It is clear from
Eq.~\eqref{eq:CFIattainChain}
that~$\QFIsp$ is attained
whenever the first two inequalities
in Eq.~\eqref{eq:CFIattainChain}
are saturated.
The first inequality in
Eq.~\eqref{eq:CFIattainChain}
is saturated
if and only if the two inequalities
in Eq.~\eqref{eq:proofEq1} are saturated,
requiring that
\begin{enumerate}
    \item $\qtrace({\pi_j}_\theta \rho {L_j}_\theta)$ is real for all~$j \in [m]$
    and all~$\theta \in \Theta$,
    \item $ \rho {L_j}_\theta {\pi_j}_\theta \propto \rho {\pi_j}_\theta$ for all~$j\in[m]$ and all~$\theta \in \Theta$.
\end{enumerate}
The second inequality in
Eq.~\eqref{eq:CFIattainChain}
is saturated if and only if~$\rho$ is
a projector~$\ketbra{\psi}$
onto the eigenvector~$\ket{\psi}$
of~$Q_\theta$ corresponding to its
largest eigenvalue, i.e.,
\begin{enumerate}
  \setcounter{enumi}{2}
  \item $Q_\theta \rho = \lambda^{\mathrm{max}} \rho$ where $\lambda^{\mathrm{max}} = \max \mathrm{eig}  \, \left [  Q_\theta \right ]$
\end{enumerate}
Condition~(2), for saturating
the Cauchy-Schwarz inequality,
is equivalent to
\begin{equation}
\label{eq:detestattain}
    \frac{ \rho \, {\pi_j}_\theta}{\Tr( \rho \, {\pi_j}_\theta)} = \frac{ \rho \, {L_j}_\theta \, {\pi_j}_\theta }{\Tr(\rho \, {L_j}_\theta \,  {\pi_j}_\theta )} \, ,
\end{equation}
and is satisfied
if and only if~$\rho {L_j}_\theta \propto \rho$
%or~$\rho {L_j}_\theta/\Tr(\rho {L_j}_\theta) = \rho$,
on the support of~${\pi_j}_\theta$
for each~$j\in[m]$.
This is only possible
if~$\rho$ is a projector~$\ketbra{\psi}$ onto
a simultaneous eigenstate~$\ket{\psi}$
of every~${L_j}_\theta$ for~$j\in [m]$
(or a combination of projectors
onto multiple degenerate
simultaneous eigenstates).
Criterion~(2)
is thus equivalent to all the~${L_j}_\theta$
sharing at least one common eigenvector,
whereas criterion~(3) requires this eigenvector
to also be the largest-eigenvalue
eigenvector of~$Q_\theta$.
\end{proof}

In the case where
the largest eigenvalue of~$Q_\theta$
is degenerate among multiple eigenvectors,
the definition for~$\QFIsp$
still holds,
but the DQFI-optimal probe state
can be some superposition of
these degenerate eigenvectors.
The optimal superposition state
is prescribed by criterion~(2) above
to be
a simultaneous eigenstate
of every~${L_j}_\theta$, if one exists.
If no simultaneous eigenstate exists,
DQFI~$\QFIsp$ only
upper-bounds~${\CFI_\theta}_\mathrm{max}$
but fails to reveal the optimal probe state.
Below, we prove Theorem~\ref{th:attaincrit}
in the forward direction, i.e.,
that probe states~$\rho= \ketbra{\psi}$ satisfying criteria~(1)
through~(3) attain a CFI~$\CFI_\theta[\rho, \Pi_\theta] = \QFIsp[\Pi_\theta]$.

\begin{proof}[Proof of Theorem~\ref{th:attaincrit} (Forward Direction)]
Let us assume pure probe state~$\rho = \ketbra{\psi}$
satisfies the attainability criteria
enumerated 1 to 3.
As~$\ket{\psi}$ is an SLD-simultaneous eigenstate,
and~${L_j}_\theta$ is Hermitian,
we can write~${L_j}_\theta \ket{\psi}
\, {=} \, {\lambda_j}_\theta \ket{\psi}$
for real eigenvalues~${\lambda_j}_\theta$.
Then,~$\Tr({\pi_j}_\theta \rho {L_j}_\theta) =
{\lambda_j}_\theta^* \, \Tr({\pi_j}_\theta \rho) =
{\lambda_j}_\theta \, {p_j}_\theta$.
The CFI from Eq.~\eqref{eq:proofEq1}
then equals
\begin{equation}
    \CFI_\theta[\rho, \Pi_\theta] = \sum_{j\in[m]} {\lambda_j}_\theta^2 {p_j}_\theta \, .
\end{equation}
As~$\ket{\psi}$ is also an
eigenvector of~$Q_\theta = \sum_{j\in[m]} {L_j}_\theta {\pi}_\theta {L_j}_\theta$,
and corresponds to its largest eigenvalue,~$\Vert Q_\theta \Vert_{\mathrm{sp}}^2$,
\begin{equation}
\begin{split}
    \QFIsp[\Pi_\theta] &= \Tr(Q_\theta \rho)
    = \sum_{j\in[m]} \bra{\psi} {{L_j}_\theta} {{\pi_j}_\theta} {{L_j}_\theta} \ket{\psi} \\
    &= \sum_{j\in[m]} {\lambda_j^2}_\theta \bra{\psi} {{\pi_j}_\theta} \ket{\psi}
    = \sum_{j\in[m]} {\lambda_j}_\theta^2 {{p_j}_\theta} \, ,
\end{split}
\end{equation}
so that~$\CFI_\theta[\rho, \Pi_\theta]=\QFIsp[\Pi_\theta]$, thus proving the claim.
\end{proof}

\subsection{Tight DQFI for Single-Parameter Estimation}
The spectral DQFI (Def.~\ref{def:QFI2}) is tight for phase-insensitive
measurements but not for the general phase-sensitive case.
However, a minor modification to
the spectral technique,
specifically to the defining equation~\eqref{eq:SLDdetector},
can produce the tight bound,~$\QFIext$,
for general single-parameter detector models~\cite{Sarovar2006,Fujiwara2008,Matsumoto2010,Escher2011,RDD2012}.
By considering non-Hermitian SLD (nSLD) operators~${L'_j}_\theta$~\cite{ExtCon15}
that satisfy
\begin{equation}
\label{eq:nSLDMethods}
    {L'_j}_\theta^\dagger \, {\pi_j}_\theta
    + {\pi_j}_\theta \, {L'_j}_\theta = 2 \, \partial_\theta {\pi_j}_\theta
\end{equation}
instead of Eq.~\eqref{eq:SLDdetector},
we show in Supp. Mat.~\ref{sec:tightboundsdp}
that the tight bound~$\QFIext$ may be formulated as~\cite{Fujiwara2008,Escher2011,RDD2012}
\begin{equation*}
    {\QFIext} {\coloneqq} {\min_{\{ {L'_j}_\theta \}}} \Big{\{} \big{\Vert} \sum_j {L'_j}_\theta^\dagger {\pi_j}_\theta {L'_j}_\theta \big{\Vert}_\mathrm{sp}^2 \, \Big{\vert} \,
    \mathrm{Eq.~\eqref{eq:nSLDMethods}~holds} \Big{\}}  .
\end{equation*}
This minimisation can be efficiently
solved by an SDP in terms of the POVM and
its derivative~\cite{RDD2012},
%\begin{align*}
%&\QFIext = \Big \{	\min_{t, \{S_j\}} t \, \, \Big\vert   \\
%&{\begin{pmatrix} t \mathds{1}_d &  \partial_\theta  {\pi_1}_\theta - S_1   & \dots %&   \partial_\theta  {\pi_m}_\theta - S_m  \\
%	 \partial_\theta  {\pi_1}_\theta + S_1 & {\pi_1}_\theta & 0 & 0 \\
%	\vdots  &  0 & \ddots  & 0 \\
%	 \partial_\theta  {\pi_m}_\theta + S_m &  0 & 0  & {\pi_m}_\theta \\
%	 \end{pmatrix}} \succcurlyeq 0 \Big \} \, ,
%\end{align*}
%where~$S_j$ denote skew-Hermitian matrices.
which we present in Supp. Mat.~\ref{sec:tightboundsdp}.
Lemma~\ref{lemma:extdqfivalid}
therein
proves the validity of the upper bound~$\CFImax \leq \QFIext$,
directly extending main-text Theorem~\ref{th:upperbound}
and its proof in Eq.~\eqref{eq:proofEq1} of Methods.
Then, Theorem~\ref{th:extdqfitight} proves
the tightness of the bound by
showing~$\CFImax = \QFIext$, extending
main-text Theorem~\ref{th:diagonaltheorem}
to general single-parameter models.
While a closed-form analytical solution
remains unknown,
the proof reveals the optimal nSLD
operators to satisfy all three attainability
criteria listed in the previous subsection.

Notably, the tight bound is based
on a channel estimation bound
that is not always tight for the original channel
but is tight for the extended channel~\cite{Fujiwara2008}.
Here, the extended channel refers
to augmenting the unknown measurement
with additional ancilla dimensions,~$\Pi_\theta \otimes \mathds{1}_A$,
that can be jointly probed using ancilla-entangled states
(see Fig.~\ref{fig:extendedDetector}
in Supp.~Mat.~\ref{sec:tightboundsdp}).
Unlike general quantum channels~\cite{Escher2011,RDD2012}, however,
for single-parameter detector estimation,
ancilla-entangled probes do not offer any advantage
over separable probes
(Theorem~\ref{th:extdqfitight}
and Corollary~\ref{corr:sepopt}
in Supp. Mat.~\ref{sec:tightboundsdp}).

\subsection{Experimental Parameters}

The quantum circuit
implementing the estimation experiment
is shown in Fig.~\ref{fig:DephasedDetIBM}(c).
The first qubit, representing the probe,
is prepared in a pure state
at polar angle~$\theta_\mathrm{in}$
and azimuthal angle~$\phi = 0$
(we choose~$\phi = 0$
for the ideal measurement~$\Pi_\mathrm{ideal}$
so only the~$\phi=0$ cross-section of
the Bloch surface is relevant).
This probe undergoes dephasing noise
of strength~$p$, implemented via
interaction with an ancilla qubit,
before the final projection measurement along
polar angle~$\theta$ and azimuthal angle~$\phi=0$.
By scanning probe angle~$\theta_\mathrm{in}$
over the interval~$[0.2, 1]$
while fixing detector angle~$\theta = \pi/8$,
we estimate parameter~$p$.
By repeating the process~$N=10^5$
times for each~$\theta_\mathrm{in}$,
we report the empirical MSE of estimates
(green dots) in Fig.~\ref{fig:DephasedDetIBM}(d)
(with error bars in grey).

The empirical MSEs agree well with
theoretically-expected MSEs (black curve)
(for true value of~$p$ chosen to be~0.2),
and both are lower-bounded by~$1/\QFI_{\Vert, p}=1.093$
from Eq.~\eqref{eq:DepDetQFIs},
in accordance with the spectral QCRB
(blue dashed line).
Theoretically, the optimal
probe state angle
is~$\arctan(\tan\theta/(1-2p)) = 0.6042$ rad
(grey dashed line).
From the simulation data points,
which form a grid of step size~$0.006$ rad
within~$[0.5, 0.7]$ rad,
the optimal inferred angle is~$0.6162$ rad,
with~90\% confidence interval~$[0.6121, 0.6202]$ rad
(green shaded region).
From experimental data points,
which form a grid of step size~$0.014$ rad
within~$[0.5, 0.7]$ rad,
the optimal inferred angle is~$0.6121$ rad,
with~90\% confidence interval~$[0.5394, 0.6848]$ rad
(blue shaded region).
%TC:endignore

\vspace{4em}
\newpage

%\bibliographystyle{naturemag}
%\bibliography{sample}

%TC:ignore
\section{Acknowledgements}
\noindent
We acknowledge the use of IBM Quantum services for this work.
The views expressed are those of the authors
and do not reflect the official policy or position of IBM or the IBM Quantum team.
We thank Daoyi Dong and Shuixin Xiao
for valuable discussions.
This research is funded by the Australian
Research Council Centre of Excellence CE170100012.
This research was also supported by A\!*STAR C230917010,
Emerging Technology and A\!*STAR C230917004, Quantum Sensing.
This project is supported by the National Research Foundation, Singapore
through the National Quantum Office,
hosted in A\!*STAR, under its
Centre for Quantum Technologies Funding Initiative (S24Q2d0009).
This project is also supported by
the National Research Foundation of Korea
(RS-2024-00509800).
%TC:endignore

\section{Author Contributions}
\noindent
A.D., J.Z. and S.M.A. conceptualised the project.
A.D., S.K.Y., L.O.C.,  J.Z. and S.M.A. developed the theoretical analysis.
A.D. performed the experiment.
A.D. wrote the original draft and
all authors
(A.D., S.K.Y., L.O.C., O.E., A.W., Y.-S.K., P.K.L.,
J.Z. and S.M.A.)
contributed
to reviewing and editing the manuscript.
The project was supervised by J.Z., S.M.A. and L.O.C.

\noindent\rule{\columnwidth}{0.4pt}

%\clearpage
%\newpage

\vspace{2em}

\appendix

%\include{DetectorTomographySupplementalFinal}

%\import{}{DetectorTomographySupplementalArxiv}

%\includestandalone{DetectorTomographySupplementalArxiv}

\let\oldsec=\section
\let\oldsubsec=\subsection
\let\oldsubsubsec=\subsubsection
%\let\definition=\subsubsection

% Left-align and bold sections, unnumbered
\makeatletter
\renewcommand\section[1]{%
  \par
  \vspace{1.5ex}%
  {\raggedright\normalfont\large\bfseries #1\par}%
  \vspace{1ex}%
}

% Bold, left-aligned, unnumbered subsections
\renewcommand\subsection[1]{%
  \par
  \vspace{1.5ex}%
  {\raggedright\normalfont\bfseries #1\par}%
  \vspace{1ex}%
}

% Inline, bold, left-aligned, unnumbered subsubsections
\renewcommand\subsubsection[1]{%
  \noindent
  {\bfseries #1.}%
}
\makeatother

\title{Supplemental Material: Precision Bounds for Characterising Quantum Measurements}

\onecolumngrid

\noindent\rule{\textwidth}{0.4pt}

\let\section=\oldsec
\let\subsection=\oldsubsec
\let\subsubsection=\oldsubsubsec

%\tableofcontents

%\renewcommand{\thesection}{Supplemental Note \arabic{section}}
%\renewcommand{\thesubsection}{\arabic{section}.\arabic{subsection}}
%\renewcommand{\thesubsubsection}{\arabic{section}.\arabic{subsection}.\arabic{subsubsection}}

%\titleformat{\section}
%  {\normalfont\large\bfseries}{\thesection.}{1em}{}
%\titleformat{\subsection}
%  {\normalfont\normalsize\bfseries}{\thesubsection.}{1em}{}
%\titleformat{\subsubsection}
%  {\normalfont\normalsize\bfseries}{\thesubsubsection.}{1em}{}

% Redefine the section counter to use numbers instead of letters
\renewcommand{\thesection}{\Roman{section}}

% Redefine the subsection counter to use numbers instead of letters
\renewcommand{\thesubsection}{\Alph{subsection}}

% Redefine the equation numbering
\renewcommand{\theequation}{SM.\thesection.\arabic{equation}}

% Reset equation counter at the start of each section
\counterwithin*{equation}{section}

% Store the old \section command
\let\oldSection\section

% Redefine the section command to include "Supplemental Note" in the title
\renewcommand{\section}[1]{%
  \refstepcounter{section}%
  %\clearpage% Start each section on a new page
  \setcounter{equation}{0}% Reset equation counter for each section
  \oldSection{Supplemental Material \thesection. #1}%
}

% Store the old \subsection command
\let\oldSubSection\subsection

% Redefine the subsection command to include "Supplemental Note" in the title
\renewcommand{\subsection}[1]{%
  \refstepcounter{subsection}%
  %\clearpage% Start each section on a new page
  \oldSubSection{\thesubsection. #1}%
}

% Redefine \appendixname to be empty (removes "Appendix" from section titles)
\renewcommand{\appendixname}{}

% Reset equation counter at the start of each section
\counterwithin*{equation}{section}

\noindent
\begin{center}
{\large \textbf{Supplemental Material: Precision Bounds for Characterising Quantum Measurements}}
\end{center}
\noindent\rule{\textwidth}{0.4pt}

\section{Maximum CFI and optimal probe states for diagonal measurements}
\label{supp:CFIoptdiag}

In this section,
we prove that for
measurement operators~$\Pi_\theta \equiv \{{\pi_j}_\theta\}$
diagonal in some parameter-independent basis,
the CFI-optimal probe state is a basis state
of the parameter-independent basis.
For~$\Pi_\theta$ diagonal
in the standard basis, this reduces to main-text Theorem~\ref{th:diagonaltheorem}.
Let~$\{\ket{0}_d, \dots, \ket{d-1}_d\}$
be the parameter-independent basis in
which POVM elements~$\{{\pi_j}_\theta\}$ are all diagonal.
Define the unitary matrix transforming
the standard basis~$\{\ket{0}, \dots, \ket{d-1}\}$
into this parameter-independent basis
to be~$U$.
The transformed diagonal POVM~$\Pi'_\theta$
with elements~${\pi'_j}_\theta \coloneqq U \, {\pi_j}_\theta \, U^\dagger$
can now be used to infer the optimal probe state.

We first show that transforming
both probe states and measurement operators
to the basis~$\{\ket{k}_d\}$
leaves the CFI invariant, i.e., $\CFI_\theta[\rho, \, \Pi_\theta] = \CFI_\theta[ U \rho U^\dagger ,  \Pi'_\theta]$.
Notably, this invariance does not hold for
transforming to a parameter-dependent basis
because the parameter derivative of the probabilities
are not invariant under such a transformation.
The probability of detecting outcome~$j$
in the standard basis is~$\Tr(\rho \, {\pi_j}_\theta)$,
whereas in the diagonal basis,
it is~$\Tr( U \rho U^\dagger U {\pi_j}_\theta U^\dagger) = \Tr(\rho {\pi_j}_\theta)$. Thus the probabilities
remain invariant under any basis-change transformation.
On the other hand,
the parameter derivative of
the probabilities,~$\partial_\theta p_\theta(j) = \Tr(\rho \, \partial_\theta \, {\pi_j}_\theta)$,
transforms to~$\Tr \left (U\rho U^\dagger \partial_\theta \left ( U \, {\pi_j}_\theta U^\dagger\right ) \right )$,
which is equal to~$\Tr(\rho \, \partial_\theta \, {\pi_j}_\theta)$
if~$U$ is parameter-independent.
For parameter-independent~$U$, thus,
the parameter-derivative of the probabilities
and the probabilities themselves are preserved,
meaning
the CFI remains unchanged.

The invariance of the CFI
when transforming to parameter-independent bases
means that the maximum CFI~${\CFI_\theta}_\mathrm{max}[\Pi'_\theta]$
is equal to the maximum CFI~${\CFI_\theta}_\mathrm{max}[\Pi_\theta]$.
Accordingly, we can perform the maximisation
as per main-text Eq.~\eqref{eq:maxCFIdef}
in the diagonal basis~$\{\ket{k}_d\}$,
and transform back to the standard basis to find the optimal probe states.
Suppose, in the diagonal basis, the CFI-maximising probe state
is~$\rho_\mathrm{opt}$.
Let us denote by~$\mathrm{Diag}[\{A_{jk}\}_{j,k\in[d]}] \coloneqq \{\delta_{jk} A_{jk} \}_{j,k\in[d]}$ the operation of dropping
the off-diagonal terms of a matrix.
It is straightforward to see then,
that
\begin{equation}
    \Tr(\rho_\mathrm{opt} U {\pi_j}_\theta U^\dagger) = \Tr(\mathrm{Diag}[\rho_\mathrm{opt}] U {\pi_j}_\theta U^\dagger ) \, ,
\end{equation}
because~$U {\pi_j}_\theta U^\dagger$ is diagonal. The same equality
holds for the parameter-derivatives of the probabilities,
meaning~$\mathrm{Diag}[\rho_\mathrm{opt}]$ is also
a CFI-maximising state in the diagonal basis.
This establishes that in the diagonal basis,
a CFI-maximising state can always be found
amongst diagonal probe states.

Finally, we show that if
the diagonal basis states~$\{\ket{k}_d\}_{k=0}^{d-1}$
have distinct CFIs,
the CFI-maximising diagonal probe state
is a basis state.
To see this,
first define a function~$F(\vec{x}) = F(x_1, x_2) = \frac{x_2^2}{x_1}$
for~$x_1>0$.
This function of two variables is convex in its vector argument,
i.e., for~$0\leq c_j \leq 1$ and~$\sum_{j\in[d]} c_j = 1$,
\begin{equation}
    F \left (\sum_j c_j \vec{x}_j \right ) = F \left (\sum_j c_j {x_j}_1, \sum_j c_j {x_j}_2 \right ) \leq \sum_j c_j F\left ( {x_j}_1, {x_j}_2\right ) = \sum_j c_j F \left (\vec{x}_j \right ) \, .
\end{equation}
This can be seen from the
Hessian matrix of~$F(x_1, x_2)$,
consisting of the second partial derivatives of~$F(x_1, x_2)$,
being positive semi-definite, i.e.,
\begin{equation}
    \begin{pmatrix}
 \frac{2 x_2^2}{x_1^3} & -\frac{2 x_2}{x_1^2} \\
 -\frac{2 x_2}{x_1^2} & \frac{2}{x_1} \\
\end{pmatrix} \succcurlyeq 0 \, ,
\end{equation}
because its eigenvalues are~$0$ and~$2\frac{x_1^2+x_2^2}{x_1^3}$,
which are non-negative for~$x_1>0$.
Now, the optimal diagonal probe
state~$\mathrm{Diag}[\rho_\mathrm{opt}]$
can be written as
\begin{equation}
    \rho_\mathrm{opt}^{\mathrm{(diag)}} = \begin{bmatrix}
        c_1 & 0 & \dots & 0 \\
        0 & c_2 & \dots & 0 \\
        \vdots & \vdots & \ddots & \vdots \\
        0 & 0 & \dots & c_d \, ,
    \end{bmatrix} = \sum_{j=1}^d c_j \ketbra{j-1}_d \, ,
\end{equation}
where~$c_j$ may depend on~$\theta$ but~$\sum_{j=1}^d c_j = 1$ and~$0 \leq c_j \leq 1$.
The diagonal measurement operators
for~$\Pi'_\theta$ can be expanded
in the~$\{\ket{k}_d\}$ basis
as~${\pi'_j}_\theta = \sum_{k=1}^d {\pi'_{jk}}_\theta \ketbra{k-1}_d$.
The CFI of this state is then
\begin{equation}
    \CFI_\theta[\rho_\mathrm{opt}^{(\mathrm{diag})}, \Pi'_\theta] = \sum_{j\in[m]} \frac{ \Tr[\rho_\mathrm{opt}^{\mathrm{(diag)}} \partial_\theta  {\pi'_j}_\theta ]^2}{\Tr[\rho_\mathrm{opt}^{\mathrm{(diag)}} {\pi'_j}_\theta ]}
    =    \sum_{j\in[m]} \frac{(\sum_{k\in[d]} c_k \partial_\theta {\pi'_{jk}}_\theta)^2}{\sum_{k\in[d]} c_k  {\pi'_{jk}}_\theta} \, .
\end{equation}
Each term of the sum over~$j$
in this last expression
is of the form~$F \left (\sum_k c_k {x_k}_1, \sum_k c_k {x_k}_2 \right )$,
upon identifying~${x_k}_1 \leftrightarrow {\pi'_{jk}}_\theta$
and~${x_k}_2 \leftrightarrow \partial_\theta {\pi'_{jk}}_\theta$,
and is thus convex in~$\{c_1, \dots, c_d\}$.
As a result, the sum itself
and the CFI
are convex in~$\{c_1, \dots, c_d\}$,
so that
\begin{equation}
\label{eq:defCFIconvex}
     \CFI_\theta[\rho_\mathrm{opt}^{(\mathrm{diag})}, \Pi'_\theta] \leq \sum_{k\in[d]} c_k \sum_{j\in[m]} \frac{( \partial_\theta {\pi'_{jk}}_\theta)^2}{{\pi'_{jk}}_\theta}  = \sum_{k\in[d]} c_k \CFI_\theta[ \, \ketbra{k-1}_d , \Pi'_\theta \, ] \leq \max_{k\in[d]}  \CFI_\theta[ \, \ketbra{k-1}_d , \Pi'_\theta \, ] \, .
\end{equation}
As optimal probe states can always be found
within the family of diagonal states,
Eq.~\eqref{eq:defCFIconvex} implies that
\begin{equation}
    \max_{\rho_\mathrm{in}} \CFI[\rho_\mathrm{in}, \Pi_\theta] = \max_{k\in[d]} \CFI_\theta[ \, \ketbra{k-1}_d , \Pi'_\theta \, ]
\end{equation}
and that a CFI-maximising state
is~$\ket{k^*-1}_d$
where~$k^* = \arg \max_{k\in[d]} \CFI_\theta[ \, \ketbra{k-1}_d , \Pi'_\theta\, ]$,
i.e., the diagonal basis state with the largest CFI.

\section{Tight bound for detector estimation}
\label{sec:tightboundsdp}

In this section,
we present a tight bound for single-parameter detector estimation
that can be solved by an SDP.
This bound is based on an extension technique---namely
channel extension~\cite{Fujiwara2008,Sarovar2006}---where instead of considering
a quantum channel~$\mathcal{N}_\theta$, one considers the
extended channel~$\mathcal{N}_\theta \otimes \mathds{1}$~\cite{Fujiwara2008}.
(Here~$\mathds{1}$ represents an identity channel in additional dimensions
that can be probed using ancillary states.)
It is well-known that information content of the extended channel
(for any choice of ancilla dimension) upper-bounds
the information content of the original channel in general~\cite{Fujiwara2008,Sarovar2006,Escher2011,RDD2012}.
While this upper bound is not typically tight for quantum-to-quantum channels~\cite{Fujiwara2008,Sarovar2006,Escher2011}
as considered in~\cite{RDD2012}, for quantum detectors,
which correspond to quantum-to-classical channels,
we can expect this upper bound to be tight
in the single-parameter case~\cite{Sarovar2006,RDD2012}.

\subsection{Quantum-Classical Channel Representation}

Given a parametrised POVM~$\Pi_\theta = \{{\pi_j}_\theta\}_{j\in [m]}$ with~$m$ outcomes
in~$d$ dimensions,
we can obtain its channel representation through Kraus operators~\cite{Wilde13} defined as
\begin{equation}
\label{eq:KrausSet1}
	{K_{j,k}}_\theta = \ket{j} \bra{k} \sqrt{{\pi_j}_\theta} \, , \quad \quad j \in [m] \, , k \in [d] \, .
\end{equation}
Here~$\{ \ket{k} \}_{k\in[d]}$ is a basis for the~$d$-dimensional
state space and~$\{\ket{j}\}_{j\in[m]}$ represents the classical outcome space
of the measurement.
The set~$\mathcal{K} \coloneqq \{{K_{j,k}}_\theta\}$ therefore comprises~$m d$ Kraus elements,
each of which is an~$m \times d$ complex matrix.
It is easy to verify that the elements of~$\mathcal{K}$ satisfy the channel condition,
\begin{equation}
\label{eq:channelcondition}
	\sum_{j,k} {K_{j,k}}_\theta^\dagger {K_{j,k}}_\theta = \mathds{1}_d \, ,
\end{equation}
and, acting on an input state~$\rho_\mathrm{in}\in \mathcal{D}(\mathcal{H}_d)$,
produces the classical output state
\begin{equation}
\label{eq:channeloutput}
	\rho_\mathrm{out} \coloneqq \sum_{j,k}  {K_{j,k}}_\theta \rho_\mathrm{in}  {K_{j,k}}_\theta^\dagger =  \mathrm{Diag}[ p_1, \dots, p_m] \, ,
\end{equation}
where~$p_j$ is the probability of the~$j^\text{th}$ outcome.
However, the Kraus set~$\mathcal{K}$
defined in Eq.~\eqref{eq:KrausSet1} is not the only one
to satisfy Eqs.~\eqref{eq:channelcondition} and~\eqref{eq:channeloutput}~\cite{Sarovar2006,RDD2012}.
There are an infinite number of equivalent Kraus operators,
related unitarily to the one in Eq.~\eqref{eq:KrausSet1},
all of which describe the same physical transformation.

\subsection{Channel Bounds for Kraus Operators}
\label{subsec:ChannelBoundKrausOps}

Nonetheless,
with the above channel representation
for the measurement~$\Pi_\theta$,
the channel estimation technique~\cite{Escher2011,RDD2012}
may be adapted to
detector estimation. In particular,
an upper bound~\cite{Fujiwara2008} to the maximum CFI
of measurement outcomes
is given by
\begin{equation}
	\label{eq:extboundH0}
	\begin{split}
	{\mathcal{F}_\theta}_\mathrm{max} = \max_{\rho_{\mathrm{in}}} \mathcal{I}_\theta[ \rho_\mathrm{out}]
	&\leq 4 \left\Vert \sum_{j,k} \partial_\theta {K_{j,k}}_\theta^\dagger \partial_\theta {K_{j,k}}_\theta \right\Vert_{\mathrm{sp}}^2  \\
	& = 4  \left\Vert \sum_{j} (\partial_\theta \sqrt{{\pi_j}_\theta})^2 \right \Vert_{\mathrm{sp}}^2 \, .
	\end{split}
\end{equation}
Although this bound is not generally tight,
and is inferior to the spectral DQFI~$\QFIsp$
defined in the main text (Definition~\ref{def:QFI2}),
if we optimise over all equivalent Kraus representations,
we should expect a tight bound, in principle~\cite{Sarovar2006, Escher2011,RDD2012}.
More precisely,
we consider Kraus sets~$\mathcal{K}'$
obtained by unitarily transforming~$\mathcal{K}$,
\begin{equation}
	\mathcal{K}' \coloneqq \{ {K'_{j,k}}_\theta \} \, , \quad {K'_{j,k}}_\theta = \sum_{j', \, k'} {U_{jk,j'k'}}_\theta \,  {K_{j',k'}}_\theta \, ,
\end{equation}
where~$U_\theta$ is a parameter-dependent~$md \times md$ unitary matrix.
Then, the tight bound for
single-parameter detector estimation can be
written as~\cite{Escher2011}
\begin{equation}
\label{eq:CextDef1}
\mathcal{J}_\mathrm{ext} \coloneqq 4 \min_{U_\theta}  \left\Vert \sum_{j,k} \partial_\theta {K'_{j,k}}_\theta^\dagger \partial_\theta {K'_{j,k}}_\theta \right\Vert_{\mathrm{sp}}^2 \, .
\end{equation}
For local estimation in the separable setting,
as considered in our manuscript,
we can assume~$U_\theta$ to be at most linear in~$\theta$,
and therefore of the form~$U_\theta \sim e^{i (\theta - \theta^*) H}$,
for some Hermitian generator~$H$ and true value~$\theta^*$~\cite{RDD2012}.
The modified Kraus derivatives are then
\begin{equation}
	\partial_\theta {K'_{j,k}}_\theta = \partial_\theta {K_{j,k}}_\theta - i \sum_{j', k'} H_{jk, j'k'} {K_{j',k'}}_\theta \, ,
\end{equation}
and the minimisation over~$U_\theta$
in Eq.~\eqref{eq:CextDef1} reduces
to a minimisation over the Hermitian generator~$H$,
which can be formulated as a semi-definite program~\cite{RDD2012,Albarelli2022}.

\subsection{Detector Extension Bound}
In fact, for detectors, the~$md \times md$ matrix~$H$
is block-diagonal with~$m$ number of~$d \times d$ blocks,
representing the fact that we only need to consider unitary mixtures
of Kraus operators that correspond to the same measurement outcome,
i.e.,~$H_{jk, j'k'} = 0$ whenever~$j\neq j'$.
This is because we are effectively choosing the optimal (parameter-dependent)
measurement basis~$\{\ket{k'}_\theta\}$
instead of~$\{\ket{k}\}$ in Eq.~\eqref{eq:KrausSet1}, for each
measurement outcome~$j\in[m]$.
Accordingly, denoting the~$j^\text{th}$
diagonal Hermitian block of~$H$ as~$h_j$,
we may re-evaluate the bound in
Eq.~\eqref{eq:extboundH0}
using the optimal derivatives~${\mathcal{D}_j}_\theta$ instead of~$\partial_\theta \sqrt{ {\pi_j}_\theta}$,
\begin{equation}
\label{eq:tightsdpprimalform}
	{\mathcal{F}_\theta}_\mathrm{max} \leq 4 \min_{h_j}  \bigg\Vert \sum_{j} {\mathcal{D}_j^\dagger}_\theta {\mathcal{D}_j}_\theta \bigg \Vert_{\mathrm{sp}}^2  ,
	\quad \mathrm{where} \, \, {\mathcal{D}_j}_\theta \coloneqq \partial_\theta \sqrt{{\pi_j}_\theta} - i h_j \sqrt{{\pi_j}_\theta} \, .
\end{equation}

This expression can be further simplified to remove the square-root derivatives
to get
\begin{equation}
\label{eq:tightsdpsecondform}
	\mathcal{J}_{\mathrm{Ext},\theta} \coloneqq \min_{\{S_j\}} \bigg \Vert  \sum_{j} (\partial_\theta {\pi_j}_\theta - S_j) \,  {\pi_j}_\theta^{-1} \, (\partial_\theta  {\pi_j}_\theta + S_j) \bigg \Vert_{\mathrm{sp}}^2 \, ,
\end{equation}
where~$S_j = 2 \sqrt{{\pi_j}_\theta} \partial_\theta \sqrt{{\pi_j}_\theta}  - \partial_\theta {\pi_j}_\theta  - 2 i  \sqrt{{\pi_j}_\theta} h_j \sqrt{{\pi_j}_\theta}$ are skew-Hermitian matrices~$(S_j^\dagger = - S_j)$ for~$j\in[m]$.
We call the bound~${\CFI_\theta}_\mathrm{max} \leq \mathcal{J}_{\mathrm{Ext},\theta}$ the detector extension bound, and the quantity~$\mathcal{J}_{\mathrm{ext}, \theta}$ the extended DQFI.
The minimisation in
Eq.~\eqref{eq:tightsdpsecondform} can be readily formulated
as the following SDP,
\begin{equation}
\mathcal{J}_{\mathrm{Ext},\theta} = \Big \{	\min_{t,\, \{S_j\}} t \, \, \Big\vert   \begin{pmatrix} t \mathds{1}_d &  (\partial_\theta  {\pi_1}_\theta - S_1) {\pi_1}_\theta^{-1/2}  & \dots &   (\partial_\theta  {\pi_m}_\theta - S_m) {\pi_m}_\theta^{-1/2} \\
	{\pi_1}_\theta^{-1/2}  (\partial_\theta  {\pi_1}_\theta + S_1) & \mathds{1}_d & 0 & 0 \\
	\vdots  &  0 & \ddots  & 0 \\
	{\pi_m}_\theta^{-1/2}  (\partial_\theta  {\pi_m}_\theta + S_m) &  0 & 0  & \mathds{1}_d \\
	 \end{pmatrix} \succcurlyeq 0 \Big \} \, ,
\end{equation}
which simplifies to
\begin{equation}
\mathcal{J}_{\mathrm{Ext},\theta} = \Big \{	\min_{t,\, \{S_j\}} t \, \, \Big\vert   \begin{pmatrix} t \mathds{1}_d &  \partial_\theta  {\pi_1}_\theta - S_1   & \dots &   \partial_\theta  {\pi_m}_\theta - S_m  \\
	 \partial_\theta  {\pi_1}_\theta + S_1 & {\pi_1}_\theta & 0 & 0 \\
	\vdots  &  0 & \ddots  & 0 \\
	 \partial_\theta  {\pi_m}_\theta + S_m &  0 & 0  & {\pi_m}_\theta \\
	 \end{pmatrix} \succcurlyeq 0 \Big \} \, .
\end{equation}
Note that the choice~$S_j = 0$ reverts to the
unoptimised bound in Eq.~\eqref{eq:extboundH0}. On the other hand,
the optimal skew Hermitian matrices~$S_j$,
which are related to the optimal~$h_j$
in Eq.~\eqref{eq:tightsdpprimalform}
or the optimal~$U_\theta$ in Eq.~\eqref{eq:CextDef1},
are not known in terms of~$\Pi_\theta$ and~$\partial_\theta \Pi_\theta$.

\subsection{Connection  and Comparison with SLD Approach}
Remarkably,
the SLD operators~${L_j}_\theta$ defined in main-text Eq.~\eqref{eq:SLDdetector} yield
a good ansatz for the optimal~$S_j$.
Consider the choice~$S_j = {\pi_j}_\theta {L_j}_\theta - \partial_\theta {\pi_j}_\theta$,
where skew-Hermicity, i.e.~$S_j + S_j^\dagger = 0$, follows
from the SLD definition (Eq.~\eqref{eq:SLDdetector}).
It is easy to verify that
using this ansatz in Eq.~\eqref{eq:tightsdpsecondform}
produces the spectral DQFI
(Definition~\ref{def:QFI2}~in main text),
which implies~$\mathcal{J}_{\mathrm{Ext},\theta}  \leq \mathcal{J}_{\Vert,\theta}$.
Moreover, whenever the spectral DQFI is tight, i.e.,
whenever the attainability criteria~(1)--(3)
(Methods, Section~\textbf{Attainability Criteria of the DQFI})
are satisfied,
this ansatz is exactly the optimal choice for~$S_j$.
For instance, in Example~\ref{eg:example2dephasing} of
the main-text characterising a dephased PVM, the spectral DQFI
is tight, and therefore~$S_j = {\pi_j}_\theta {L_j}_\theta - \partial_\theta {\pi_j}_\theta$
is the optimal solution for the minimisation in
Eq.~\eqref{eq:tightsdpsecondform}.
More generally,
whenever the attainability criteria are satisfied,
this ansatz serves as an analytical solution to
the minimisation in Eq.~\eqref{eq:tightsdpsecondform}.

In fact, even when the attainability criteria are not satisfied,
this ansatz produces a good approximation to the tight bound.
To demonstrate this, we numerically compare the tight bound
and the spectral DQFI for randomly-generated detector estimation problems.
Specifically, we sample 10,000 random qubit POVMs
and their derivatives~($d=2$)
corresponding to two-outcome measurements~($m=2$).
%for two ($m=2$) and three ($m=3$) outcomes, each.
Our results, shown below in Figs.~\ref{fig:singleparacomb}~\&~\ref{fig:singleparacomp},
reveal that the gap between the spectral DQFI~$ \mathcal{J}_{\Vert,\theta}$ and the
extended DQFI~$ \mathcal{J}_{\mathrm{Ext},\theta}$
can be small, averaging between~2-3\% over~10,000 random models.

\begin{figure}[htb]
    \centering
    \includegraphics[width=0.97\linewidth]{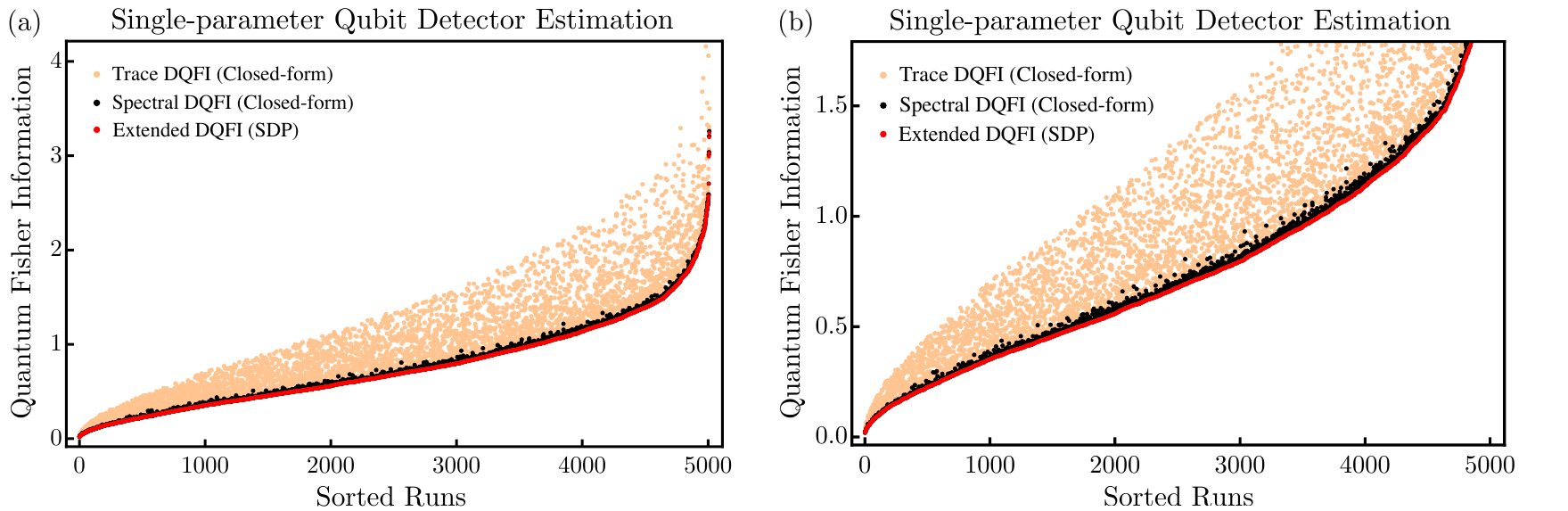}
    \caption{Comparison of detector quantum Fisher information (DQFI) measures
    for single-parameter estimation from randomly-generated
    qubit measurement models.
    (a) Across~10,000 random models, the extended DQFI
    SDP (red) provides a tighter bound for single-parameter
    estimation
    than the trace DQFI (light orange) and the spectral DQFI (black).
    (b) A zoomed-in version of (a)
    shows that the trace DQFI
    can be far from the extended DQFI but the spectral DQFI is typically
    close. In (a) and (b), the scatter points
    are sorted in the increasing order of the extended DQFI.}
    \label{fig:singleparacomb}
\end{figure}
\begin{figure}[htb]
    \centering
    \includegraphics[width=0.92\linewidth]{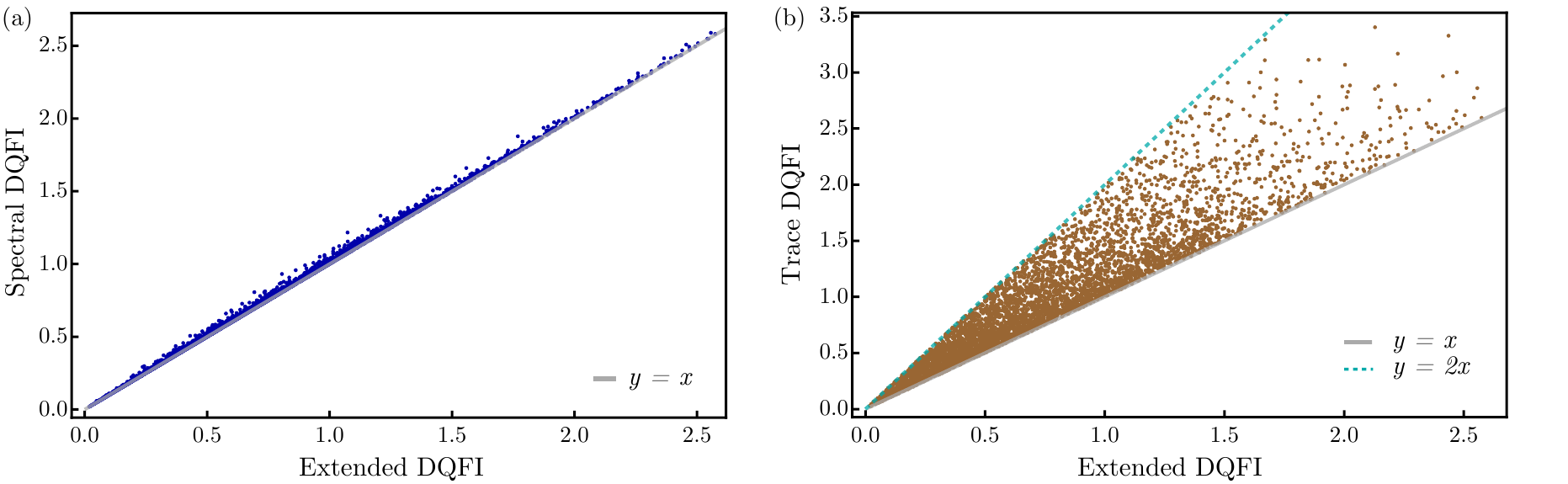}
    \caption{One-to-one comparison of the spectral DQFI
    (a) and the trace DQFI (b) with
    the extended DQFI,
    over~10,000 randomly-generated, single-parameter,
    qubit measurement models.
    (a) The spectral DQFI is not always tight
    but is generally close to the extended DQFI.
    (b) The trace DQFI and the extended DQFI
    can disagree by up to a factor of two.}
    \label{fig:singleparacomp}
\end{figure}

These results suggest a further simplification
of Eq.~\eqref{eq:tightsdpsecondform} that can help us better understand
the attainable precision limit in single-parameter detector estimation.
Consider the substitution~$S_j = {\pi_j}_\theta {L_j}_\theta - \partial_\theta {\pi_j}_\theta + S'_j$
in Eq.~\eqref{eq:tightsdpsecondform}, where~$S'_j$ are skew-Hermitian matrices.
This results in the following bound:
\begin{equation}
	\label{eq:tightsdpthirdform}
	\mathcal{J}_{\mathrm{Ext},\theta} = \min_{S'_j} \bigg \Vert  \sum_{j} ( {L_j}_\theta +  {\pi_j}_\theta^{-1} S'_j)^\dagger  \,  {\pi_j}_\theta \, ( {L_j}_\theta +  {\pi_j}_\theta^{-1} S'_j) \bigg \Vert \, .
\end{equation}
If we identify~${L'_j}_\theta \coloneqq  {L_j}_\theta +  {\pi_j}_\theta^{-1} S'_j$
as the modified SLD operators, the objective above becomes~$\Vert \sum_j {L'_j}_\theta^\dagger
{\pi_j}_\theta {L'_j}_\theta \Vert$, resembling the spectral DQFI.
This is justified because,
whereas~${L_j}_\theta$ is Hermitian and satisfies the equation~${L_j}_\theta {\pi_j}_\theta + {\pi_j}_\theta  {L_j}_\theta = 2 \partial_\theta  {\pi_j}_\theta$,
the modified SLD operators are non-Hermitian and satisfy
\begin{equation}
	\label{eq:nonHSLDops}
	{L'_j}_\theta^\dagger  {\pi_j}_\theta + {\pi_j}_\theta  {L'_j}_\theta = 2 \partial_\theta  {\pi_j}_\theta \, ,
\end{equation}
for any choice of~$S'_j$. In fact, any general solution to Eq.~\eqref{eq:nonHSLDops}
is of the form~${L'_j}_\theta = {L_j}_\theta +   {\pi_j}_\theta^{-1} S'_j$ for some skew-Hermitian~$S'_j$.
Therefore, the tight bound may be reformulated as follows:
\begin{equation}
\label{eq:DetExtQFIfinalDef}
	\mathcal{J}_{\mathrm{Ext},\theta}  = \min_{{L'_j}_\theta} \bigg\{  \big\Vert  \sum_j {L'_j}_\theta^\dagger
	{\pi_j}_\theta {L'_j}_\theta \big\Vert \, \,  \bigg\vert \, \,  {L'_j}_\theta^\dagger  {\pi_j}_\theta + {\pi_j}_\theta  {L'_j}_\theta = 2 \partial_\theta  {\pi_j}_\theta \bigg \} \, .
\end{equation}
This formulation reveals that while the SLD approach, as used to derive
the spectral DQFI, does not always result in a tight bound,
including a non-Hermitian component in the SLD operator
and optimising this component produces the tight bound.
Therefore, a minor modification to the SLD approach still yields
the tight bound for detector estimation,
underscoring the versatility of this technique.
Notably, in state estimation,
the optimal choice for the nSLD operators
that leads to the tight bound is the Hermitian SLD~\cite{ExtCon15};
for detector estimation, in contrast,
the optimal nSLD operators that produce the tight bound
are non-Hermitian.

\subsection{Proof of Validity \& Attainability Criteria -- Detector Extension Bound}
The connection between the extended bound
and the SLD approach has further far-reaching consequences.
Here, we present some key observations.

First, Theorem~\ref{th:upperbound}
and its proof
in Eq.~\eqref{eq:proofEq1} of Methods~\cite{Paris2009,Paris2004}
directly
generalises to the
extended bound,
as follows.

\begin{lemma}[Extended DQFI upper-bounds maximum CFI]
\label{lemma:extdqfivalid}
For estimating quantum detectors,
the detector extension bound~$\mathcal{J}_{\mathrm{Ext},\theta}$
upper-bounds the maximum CFI over separable probe states,~${\CFI_\theta}_\mathrm{max}$, i.e.,~${\CFI_\theta}_\mathrm{max} \leq
\mathcal{J}_{\mathrm{Ext},\theta}$.
\end{lemma}
\begin{proof}
For any
valid set of operators~$\{{L'_j}_\theta\}_{j=1}^m$ satisfying
the constraints
${L'_j}_\theta^\dagger  {\pi_j}_\theta + {\pi_j}_\theta  {L'_j}_\theta = 2 \partial_\theta  {\pi_j}_\theta$,
it holds that
\begin{equation}
\label{eq:proofTightBoundEq1}
\begin{aligned}
    &\CFI_\theta[\rho,\Pi_\theta] = \sum_{j\in[m]} \frac{\left ( \Re \left [ \Tr({\pi_j}_\theta \, \rho \, {L'_j}_\theta^\dagger ) \right ] \right )^2}{\Tr(\rho \, {\pi_j}_\theta)}
    \leq  \sum_{j\in[m]} \left \vert \frac{ \Tr({\pi_j}_\theta \, \rho \, {L'_j}_\theta^\dagger ) }{\sqrt{\Tr(\rho \, {\pi_j}_\theta)}} \right \vert^2
    = \sum_{j\in[m]} \left \vert \Tr \left ( \frac{ \sqrt{{\pi_j}_\theta} \, \sqrt{\rho}}{\sqrt{\Tr(\rho \, {\pi_j}_\theta)}}  \; \;  \sqrt{\rho} \, {L'_j}_\theta^\dagger \sqrt{{\pi_j}_\theta} \right )  \right \vert^2\\
    \leq &\sum_{j\in[m]} \Tr \left ( \frac{\sqrt{{\pi_j}_\theta} \rho \sqrt{{\pi_j}_\theta}}{\Tr(\rho \, {\pi_j}_\theta)}\right ) \Tr \left ( \sqrt{\rho} {L'_j}_\theta^\dagger {\pi_j}_\theta {L'_j}_\theta \sqrt{\rho} \right )
    = \sum_{j\in[m]} \Tr \left ({L'_j}_\theta^\dagger {\pi_j}_\theta {L'_j}_\theta \rho \right )  \, ,
\end{aligned}
\end{equation}
and therefore
\begin{equation}
\label{eq:tightboundproofeq2}
{\CFI_\theta}_\mathrm{max} = \max_\rho \CFI_\theta[\rho,\Pi_\theta] \leq \min_{{L'_j}_\theta} \Big\Vert \sum_{j\in[m]} {L'_j}_\theta^\dagger {\pi_j}_\theta {L'_j}_\theta  \Big \Vert_\mathrm{sp}^2 \equiv \QFIext \, ,
\end{equation}
which proves the upper bound.
\end{proof}

Second, the attainability criteria~(1)--(3) for
the spectral DQFI
(Methods, Section~\textbf{Attainability Criteria of the DQFI}) directly
generalises to the extended DQFI,
as can be checked by inspecting the
three inequalities in Eqs.~\eqref{eq:proofTightBoundEq1}--\eqref{eq:tightboundproofeq2}:
\begin{enumerate}
    \item $\qtrace({\pi_j}_\theta \rho {L'_j}_\theta^\dagger)$ is real for all~$j \in [m]$
    and all~$\theta \in \Theta$,
    \item $ \rho {L'_j}_\theta^\dagger {\pi_j}_\theta \propto \rho {\pi_j}_\theta$ for all~$j\in[m]$ and all~$\theta \in \Theta$.
    \item $\sum_{j\in[m]} \left ({L'_j}_\theta^\dagger {\pi_j}_\theta {L'_j}_\theta \right ) \rho = \lambda^{\mathrm{max}} \rho$ where $\lambda^{\mathrm{max}} = \max \mathrm{eig}  \, \left [  \sum_{j\in[m]} \left ({L'_j}_\theta^\dagger {\pi_j}_\theta {L'_j}_\theta \right ) \right ]$
\end{enumerate}
These criteria can always be satisfied
for any single-parameter POVM~$\Pi_\theta$,
for the optimal probe state~$\rho = \rho^\mathrm{opt}$.
This optimal probe state
is simply the largest-eigenvalue
eigenvector of~$\sum_{j\in[m]} \left ({L'_j}_\theta^\dagger {\pi_j}_\theta {L'_j}_\theta \right )$.
This shows that even if
the SLD operators~${L_j}_\theta$ do not share a common eigenvector,
we can add non-Hermitian components to them such that
the modified nSLD operators~${L'_j}_\theta$
share a common eigenvector---this
eigenvector is always the largest-eigenvalue eigenvector of~$\sum_j  {L'_j}_\theta^\dagger
{\pi_j}_\theta {L'_j}_\theta$ and corresponds to the optimal probe state.
These facts follow from the proof of tightness
of the detector extension bound~\cite{Sarovar2006,Fujiwara2008,RDD2012},
i.e.,~$\mathcal{J}_{\mathrm{Ext},\theta} = {\CFI_\theta}_\mathrm{max}$,
presented below.

\subsection{Proof of Tightness for Detector Extension Bound}
\label{subsec:tightproof}

Now we prove that
in single-parameter detector estimation,
the extended DQFI~$\QFIext$ equals the attainable DQFI,
given by the maximum CFI~$\CFImax$ over separable probe states.
This is not generally true for channel estimation,
particularly for quantum-to-quantum channels,
where~$\QFIext \geq \CFImax$.
In fact, the bound on which~$\QFIext$
is based
(introduced in Ref.~\cite{Fujiwara2008})
is generally only tight for the extended channel,
which means that attaining the bound~$\QFIext$ requires
using probe-ancilla entangled input states
to the extended system~$\Pi_\theta \otimes \mathds{1}$,
as shown in Fig.~\ref{fig:extendedDetector}.
For instance, the use of ancilla-entangled probes
can produce more information than any
separable probe in certain channel estimation problems~\cite{DD14,Huang16}.
In contrast, for single-parameter detector estimation,
or the estimation of quantum-to-classical channels,
such entangled ancilla states are not beneficial,
and the extended bound can be achieved by a single
separable pure state---precisely~$\rho^\mathrm{opt}$
defined in main-text Eq.~\eqref{eq:maxCFIdef}.

\begin{figure}[htb]
    \centering
    \includegraphics[width=0.46\linewidth]{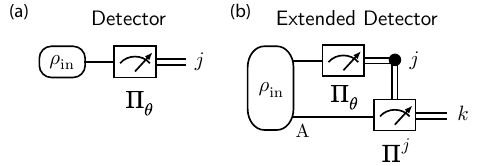}
    \caption{Comparison of detector estimation (a) and extended detector estimation (b). Estimating the extended detector requires
    an ancillary system (A) and conditional ancilla measurements ($\Pi^j$)~\cite{DAriano2004}.}
    \label{fig:extendedDetector}
\end{figure}

We now prove that the detector extension bound
is tight, meaning that~$\QFIext = \CFImax$.
First, we prove a lemma that allows us to write
the CFI~$\CFI_\theta[\rho, \Pi_\theta]$ for any pure state
in the form~$\Tr \left [\left( \sum_j {L'_j}^\dagger_\theta {\pi_j}_\theta
{L'_j}_\theta \right ) \rho \right ]$.

\begin{lemma}
\label{lemma:CFIDetExtTransform}
    For any pure state~$\rho = \ketbra{\psi}{\psi}$,
    and for the corresponding choice~${L'_j}^\mathrm{cand}_\theta
    \coloneqq  {\pi_j}_\theta^{-1} \left ( \partial_\theta {\pi_j}_\theta + S_j^\mathrm{cand} \right )$, where
    \begin{equation}
    \label{eq:someSjcand}
        S_j^\mathrm{cand} = \frac{\bra{\psi}\partial_\theta {{\pi_j}_\theta} \ket{\psi}}{\bra{\psi} {{\pi_j}_\theta} \ket{\psi}}  \left [ {\pi_j}_\theta, \, \rho \right] - \left [ \partial_\theta {\pi_j}_\theta , \, \rho \right ]  \, ,
    \end{equation}
    the quantity~$\CFI_\theta[\rho, \Pi_\theta] =
    \Tr \left [\left( \sum_j {{L'_j}^\mathrm{cand}_\theta}^\dagger {\pi_j}_\theta
    {L'_j}^\mathrm{cand}_\theta \right ) \rho \right ]$.
\end{lemma}
\begin{proof}
    As~$S_j^\mathrm{cand}$ is a commutator between Hermitian operators,
    it is skew-symmetric, i.e.,~$S_j^\mathrm{cand}
    + S_j^{\mathrm{opt} \,\dagger} = 0$,
    and therefore leads to a feasible candidate~${L'_j}^\mathrm{cand}_\theta$
    for the
    minimisation in the extended DQFI definition
    (Eq.~\eqref{eq:tightsdpsecondform}).
    First, we calculate~$S_j^\mathrm{cand}\rho$
    to find
    \begin{align*}
        S_j^\mathrm{cand} \ketbra{\psi}{\psi} &= \frac{\bra{\psi}\partial_\theta {{\pi_j}_\theta} \ket{\psi}}{\bra{\psi} {{\pi_j}_\theta} \ket{\psi}} \left ( {\pi_j}_\theta \ketbra{\psi}{\psi} - \ketbra{\psi}{\psi} {\pi_j}_\theta \ketbra{\psi}{\psi} \right )
        + \ketbra{\psi}{\psi} \partial_\theta {\pi_j}_\theta \ketbra{\psi}{\psi} - \partial_\theta {\pi_j}_\theta \ketbra{\psi}{\psi}\\
        &= \left ( \frac{\bra{\psi}\partial_\theta {{\pi_j}_\theta} \ket{\psi}}{\bra{\psi} {{\pi_j}_\theta} \ket{\psi}} {\pi_j}_\theta - \partial_\theta {\pi_j}_\theta \right ) \ketbra{\psi}{\psi} \, ,
    \end{align*}
    which leads to
    \begin{equation}
        {L'_j}^\mathrm{cand}_\theta \ketbra{\psi}{\psi} = {\pi_j}_\theta^{-1}
        \left ( \partial_\theta {\pi_j}_\theta + S_j^\mathrm{cand} \right ) \ketbra{\psi}{\psi} = \frac{\bra{\psi}\partial_\theta {{\pi_j}_\theta} \ket{\psi}}{\bra{\psi} {{\pi_j}_\theta} \ket{\psi}} \ketbra{\psi}{\psi} \propto \ketbra{\psi}{\psi} \, ,
    \end{equation}
    meaning that~$\ket{\psi}$ is a common eigenstate to
    each nSLD operator~${L'_j}^\mathrm{cand}_\theta$
    with real eigenvalues. Notice that this proves
    that the attainability criteria~(1) \&~(2) for the extended DQFI
    (presented in the previous subsection) are satisfied
    for this choice of the nSLD~$\{{L'_j}_\theta\}$.
    Now it is straightforward to compute~$\Tr \left [\left( \sum_j {{L'_j}^\mathrm{cand}_\theta}^\dagger {\pi_j}_\theta
    {L'_j}^\mathrm{cand}_\theta \right ) \rho \right ]$,
    as below,
    \begin{equation}
        \begin{aligned}
            \Tr \bigg [  \sum_j {{L'_j}^\mathrm{cand}_\theta}^\dagger {\pi_j}_\theta
            {L'_j}^\mathrm{cand}_\theta  \rho \bigg ]
            &= \sum_j  \frac{\bra{\psi}\partial_\theta {{\pi_j}_\theta} \ket{\psi}}{\bra{\psi} {{\pi_j}_\theta} \ket{\psi}} \Tr \bigg [  {{L'_j}^\mathrm{cand}_\theta}^\dagger \,  {\pi_j}_\theta \,
            \rho \, \bigg ]
            = \sum_j  \frac{\bra{\psi}\partial_\theta {{\pi_j}_\theta} \ket{\psi}}{\bra{\psi} {{\pi_j}_\theta} \ket{\psi}} \Tr \bigg [    {\pi_j}_\theta \,
            \rho \, {{L'_j}^\mathrm{cand}_\theta}^\dagger \, \bigg ] \\
            &= \sum_j \left ( \frac{\bra{\psi}\partial_\theta {{\pi_j}_\theta} \ket{\psi}}{\bra{\psi} {{\pi_j}_\theta} \ket{\psi}} \right)^2 \Tr[ \,  {\pi_j}_\theta \, \rho \,  ]
            = \sum_j \frac{\left ( \bra{\psi}\partial_\theta {{\pi_j}_\theta} \ket{\psi}\right)^2}{\bra{\psi} {{\pi_j}_\theta} \ket{\psi}} = \CFI_\theta[\rho, \Pi_\theta] \, ,
        \end{aligned}
    \end{equation}
    which is simply restating that
    the Cauchy-Schwarz inequality
    is tight for a common eigenstate,
    thereby proving the lemma.
\end{proof}

\begin{theorem}[Extended DQFI equals Maximum CFI]
\label{th:extdqfitight}
For single-parameter detector estimation,
the extended DQFI~$\QFIext$ is tight,
meaning that~$\QFIext = \CFImax$.
\end{theorem}
\begin{proof}
    From Lemma~\ref{lemma:extdqfivalid},
    we have~${\CFI_\theta}_\mathrm{max} \leq
    \mathcal{J}_{\mathrm{Ext},\theta}$.
    On the other hand,
    the detector extension bound
    can be written as
    \begin{equation}
    \begin{aligned}
        \QFIext &= \min_{ {L'_j}_\theta } \max \mathrm{eig}  \big[ \sum_j {{L'_j}_\theta}^\dagger {\pi_j}_\theta
            {L'_j}_\theta  \big ]
        = \min_{ {L'_j}_\theta } \max_\rho \,  \Tr \big[ \sum_j {{L'_j}_\theta}^\dagger {\pi_j}_\theta
            {L'_j}_\theta \, \rho \,  \big ] \\
            &= \max_\rho \min_{ {L'_j}_\theta} \Tr \big[ \sum_j {{L'_j}_\theta}^\dagger {\pi_j}_\theta
            {L'_j}_\theta \, \rho \,  \big ] \, ,
    \end{aligned}
    \end{equation}
    where in the first line
    we have used the operator norm definition,
    and in the second line we have flipped
    the order of the two optimisations.
    The operator norm definition implies
    that the optimal~$\rho$ can always
    be considered pure.
    Next, if we consider
    the minimisation~$\min_{ {L'_j}_\theta} \Tr [ \sum_j {{L'_j}_\theta}^\dagger {\pi_j}_\theta
    {L'_j}_\theta \, \rho \, ]$ for any fixed pure state~$\rho$,
    we know from Lemma~\ref{lemma:CFIDetExtTransform}
    that~${L'_j}^\mathrm{cand}_\theta$ defined using
    this state~$\rho$ provides an upper bound
    to this minimisation.
    Further, Lemma~\ref{lemma:CFIDetExtTransform}
    shows this upper bound to simply be
    the CFI~$\CFI_\theta[\rho, \Pi_\theta]$.
    All together, we have
    \begin{equation*}
        \CFImax \leq \QFIext = \max_\rho \min_{ {L'_j}_\theta} \Tr \big[ \sum_j {{L'_j}_\theta}^\dagger {\pi_j}_\theta
            {L'_j}_\theta \, \rho \,  \big ] \leq \max_\rho \Tr \big[ \sum_j {{L'_j}^\mathrm{cand}_\theta}^\dagger  \, {\pi_j}_\theta  \,
            {L'_j}^\mathrm{cand}_\theta \, \rho \,  \big ] = \max_\rho \CFI_\theta [\rho, \Pi_\theta] = \CFImax \, ,
    \end{equation*}
    which proves that~$\CFImax = \QFIext$.
\end{proof}

Theorem~\ref{th:extdqfitight}
implies that for single-parameter
detector estimation,
the extended bound based on Fujiwara's bound~\cite{Fujiwara2008}
can be attained by a single separable pure state probe,
without the need for entangled ancilla states
or conditional measurements.

\begin{corollary}
\label{corr:sepopt}
The extended DQFI~$\QFIext$ is always attainable using
a single separable pure probe state
in the original Hilbert space~$\mathcal{H}_d$.
\end{corollary}

This result aligns with expectations based on
previous literature addressing
single-parameter quantum-classical
channels~\cite{Sarovar2006,Fujiwara2008,Matsumoto2010,RDD2012}.
Although the multi-parameter setting could
lead to different conclusions regarding the utility
of entangled ancillae in detector estimation,
this remains the scope of future work.
Here, we have proven that in the single-parameter setting,
entanglement in the form of quantum correlations between probe
and ancillary systems does not offer any metrological advantage
in detector estimation.

\section{Application to noisy qubit measurements}
\label{sec:SuppNoteEgsApplication}

In this section, we provide applications
of the DQFI framework
to qubit measurements
relevant for quantum computing~\cite{Lundeen2008,LWH+21}
and quantum communication~\cite{Watanabe2008}.
A problem currently relevant
for qubit platforms is the characterisation of,
and mitigation against,
SPAM errors that combine the effects of noisy
preparation and noisy measurement~\cite{Chen2019,LWH+21}.
A noisy measurement may be modelled
as a noisy quantum channel~$\mathcal{N}_p$
followed by a noiseless or ideal projective measurement~$\Pi_\mathrm{ideal}$,
as in Fig.~\ref{fig:noisyqubitdetector}(a).
However, the problem of characterising
the noisy channel~$\mathcal{N}_p$,
by estimating unknown parameter~$p$,
is distinct from the problem of
channel parameter estimation:
the final measurement here is fixed to be~$\Pi_\mathrm{ideal}$
and cannot be tuned or optimised.
As a result, the state estimation toolbox,
which assumes that measurements can be optimised,
is not applicable here
and our DQFI fills this gap.
The most common noise processes
affecting qubit platforms are modelled
using bit-flip, dephasing,
depolarising, and amplitude damping
quantum channels~\cite{NC10}.
The bit-flip and dephasing channels
were discussed in
main-text Examples~\ref{eg:bitflipdet} \&~\ref{eg:example2dephasing}.
Here we
establish a general framework
for treating noisy projective measurements
for qubits
and then apply to specific
channels one-by-one.

\paragraph{Generic noisy-qubit projective measurement:}
Consider a noise process
modelled by a quantum channel~$\mathcal{N}$
with a Kraus representation~$\{\mathcal{K}_j\}_{j\in[K]}$,~$K$
denoting the number of Kraus operators.
A state~$\rho_\mathrm{in}$
input to this channel transforms into
the state~$\rho_\mathrm{out} \coloneqq \mathcal{N}(\rho_\mathrm{in}) = \sum_{j\in[K]} \mathcal{K}_j \, \rho_\mathrm{in} \, \mathcal{K}_j^\dagger$.
A PVM on this state
along the direction~$(\theta, \phi)$
then corresponds to computing overlaps
of~$\rho_\mathrm{out}$
with the projectors,
\begin{equation}
    \ket{\theta,\phi}_+ = \begin{bmatrix} \cos(\theta/2) \\ e^{i \phi} \sin(\theta/2) \end{bmatrix} \,  \quad \& \quad
    \ket{\theta,\phi}_- = \begin{bmatrix} \sin(\theta/2) \\ -e^{i \phi} \cos(\theta/2) \end{bmatrix} \, ,
\end{equation}
where~$\theta\in[0, \pi]$ and~$\phi\in[0, 2 \pi]$
are polar and azimuthal angles
of the projector on the Bloch sphere, respectively.
Denoting the POVM for
these projectors by~$\pi_{1,2} = \ketbra{\theta,\phi}_{\pm}$
and~$\Pi_\mathrm{ideal} \equiv \{\pi_1, \pi_2\}$,
the measurement outcome probabilities
can be alternatively obtained as
\begin{equation}
    \begin{split}
        p_1 &= \bra{\theta, \phi}_+ \mathcal{N}(\rho_\mathrm{in}) \ket{\theta,\phi}_+ =   \Tr[ \mathcal{N}(\rho_\mathrm{in}) \ketbra{\theta,\phi}_+]
        = \Tr \Bigr [ \rho_\mathrm{in}  \sum_{j\in[K]} \mathcal{K}^\dagger_j \pi_1\mathcal{K}_j \Bigr ]  \, , \\
        p_2 &= \bra{\theta, \phi}_- \mathcal{N}(\rho_\mathrm{in}) \ket{\theta,\phi}_- =   \Tr[ \mathcal{N}(\rho_\mathrm{in}) \ketbra{\theta,\phi}_-]
        = \Tr \Bigr [\rho_\mathrm{in}  \sum_{j\in[K]} \mathcal{K}^\dagger_j \pi_2 \mathcal{K}_j \Bigr ]  \, .
    \end{split}
\end{equation}
Clearly, we can construct
an effective POVM~$\Pi_\mathrm{eff} \equiv \{\pi^\mathrm{eff}_1, \pi^\mathrm{eff}_2\}$ by
reverse evolving
the PVM~$\Pi_\mathrm{ideal}$
through the
channel~$\mathcal{N}_p$,
\begin{equation}
\label{eq:effpovm}
    \pi^\mathrm{eff}_1 \coloneqq \sum_{j\in[K]} \mathcal{K}^\dagger_j \pi_{1} \mathcal{K}_j  \quad \& \quad \pi^\mathrm{eff}_2 \coloneqq \sum_{j\in[K]} \mathcal{K}^\dagger_j \pi_{2} \mathcal{K}_j \, ,
\end{equation}
allowing us to recover measurement
probabilities~$p_1$ and~$p_2$
directly from the input state~$\rho_\mathrm{in}$ via
\begin{equation}
    p_1 = \Tr \left [ \rho_\mathrm{in}  \pi^\mathrm{eff}_1 \right ] \quad \& \quad p_2 = \Tr \left [ \rho_\mathrm{in} \pi^\mathrm{eff}_2 \right ] \, .
\end{equation}
The DQFI framework
applied to POVM~$\Pi_\mathrm{eff}$
now quantifies the efficiency
in estimating the parameter~$p$ of
the noisy channel~$\mathcal{N}_p$
affecting the ideal measurement~$\Pi_\mathrm{ideal}$.

\paragraph{Attainability:}
Attainability of the DQFI
here is determined by whether
criteria (1)--(3) listed in Methods
hold for~$\Pi^\mathrm{eff}$.
For the channels considered
in the following subsection,
we find the DQFI to always be tight
except for the amplitude-damping channel,
where numerics show the discrepancy to be small.
Accordingly, for all other considered channels,
the DQFI-optimal states
attain the maximum CFI~${\CFI_p}_\mathrm{max}$
and constitute the most informative
probing strategies
when characterising
noisy qubit measurements.
Naturally, symmetries of the channel action
translate to invariance of the
DQFI under the corresponding transformation.
For example, the dephasing action is symmetric
under rotation about the~$Z$-axis,
so its DQFI depends only on
the polar angle~$\theta$ of~$\Pi_\mathrm{ideal}$
and is unaffected by its azimuthal angle~$\phi$.
The depolarising action is
symmetric in both azimuthal
and polar directions so
its DQFI depends neither
on~$\theta$ nor on~$\phi$.
Below, we present the optimal
probing strategies for characterising
dephased, depolarised and amplitude-damped
qubit detectors.

\begin{figure}[hbtp]
    \centering
    \includegraphics[width=0.9\columnwidth]{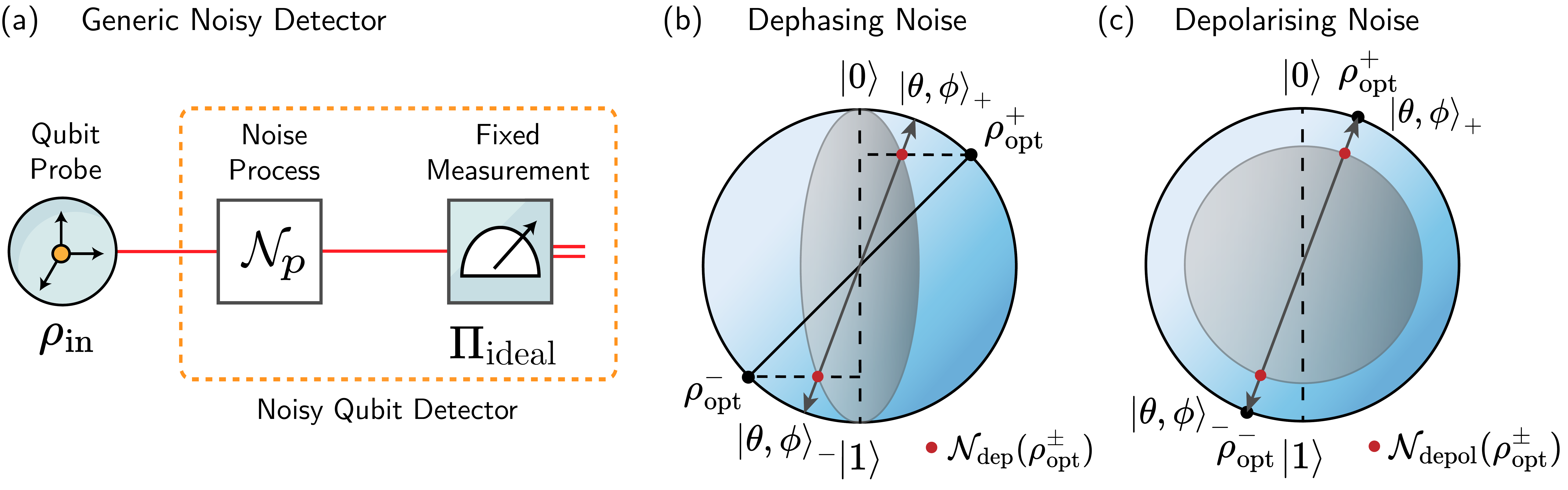}
    \caption{Estimating noise parameters
    of imperfect qubit detectors.
    (a) The noisy qubit detector
    is modelled as an ideal projection measurement
    along~$\ket{\theta,\phi}_\pm$ affected
    by noise process~$\mathcal{N}_p$
    with unknown parameter~$p$.
    (b), (c) The optimal probe states for characterising
    dephasing (b) and depolarisation (c)
    noise strengths are shown. The blue disk depicts a
    cross-section of the Bloch sphere at azimuthal angle~$\phi$, containing projectors~$\ket{\theta, \phi}_\pm$
    (arrowheads)
    and the optimal probe states~$\rho_\mathrm{opt}^{\pm}$
    (black dots).
    In (b),
    the dephasing action~$\mathcal{N}_\mathrm{dep}$
    contracts the blue disk horizontally
    (towards the vertical dashed line)
    by a factor of~$(1-2p)$
    to form the grey region.
    The red dots represent
    the dephased optimal states~$\mathcal{N}_\mathrm{dep}(\rho_\mathrm{opt}^{\pm})$;
    the optimal pure states are those for which
    the dephased state aligns with
    measurement~$\ket{\theta, \phi}_\pm$.
    In (c), the depolarising action shrinks
    the blue disk radially inward
    towards the maximally-mixed state
    at the centre to form the grey disk.
    The optimal probe states are simply
    the projectors~$\ket{\theta,\phi}_\pm$
    themselves. The red dots representing the
    corresponding depolarised states,~$\mathcal{N}_\mathrm{depol}(\rho_\mathrm{opt}^\pm)$, align
    with the measurement direction.}
    \label{fig:noisyqubitdetector}
\end{figure}

\subsection{Detector with Inherent Dephasing}

The dephasing channel~$\mathcal{N}_{\mathrm{dep}}$ erases
phase information from states,
shrinking the Bloch ball horizontally
inwards towards the~$Z$-axis
(see Fig.~\ref{fig:noisyqubitdetector}(b)).
A set of Kraus operators for~$\mathcal{N}_{\mathrm{dep}}$ are
\begin{equation}
    \mathcal{K}_1 = \sqrt{1-p} \, \mathds{1}_2 \, , \quad \mathcal{K}_2 = \sqrt{p} \,  Z \, ,
\end{equation}
where~$p$ represents the dephasing probability.
Considering dephasing to be
the leading noise mechanism
inside a detector,
we now consider the estimation
of dephasing probability~$p$
from the measurement outcomes
and find the corresponding optimal probe states.
Given input state~$\rho_\mathrm{in}$,
the channel-output state~$\mathcal{N}_{\mathrm{dep}}(\rho_\mathrm{in})$
is projectively measured
along~$\ket{\theta, \phi}_\pm$.
The effective POVM~$\Pi_\mathrm{eff}$ for this process,
computed as per
Eq.~\eqref{eq:effpovm},
yields the two effective-detector
SLD operators,~$L^\mathrm{eff}_1$ and~$L^\mathrm{eff}_2$
for~$\pi^\mathrm{eff}_1$ and~$\pi^\mathrm{eff}_2$,
respectively, as
\begin{equation}
\label{eq:SLDOpsDepDetQubit}
    L^\mathrm{eff}_1 = \frac{1}{2 p(1-p)} \begin{bmatrix}
    (1-2 p)(1-\cos\theta) & -e^{-i\phi} \sin\theta \\
    -e^{i\phi} \sin\theta & (1-2p)(1+\cos\theta)
    \end{bmatrix}
    \, \,  \& \, \,
    L^\mathrm{eff}_2 = \frac{1}{2 p(1-p)} \begin{bmatrix}
   (1-2p)(1+\cos\theta) & e^{-i\phi} \sin\theta \\
    e^{i\phi} \sin\theta & (1-2p)(1-\cos\theta)
    \end{bmatrix} \, .
\end{equation}
The matrix~$Q^\mathrm{eff}=\sum_{j\in[2]} L^\mathrm{eff}_j \pi^\mathrm{eff}_j L^\mathrm{eff}_j$
is then
\begin{equation}
    Q^\mathrm{eff} = \frac{1}{p(1-p)} \begin{bmatrix}
        \sin^2\theta & 0 \\
        0 & \sin^2\theta
    \end{bmatrix} \, ,
\end{equation}
from which, we
obtain the DQFIs
\begin{equation}
\label{eq:dephasingdetqfi}
        \mathcal{J}_{\Vert, p} = \frac{\sin^2\theta}{p(1-p)} \quad \mathrm{and} \quad
        \mathcal{J}_{\Tr, p} =  \frac{2 \sin^2\theta}{p(1-p)} \, .
\end{equation}
It is straightforward to check
that the SLD operators~$L^\mathrm{eff}_1$ and~$L^\mathrm{eff}_2$
commute, i.e.,~$[L^\mathrm{eff}_1\, ,\,  L^\mathrm{eff}_2]=0$,
so the spectral DQFI~$\QFIsp$
is attainable in this case
and stipulates the maximum CFI of measurement outcomes
over the input probe space.

However,
as the matrix~$Q^\mathrm{eff}$ has degenerate
eigenvalues,
its eigenvectors~$\ket{0}$ and~$\ket{1}$
do not correspond to the optimal state.
In fact, the optimal probe states
depend on both~$\theta$ and~$\phi$.
As discussed in the attainability
paragraph,
the optimal probe in this case
is a superposition of~$\ket{0}$ and~$\ket{1}$
that is a
common eigenstate of~$L^\mathrm{eff}_1$ and~$L^\mathrm{eff}_2$.
As~$L^\mathrm{eff}_1$ and~$L^\mathrm{eff}_2$ commute,
they share both their
eigenstates,
namely,
\begin{equation}
\label{eq:eigsL}
    \ket{\lambda_{\pm}} = \begin{bmatrix}
        (1-2p) \cos\theta \pm \sqrt{\sin^2\theta + (1-2p)^2\cos^2\theta} \\
         e^{i \phi} \sin\theta
    \end{bmatrix}
\end{equation}
or, equivalently,
\begin{equation}
\label{eq:eigsL2}
    \ket{\lambda_{+}} = \begin{bmatrix}
         \cot\left ( \arctan \left ( \frac{\tan\theta}{1-2p} \right ) /2 \right )  \\
         e^{i \phi}
    \end{bmatrix} \, ,
    \ket{\lambda_{-}} = \begin{bmatrix}
        \tan\left ( \arctan \left ( \frac{\tan\theta}{1-2p} \right ) /2 \right )  \\
         e^{i \phi}
    \end{bmatrix} \, ,
\end{equation}
where we have ignored normalisation.
The optimal probe states
are then
\begin{equation}
    \rho_\mathrm{opt}^{\pm} = \frac{\ketbra{\lambda_{\pm}}}{\Tr(\ketbra{\lambda_{\pm}})} \, ,
\end{equation}
which are pure states
lying on the circle connecting~$\ket{0},\ket{1}$
and~$\ket{\theta,\phi}_\pm$
on the surface of the Bloch sphere
(see Fig.~\ref{fig:noisyqubitdetector}(b)).
Clearly, from Eqs.~\eqref{eq:eigsL}~\&~\eqref{eq:eigsL2},
the optimal probe states
have the same phase~$\phi$
as the azimuthal angle of the PVM.
The projectors and optimal probes
thus lie on a circular cross-section
of the Bloch sphere boundary at azimuthal
angle~$\phi$, as shown in Fig.~\ref{fig:noisyqubitdetector}(b).
Figure~\ref{fig:noisyqubitdetector}(b)
provides an intuitive and geometric explanation
of the optimal probe---it is the pure state
that when dephased, aligns perfectly
with the measurement direction.
The polar angles~$\theta_{\pm}$ of
these optimal probes
are given by~$\tan\theta_{\pm} = \pm \tan\theta/(1-2p)$,
which can be deduced geometrically
from Fig.~\ref{fig:noisyqubitdetector}(b).

\subsection{Detector with Inherent Depolarisation}

The depolarisation channel~$\mathcal{N}_{\mathrm{depol}}$ uniformly erases
phase and amplitude information from states,
contracting the Bloch ball radially towards its centre,
the maximally-mixed state.
A set of Kraus operators for~$\mathcal{N}_{\mathrm{depol}}$ are
\begin{equation}
    \mathcal{K}_1 = \sqrt{1-\frac{3p}{4}} \mathds{1}_2 \, , \quad \mathcal{K}_2 = \sqrt{\frac{p}{4}} X \, , \quad \mathcal{K}_3 = \sqrt{\frac{p}{4}} Y \, ,
    \quad \mathcal{K}_4 = \sqrt{\frac{p}{4}} Z \, ,
\end{equation}
where~$p$ represents the depolarising probability
and~$X,Y,Z$ denote the three Pauli matrices.
Given input state~$\rho_\mathrm{in}$,
the channel-output state~$\mathcal{N}_{\mathrm{depol}}(\rho_\mathrm{in})$
is projectively measured
along~$\ket{\theta, \phi}_\pm$.
%The measurement-outcome probabilities
%for the entire process are
%\begin{equation}
  %  \begin{split}
%        p_1 &= \Tr[ \mathcal{N}_{\mathrm{dep}}(\rho_\mathrm{in}) \ket{\theta,\phi}_+ \bra{\theta,\phi}]
%        = \Tr[\rho_\mathrm{in} \left ( \sum_{j\in[3]} \mathcal{K}^\dagger_j \ket{\theta,\phi}_+ \bra{\theta,\phi}\mathcal{K}_j \right ) ]  \, , \\
 %       p_2 &= \Tr[ \mathcal{N}_{\mathrm{dep}}(\rho_\mathrm{in}) \ket{\theta,\phi}_- \bra{\theta,\phi}] = \Tr[\rho_\mathrm{in} \left ( \sum_{j\in[3]} \mathcal{K}^\dagger_j \ket{\theta,\phi}_-
%  \bra{\theta,\phi}\mathcal{K}_j \right ) ]  \, .
%    \end{split}
%\end{equation}
The effective POVM~$\Pi_\mathrm{eff}$ for this process,
computed as per
Eq.~\eqref{eq:effpovm},
yields the two effective-detector
SLD operators,~$L^\mathrm{eff}_1$ and~$L^\mathrm{eff}_2$
for~$\pi^\mathrm{eff}_1$ and~$\pi^\mathrm{eff}_2$,
respectively, as
%The channel-output state,
%given input state~$\rho_\mathrm{in}$,
%is
%\begin{equation}
%    \mathcal{N}_{\mathrm{depol}}(\rho_\mathrm{in})
%    = \mathcal{K}_1 \rho_\mathrm{in} \mathcal{K}^\dagger_1
%    + \mathcal{K}_2 \rho_\mathrm{in} \mathcal{K}^\dagger_2
%    + \mathcal{K}_3 \rho_\mathrm{in} \mathcal{K}^\dagger_3
%     + \mathcal{K}_4 \rho_\mathrm{in} \mathcal{K}^\dagger_4\, ,
%\end{equation}
%and upon being
%projectively measured
%along the direction~$(\theta, \phi)$,
%the effective POVMs
%from Eq.~\eqref{eq:effpovm}
%are
%\begin{equation}
%    \pi_1 = \sum_{j\in[4]} \mathcal{K}^\dagger_j \ket{\theta,\phi}_+ \bra{\theta,\phi}\mathcal{K}_j  \quad \, \mathrm{and} \,  \quad \pi_2 = \sum_{j\in[4]} \mathcal{K}^\dagger_j \ket{\theta,\phi}_- \bra{\theta,\phi}\mathcal{K}_j \, ,
%\end{equation}
%and the corresponding
%SLD operators~$L_1$ \&~$L_2$
%are
\begin{equation}
    L^\mathrm{eff}_1 = \frac{1}{p(2-p)} \begin{bmatrix}
    1-p-\cos\theta & e^{-i\phi} \sin\theta \\
    e^{i\phi} \sin\theta & 1-p+\cos\theta
    \end{bmatrix}
    \, \,  \& \, \,
    L^\mathrm{eff}_2 = \frac{1}{p(2-p)} \begin{bmatrix}
   1-p+\cos\theta & e^{-i\phi} \sin\theta \\
    e^{i\phi} \sin\theta & 1-p-\cos\theta
    \end{bmatrix} \, .
\end{equation}
The matrix~$Q^\mathrm{eff}=\sum_{j\in[2]} L^\mathrm{eff}_j \pi^\mathrm{eff}_j L^\mathrm{eff}_j$
is then
\begin{equation}
    Q^\mathrm{eff} = \frac{1}{p(2-p)} \begin{bmatrix}
        1 & 0 \\
        0 & 1
    \end{bmatrix} \, ,
\end{equation}
from which, we
obtain the DQFIs
\begin{equation}
        \mathcal{J}_{\Vert, p} = \frac{1}{p(2-p)} \quad \mathrm{and} \quad
        \mathcal{J}_{\Tr, p} =  \frac{2}{p(2-p)} \, ,
\end{equation}
which do not depend on~$\theta$
or~$\phi$ at all.
Intuitively, this reflects that a PVM
along any direction is equally
informative about depolarised states.

Similar to the previous example
with dephasing noise,
matrix~$Q^\mathrm{eff}$ has degenerate
eigenvalues, so its eigenvectors
do not constitute optimal probe states.
Instead, the optimal probe states
are the common eigenstates
of~$L^\mathrm{eff}_1$ and~$L^\mathrm{eff}_2$,
which are precisely~$\ket{\theta,\phi}_+$
and~$\ket{\theta,\phi}_-$.
In other words,
the optimal probe states are
simply the pure states that
point along the PVM direction,
\begin{equation}
    \rho^{\pm}_{\mathrm{opt}} = \ketbra{\theta, \phi}_\pm \, ,
\end{equation}
as visualised in Fig.~\ref{fig:noisyqubitdetector}(c).

\subsection{Detector with Inherent Amplitude Damping}

Lastly,
we consider a qubit detector
with inherent amplitude damping (AD)
effects
and estimate its damping
probability~$p$.
The AD channel~$\mathcal{N}_{\mathrm{AD}}$ corresponds to
a partial thermalisation
of the input state,
treating~$\ket{0}$ as the ground state
and
transforming all other states
towards this ground state,
with parameter~$p$ quantifying
the extent of thermalisation.
A set of Kraus operators for~$\mathcal{N}_{\mathrm{AD}}$ are
\begin{equation}
    \mathcal{K}_1 = \ketbra{0} + \sqrt{1-p} \ketbra{1} \, , \quad \mathcal{K}_2 = \sqrt{p} \ket{0}\bra{1} \, ,
\end{equation}
and
given input state~$\rho_\mathrm{in}$,
the channel-output state~$\mathcal{N}_{\mathrm{AD}}(\rho_\mathrm{in})$
is projectively measured
along~$\ket{\theta, \phi}_{1,2}$.
%The measurement-outcome probabilities
%for the entire process are
%\begin{equation}
  %  \begin{split}
%        p_1 &= \Tr[ \mathcal{N}_{\mathrm{dep}}(\rho_\mathrm{in}) \ket{\theta,\phi}_+ \bra{\theta,\phi}]
%        = \Tr[\rho_\mathrm{in} \left ( \sum_{j\in[3]} \mathcal{K}^\dagger_j \ket{\theta,\phi}_+ \bra{\theta,\phi}\mathcal{K}_j \right ) ]  \, , \\
 %       p_2 &= \Tr[ \mathcal{N}_{\mathrm{dep}}(\rho_\mathrm{in}) \ket{\theta,\phi}_- \bra{\theta,\phi}] = \Tr[\rho_\mathrm{in} \left ( \sum_{j\in[3]} \mathcal{K}^\dagger_j \ket{\theta,\phi}_-
%  \bra{\theta,\phi}\mathcal{K}_j \right ) ]  \, .
%    \end{split}
%\end{equation}
The effective POVM~$\Pi_\mathrm{eff}$ for this process,
computed as per
Eq.~\eqref{eq:effpovm},
yields the two effective-detector
SLD operators,~$L^\mathrm{eff}_1$ and~$L^\mathrm{eff}_2$
for~$\pi^\mathrm{eff}_1$ and~$\pi^\mathrm{eff}_2$,
respectively, as
\begin{equation}
    L^\mathrm{eff}_1 = \frac{1}{2p(1+p \cos\theta)} \begin{bmatrix}
    1-\cos\theta & - \frac{e^{-i\phi} \sin\theta}{\sqrt{1-p}} \\
    - \frac{ e^{i\phi} \sin\theta}{\sqrt{1-p}} & 1+(1+2p)\cos\theta
    \end{bmatrix}
    \, \,  \& \, \,
    L^\mathrm{eff}_2 = \frac{1}{2p(1-p \cos\theta)} \begin{bmatrix}
    1+\cos\theta &  \frac{e^{-i\phi} \sin\theta}{\sqrt{1-p}} \\
     \frac{ e^{i\phi} \sin\theta}{\sqrt{1-p}} & 1-(1+2p)\cos\theta
    \end{bmatrix} \, .
\end{equation}
From the matrix~$Q^\mathrm{eff}=\sum_{j\in[2]} L^\mathrm{eff}_j \pi^\mathrm{eff}_j L^\mathrm{eff}_j$,
which we omit for brevity,
the two DQFIs
can now be calculated as
\begin{equation}
\begin{split}
\label{eq:ADchanndetQFI}
        \mathcal{J}_{\Vert, p} &= \frac{1}{4 p (1-p)}
        + \frac{(1+p)\cos^2\theta}{4 p (1-p^2\cos^2\theta)}
        + \frac{\cos\theta \sqrt{1-p \cos^2\theta} ((1-p^2)(1+p) \cos^2\theta + \sin^2\theta)}{2 p \sqrt{1-p} \left ( (1-p^2) \cos^2\theta + \sin^2\theta \right )^2 }\\
        \mathcal{J}_{\Tr, p} &=  \frac{1+(1-2p^2) \cos^2\theta}{2p(1-p)(1-p^2 \cos^2\theta)}  \, ,
\end{split}
\end{equation}
which depend
only on~$\theta$ and not on~$\phi$.
The matrix~$Q$ has non-degenerate or
distinct eigenvalues, so
the DQFI-optimal probe state is simply
the eigenvector of~$Q$ corresponding
to the eigenvalue given by~$\mathcal{J}_{\Vert, p}$.

However,
the two SLD operators~$L^\mathrm{eff}_1$
and~$L^\mathrm{eff}_2$ do not commute
in this case
and do not share any common eigenvectors.
Thus, the DQFI~$\mathcal{J}_{\Vert, p}$
is not attainable, and only upper-bounds
the maximum CFI over the input probe space
for this noise model.

\subsection{Comparison of State
and Detector Approaches to
Optimal Sensing of Channel Parameters}
\label{sec:ChannelQPT}

We now change perspective
from detector parameter estimation
to channel parameter estimation
and reconsider the noisy
qubit channels studied above.
In this scenario,
the measurement~$\Pi_\mathrm{ideal}$
from Fig.~\ref{fig:noisyqubitdetector}
is no longer intrinsically fixed
and, therefore, can be optimised
or tuned to attain higher
sensitivity to channel parameters.
The state estimation approach
to channel parameter estimation
involves maximising
the SQFI~$\mathcal{I}_p[\rho_{\mathrm{out},p}]$
of the channel-output state~$\rho_{\mathrm{out},p}$
over all input probe states~$\rho_\mathrm{in}$~\cite{Meyer2021}.
The detector estimation framework presented
herein
offers an alternative approach:
maximising the DQFI~$\mathcal{J}_{\Vert, p}[\Pi_{\mathrm{eff}, p}]$
of the effective measurement~$\Pi_{\mathrm{eff}}$
over all actual measurements~$\Pi$
(see main-text Fig.~\ref{fig:processestimation}).
We now illustrate the equivalence
of these two approaches
for the three noisy qubit channels considered above.

In the previous subsections, we had fixed
the PVM~$\Pi_\mathrm{ideal}$ to be along
projectors~$\ket{\theta,\phi}_{\pm}$.
For a comparison with the state approach,
we now
consider the probe state~$\rho_\mathrm{in}$
to be a pure state~$\ket{\theta_\mathrm{in}, \phi_\mathrm{in}}$
where~$\theta_\mathrm{in}$ and~$\phi_\mathrm{in}$ represent
the polar and azimuthal angles of this state
on the Bloch sphere.\footnote{For estimating a single parameter
of the channel, pure states are optimal
due to the convexity of the CFI;
see Eq.~\eqref{eq:additiveCFI} and discussion
in Methods.}

For the dephasing channel,
the output state~$\rho_{\mathrm{out}, p}$
corresponding to input~$\ket{\theta_\mathrm{in}, \phi_\mathrm{in}}$
attains
an SQFI~$\mathcal{I}_p[\rho_{\mathrm{out}, p}] = \frac{\sin^2\theta_\mathrm{in}}{p(1-p)}$.
Clearly, using~$ \mathcal{J}_{\Vert, p}[\Pi_{\mathrm{eff},p}]$
from Eq.~\eqref{eq:dephasingdetqfi},
\begin{equation}
    \max_{\theta_\mathrm{in}, \phi_\mathrm{in}} \mathcal{I}_p [\rho_{\mathrm{out}, p}] = \frac{1}{p(1-p)} = \max_{\theta, \phi} \mathcal{J}_{\Vert, p}[\Pi_{\mathrm{eff},p}].
\end{equation}

For the depolarising channel, due to its symmetry,
both the SQFI~$\mathcal{I}_p[\rho_{\mathrm{out}, p}]$
and the DQFI~$\mathcal{J}_{\Vert, p}[\Pi_{\mathrm{eff},p}]$
equal~$\frac{1}{p(2-p)}$ for any~$(\theta_\mathrm{in}, \phi_\mathrm{in})$ and any~$(\theta,\phi)$, respectively.
The equivalence~$\max_{\theta_\mathrm{in}, \phi_\mathrm{in}} \mathcal{I}_p [\rho_{\mathrm{out}, p}]= \max_{\theta, \phi} \mathcal{J}_{\Vert, p}[\Pi_{\mathrm{eff},p}]$ thus holds trivially.

For the amplitude damping channel,
the SQFI of the channel output is
\begin{equation}
    \mathcal{I}_p [\rho_{\mathrm{out}, p}] = \frac{((1+p)-(1-p)\cos\theta_\mathrm{in})\sin^2(\theta_\mathrm{in}/2)}{2 p(1-p)} \, ,
\end{equation}
which is maximised for~$\theta_\mathrm{in}=\pi$, i.e.,
for probing with the input state~$\ket{1}$.
This optimum probe attains a SQFI~$\mathcal{I}_p$ equal to~$\frac{1}{p(1-p)}$,
whereas from Eq.~\eqref{eq:ADchanndetQFI},
the maximum effective DQFI is
\begin{equation}
    \max_{\theta, \phi} \mathcal{J}_{\Vert, p}[\Pi_{\mathrm{eff},p}] = \frac{1}{p(1-p)}
\end{equation}
with maximum attained at~$\theta = 0$,
i.e., for measuring in the~$Z$-basis.
Hence,
the equality~$\max_{\theta_\mathrm{in}, \phi_\mathrm{in}} \mathcal{I}_p [\rho_{\mathrm{out}, p}]= \max_{\theta, \phi} \mathcal{J}_{\Vert, p}[\Pi_{\mathrm{eff},p}]$
holds.
Moreover,
the optimal measurement
for probing with state~$\ket{1}$
%(corresponding to~$\theta_\mathrm{in}=\pi$)
according to the SQFI
is~$\theta=0$,
meaning a~$Z$-basis measurement.
Simultaneously,
the optimal probe state
for a measurement along the~$Z$-basis,
as per the DQFI,
is the probe state with~$\theta_\mathrm{in}=\pi$,
or~$\ket{1}$.
In other words,
the two ways
of finding the optimal probe state
and the optimal measurement
for channel estimation
(left and right halves of main-text
Fig.~\ref{fig:processestimation})
agree with
each other,
in addition to the optimal QFI for both approaches being equal.

\section{Experimental detector estimation of dephasing on IBM platform}
\label{app:IBMExp}

For the experiment on IBM
quantum computers,
a set of~33 pure probe states
are chosen on the Bloch sphere,
each with phase or azimuthal angle~$\phi=0$
and with polar angle~$\theta_{\mathrm{in}}$
ranging from~$0.2$ to~$1$ radian.
For each choice of~$\theta_\mathrm{in}$,
the quantum circuit is executed
with~$10^5$
shots,
with true value of~$p$ assigned~$p^*=0.2$.
The quantum processor returns
probabilities~$p_0$ and~$p_1$
for obtaining~0 and~1 as the measurement outcome.
A comparison of the empirical~$p_0$ values
(purple points against left vertical axis) with
theoretically-expected values (blue curve) for all~33 probe states
is shown in Fig.~\ref{fig:IBMQProbzeroEstimatedDephasing}(a).

\begin{figure}[ht]
    \centering
    \includegraphics[width=\columnwidth]{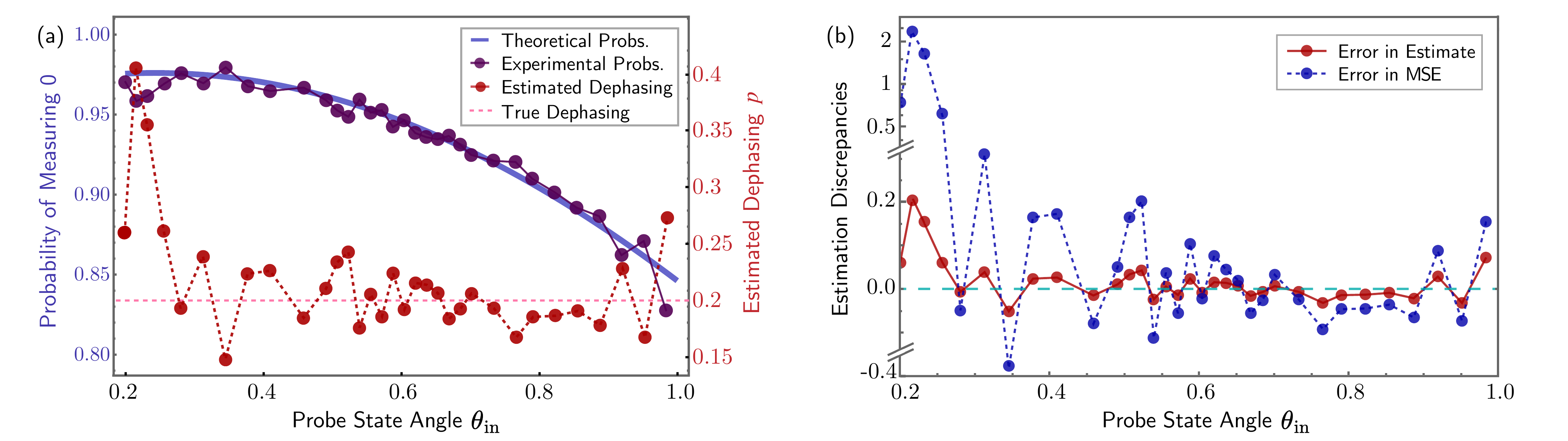}
    \caption{Experimental measurement
    probabilities,
    estimated dephasing
    and MSEs compared to theory
    for range of probe states.
    (a) The probability
    of measuring~0 on executing the circuit
    in main-text Fig.~\ref{fig:DephasedDetIBM}(c) is
    plotted along the left
    vertical axis (blue),
    comparing experimental points (purple dots)
    with theory (blue curve).
    The estimated dephasing
    is plotted along
    the right vertical axis (red)
    comparing
    experimental estimates (red dots)
    to the true value~$0.2$ (horizontal dashed pink line).
    (b) Discrepancy or bias in estimated dephasing (red dots)
    compared to the deviation of empirical MSE from theoretical MSE (blue dots).
    The bias and MSE deviations are correlated
    as expected.}
    \label{fig:IBMQProbzeroEstimatedDephasing}
\end{figure}

\begin{figure}[hbtp]
    \centering
    \includegraphics[width=\columnwidth]{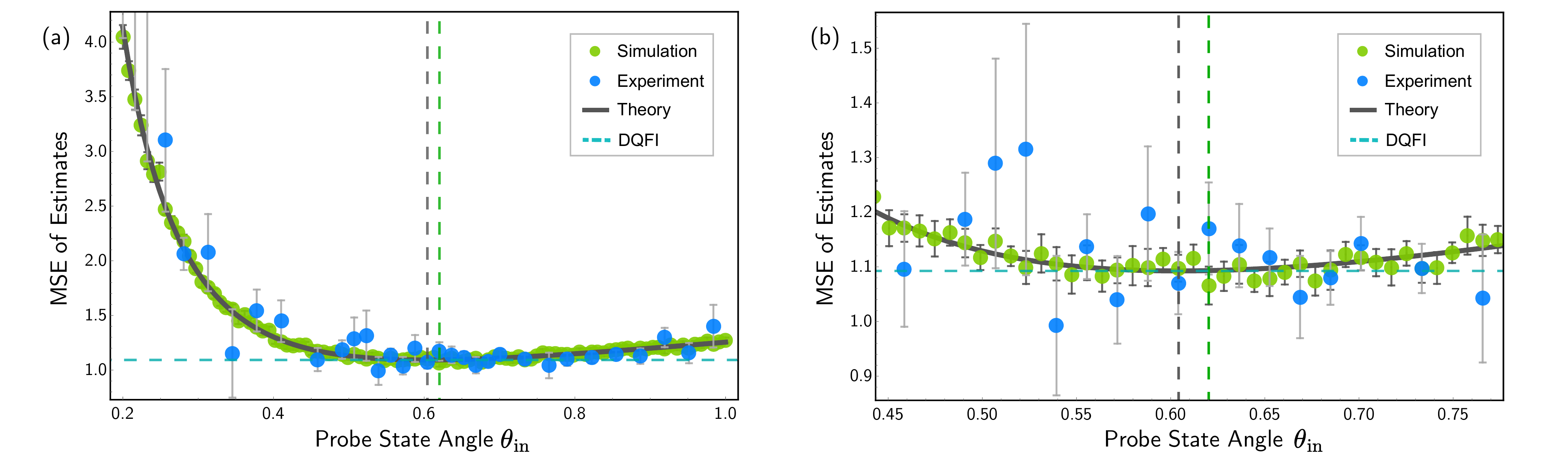}
    \caption{MSE of estimates for optimal and sub-optimal
    probe states. The empirical MSE of estimated dephasing
    (blue dots, light grey error bars)
    and the MSE for noiseless simulation
    (green dots, dark grey error bars)
    are compared to theoretical MSE from the CFI
    (dark grey curve)
    for probe states with polar angle~$\theta_\mathrm{in}$
    ranging from (a)~0.2 to~1 rad
    and (b)~0.45 to 0.78 rad.
    The dashed horizontal line (blue) corresponds to the
    DQFI lower bound~$1/\mathcal{J}_{\Vert, p}$
    whereas the dashed vertical lines (grey \& green)
    correspond to the optimal probe angles
    (theory \& simulation, respectively)
    based on raw data.}
    \label{fig:ExpSuppMatPlots}
\end{figure}

The theoretical values of~$p_0$
plotted in Fig.~\ref{fig:IBMQProbzeroEstimatedDephasing}(a)
are from
\begin{equation}
    p_{0,1} = \frac{1}{2} \left ( 1 \pm \cos(\pi/8) \cos\theta_\mathrm{in} \pm (1-2p) \sin(\pi/8) \sin\theta_\mathrm{in} \right ) \, ,
\end{equation}
combining which produces
an unbiased estimator~$\hat{p}$ for the dephasing strength~$p$~\cite{MLEQM01},
\begin{equation}
\label{eq:estimatepeq}
    \hat{p} = \frac{1}{2} + \frac{1}{2 \tan (\pi/8 ) \tan (\theta_\mathrm{in})}-\frac{p_0-p_1}{2 \sin (\pi/8 )  \sin \theta_\mathrm{in}} \equiv \alpha \,  p_0 + \beta \, p_1 \, ,
\end{equation}
where, in the last step,
we rewrite the constant terms as~$p_0 + p_1$
to get an expression linear in the probabilities.
Notably, the linear estimator~$\hat{p}$
could, in principle, yield unphysical estimates
if the true value~$p^*$ were close to~0 or~1
and the number of samples was insufficient.
For large numbers of samples or true values
inside the interior, this is
unlikely to occur due to the theory of large deviations~\cite{Varadhan2008}.
Nonetheless, for ensuring physical estimates at low sample numbers,
alternative methods like
constrained maximum-likelihood could be employed for unbiased
estimation, or region estimation could be used
instead of point estimation,
or Bhattacharyya bounds~\cite{Unbiased25} or Barankin bounds~\cite{BBound24}
could be used in place of Cram\'{e}r-Rao bounds.

The theoretically-predicted MSE in
estimating~$p$ is then given by~$V_\text{th} = \alpha^2 p_0 + \beta^2 p_1 - {p^*}^2$~\cite{Lorcan21},
which is plotted as the grey curve
in main-text Fig.~\ref{fig:DephasedDetIBM}(d)
and supplemental Figs.~\ref{fig:ExpSuppMatPlots}(a) and (b).
Experimentally, the empirical values
for~$p_0$ and~$p_1$
returned by the quantum processor
are used in Eq.~\eqref{eq:estimatepeq}
to obtain estimates of~$p$
for each
probe angle~$\theta_\mathrm{in}$,
which are plotted against
the true value in Fig.~\ref{fig:IBMQProbzeroEstimatedDephasing}(a)
(red points against right vertical axis).
Fig.~\ref{fig:IBMQProbzeroEstimatedDephasing}(a)
shows that the estimated values deviate
further from the true value for both very low
probe angles and
and very high probe angles,
in contrast to intermediate probe angles
where the deviations are smaller.

The empirical MSE of these estimates
across all shots is bootstrapped
from the experimental distribution and is plotted
in blue in main-text
Fig.~\ref{fig:DephasedDetIBM}(d)
and supplemental Fig.~\ref{fig:ExpSuppMatPlots}.
The simulation points in green
in main-text
Fig.~\ref{fig:DephasedDetIBM}(d)
and supplemental Fig.~\ref{fig:ExpSuppMatPlots}
are produced similarly
by processing values
for~$p_0$
and~$p_1$ returned by IBM's
noiseless quantum simulator.
In both cases,
the error bars for data points
are obtained by repeating
the entire process of computing empirical
MSEs~$50$ times
and calculating the standard deviation
of the~50 MSEs.
The MSEs and
and their standard deviations are then
compared to the spectral QCRB
for the same number of samples,
by instead plotting
the scaled-up MSEs
against~$1/\QFIsp$.

The discrepancies between theory
and experiment are summarised in Fig.~\ref{fig:IBMQProbzeroEstimatedDephasing}(b).
For each probe angle,
the difference between estimated dephasing
and its true value~$p^*$,
also known as bias,
is plotted in red (joined by continuous lines),
whereas the difference between the
empirically-found MSE and the theoretically-predicted
MSE is plotted in blue (joined by dashed lines).
From Fig.~\ref{fig:IBMQProbzeroEstimatedDephasing}(b),
it is clear that deviations of the experimental MSE
(blue points in Fig.~\ref{fig:ExpSuppMatPlots})
from theory (grey curve)
are perfectly correlated to systematic
bias in the estimated values,
attributable to platform noise including
state preparation, gate and measurement errors
at the circuit level.
The root-mean-squared bias across all probe states
is of the order of~$1\%$, and the total median error rate
for circuit errors on the IBM machine used are also of the order of~$1\%$. Moreover,
the cumulative median error rate combining state preparation
and gate errors on this machine are of the order of~$0.01\%$,
compared to a median readout or measurement error rate
of the order of~$1\%$.
This allows the relatively precise
preparation of optimal probe states
for the calibration of relatively noisier detectors,
precluding the cyclic requirement of well-characterised states
for precise calibration of detectors and well-characterised detectors
for the precise calibration of states.

\section{DQFI as Bures distance between nearby POVMs}
\label{supp:BuresDistance}

%\paragraph{Bures Metric for Positive Operators:}
In this section,
we establish a connection between
the proposed DQFI~$\QFItr$
and the distance between nearby POVM measurements.
Specifically, we consider
the Bures metric defined on the space
of positive operators,
\begin{equation}
    \mathcal{D}_\mathrm{Bures}(A, B) \coloneqq \Tr(A) + \Tr(B) - 2 \mathcal{F}(A, B) \, ,
\end{equation}
where~$\mathcal{F}$ denotes the fidelity,
\begin{equation}
\label{eq:fiddef}
    \mathcal{F} \coloneqq \Tr \left ( \sqrt{\sqrt{A} B \sqrt{A}} \right )  \, ,
\end{equation}
and~$\sqrt{A}$ denotes the positive square-root of positive operator~$A$~\cite{SSW22}.
Noting that~$\mathcal{D}_\mathrm{Bures}(A,B) = \mathcal{D}_\mathrm{Bures}(B,A) > 0$ for~$A\neq B$,
which follows from~$\mathcal{D}_\mathrm{Bures}$ being a metric,
we define the total distance
\begin{equation}
    \mathcal{D}_\mathrm{total} \left ( \{\pi^{(1)}_j\}_{j\in[m]}, \{\pi^{(2)}_j\}_{j\in[m]} \right )
    \coloneqq \sum_{j\in [m]} \mathcal{D}_\mathrm{Bures} \left (\pi^{(1)}_j, \pi^{(2)}_j \right ) \, ,
\end{equation}
where it is apparent that the comparison can only
be made for POVMs with the same number of elements.
Also, due to~$\sum_{j\in[m]} \pi_j = \mathds{1}_d$
for any POVM, we can write
\begin{equation}
    \mathcal{D}_\mathrm{total}\left ( \{\pi^{(1)}_j\}_{j\in[m]}, \{\pi^{(2)}_j\}_{j\in[m]} \right ) = 2 d  - 2 \sum_{j\in[m]} \mathcal{F} \left (\pi^{(1)}_j, \pi^{(2)}_j \right )  \, .
\end{equation}
Now, given
%\paragraph{Distance between infinitesimally-distant POVMs}
a parametrised POVM,
we can now compute the total distance~$\mathcal{D}_\mathrm{total}\left(\{{\pi_j}_\theta\}_{j\in[m]}, \{{\pi_{\theta+\delta\theta}}_j\}_{j\in[m]} \right )$.
The calculation is lengthy, but our goal is to prove
that up to second order in~$\delta\theta$,
\begin{equation}
    \mathcal{D}_\mathrm{total}\left(\{{\pi_j}_\theta\}_{j\in[m]}, \{{\pi_{\theta+\delta\theta}}_j\}_{j\in[m]} \right ) = \frac{1}{4} \QFItr \, \delta\theta^2 + O(\delta\theta^3) \, .
\end{equation}
To ease readability,
the proof is divided into steps.

\step{1}
Let us start with the infinitesimal fidelity
for
a single POVM element~${\pi_j}_\theta$
($1\leq j \leq m$),
\begin{equation}
    \mathcal{F}({\pi_j}_\theta, {\pi_{\theta+\delta\theta}}_j) =  \mathcal{F}\left ({\pi_j}_\theta, {\pi_j}_\theta + \left (\partial_\theta {\pi_j}_\theta \right) \delta \theta + \left (\partial^2_\theta {\pi_j}_\theta \right ) \frac{\delta \theta^2}{2}  + O(\delta\theta^3) \right) \, ,
\end{equation}
where we have used the Taylor expansion of~${\pi_j}_\theta$
assuming that the parametrisation
of~${\pi_j}_\theta$ is at least
twice-differentiable.
The fidelity definition (Eq.~\eqref{eq:fiddef})
then reduces to
\begin{equation}
    \mathcal{F}({\pi_j}_\theta, {\pi_{\theta+\delta\theta}}_j) = \Tr \left ( \sqrt{ \sqrt{{\pi_j}_\theta} \left ( {\pi_j}_\theta + \left (\partial_\theta {\pi_j}_\theta \right) \delta \theta + \left (\partial^2_\theta {\pi_j}_\theta \right ) \frac{\delta \theta^2}{2}  + O(\delta\theta^3) \right ) \sqrt{{\pi_j}_\theta}}\right ) \, .
\end{equation}
Let us rewrite this last expression as~$\Tr\left ( \sqrt{X} \right )$,
with~$\sqrt{X}$ being the unique positive Hermitian square-root of~$X$.

\step{2}
By expanding~$X$, we get
\begin{equation}
\label{eq:Xexp}
\begin{split}
    X &= \sqrt{{\pi_j}_\theta} \left ( {\pi_j}_\theta + \left (\partial_\theta {\pi_j}_\theta \right) \delta \theta + \left (\partial^2_\theta {\pi_j}_\theta \right ) \frac{\delta \theta^2}{2}  + O(\delta\theta^3) \right ) \sqrt{{\pi_j}_\theta}  \\
    &=   {\pi_\theta}^2_j + \sqrt{{\pi_j}_\theta}  \left (\partial_\theta {\pi_j}_\theta \right) \sqrt{{\pi_j}_\theta}  \delta \theta + \sqrt{{\pi_j}_\theta}
    \left (\partial^2_\theta {\pi_j}_\theta \right ) \sqrt{{\pi_j}_\theta}  \frac{\delta   \theta^2}{2}  + O(\delta\theta^3)  \, ,
\end{split}
\end{equation}
whereas~$X$ (being a product of three positive
semi-definite matrices) is positive
as well as Hermitian.
Thus,~$X$ has a unique Hermitian and positive
square root.
Let us assume
this square-root is of
the form\footnote{It is clear that
the form in Eq.~\eqref{eq:Xguess}
is Hermitian if and only if~$\mathcal{B}$
and~$\mathcal{C}$ are Hermitian.
We defer the proof of
positivity
of this form to Lemma~\ref{lemma:pos}.}
\begin{equation}
    \label{eq:Xguess}
    \sqrt{X} = {\pi_j}_\theta + \mathcal{B} \,  \delta\theta + \mathcal{C} \, \delta\theta^2 \, ,
\end{equation}
for yet-to-be-determined Hermitian matrices~$\mathcal{B}$
and~$\mathcal{C}$.
To solve for~$\mathcal{B}$ and~$\mathcal{C}$,
we solve~$\sqrt{X} \sqrt{X}$ equal to~$X$
from Eq.~\eqref{eq:Xexp}.
The resulting equations for~$\mathcal{B}$
and~$\mathcal{C}$ are
\begin{align}
    \label{eq:B}
        {\pi_j}_\theta \, \mathcal{B} + \mathcal{B} \, {\pi_j}_\theta &= \sqrt{{\pi_j}_\theta}  \left ( \partial_\theta {\pi_j}_\theta \right ) \sqrt{{\pi_j}_\theta} \, \,  , \\
        \label{eq:C}
        {\pi_j}_\theta \, \mathcal{C}  + \mathcal{C} \, {\pi_j}_\theta + \mathcal{B}^2 &= \frac12 \sqrt{{\pi_j}_\theta} \left ( \partial^2_\theta {\pi_j}_\theta \right ) \sqrt{{\pi_j}_\theta} \, \, .
\end{align}

\step{3}
Comparing Eq.~\eqref{eq:B}
to~${\pi_j}_\theta {L_j}_\theta + {L_j}_\theta {\pi_j}_\theta = 2 \, \partial_\theta {\pi_j}_\theta$
or, equivalently, to
\begin{equation}
    {\pi_j}_\theta
    \sqrt{{\pi_j}_\theta}  {L_j}_\theta \sqrt{{\pi_j}_\theta}  + \sqrt{{\pi_j}_\theta} {L_j}_\theta \sqrt{{\pi_j}_\theta} {\pi_j}_\theta =  2  \sqrt{{\pi_j}_\theta}
    \left ( \partial_\theta {\pi_j}_\theta \right )
\sqrt{{\pi_j}_\theta} \, ,
\end{equation}
gives us~$\mathcal{B} = \frac12 \sqrt{{\pi_j}_\theta} {L_j}_\theta \sqrt{{\pi_j}_\theta}$.
This solution for~$\mathcal{B}$
implies
$$\Tr( \mathcal{B} ) = \frac12 \Tr ( {\pi_j}_\theta \, {L_j}_\theta ) = \frac12 \Tr ( \partial_\theta \, {\pi_j}_\theta \, ) \, , $$
and, upon
inserting
into Eq.~\eqref{eq:C}, results in
\begin{equation}
\label{eq:C2}
    {\pi_j}_\theta \, \mathcal{C}  + \mathcal{C} \, {\pi_j}_\theta = \frac12 \sqrt{{\pi_j}_\theta} \left ( \partial^2_\theta {\pi_j}_\theta - \frac12 {L_j}_\theta {\pi_j}_\theta {L_j}_\theta\right ) \sqrt{{\pi_j}_\theta} \, \, ,
\end{equation}
yet another Sylvester equation.

\step{4}
To solve Eq.~\eqref{eq:C2} for~$\mathcal{C}$,
we need to express
the second-derivative~$\partial^2_\theta \, {\pi_j}_\theta$
in terms of the SLD operators~${L_j}_\theta$.
Differentiating~${\pi_j}_\theta  \, {L_j}_\theta
+ {L_j}_\theta \,  {\pi_j}_\theta
= 2 \, \partial_\theta {\pi_j}_\theta$
with respect to~$\theta$ yields
\begin{equation}
\label{eq:secder}
    \partial^2_\theta \, {\pi_j}_\theta
    = \frac12 \left \{ {\pi_j}_\theta , \partial_\theta \, {L_j}_\theta \right \}
    + \frac14 \left \{  {\pi_j}_\theta \, , {L_j}_\theta^2 \, \right \}  + \frac12 {L_j}_\theta \, {\pi_j}_\theta \, {L_j}_\theta  \, ,
\end{equation}
where~$\{ \, , \, \}$ denotes the anti-commutator.
Eq.~\eqref{eq:C2} then becomes
\begin{equation}
\begin{split}
    {\pi_j}_\theta \, \mathcal{C}  + \mathcal{C} \, {\pi_j}_\theta &= \frac12 \sqrt{{\pi_j}_\theta} \left (  \frac12 \left \{ {\pi_j}_\theta , \partial_\theta \, {L_j}_\theta \right \}
    + \frac14 \left \{  {\pi_j}_\theta \, , {L_j}_\theta^2 \, \right \}  \right ) \sqrt{{\pi_j}_\theta} \, \, \\
    &= \frac14 \sqrt{{\pi_j}_\theta}   \left \{ {\pi_j}_\theta , \partial_\theta \, {L_j}_\theta + \frac12 \,  {L_j}_\theta^2 \,  \right \}
     \sqrt{{\pi_j}_\theta} \, \,  \\
    &=    \left \{ {\pi_j}_\theta , \frac14 \sqrt{{\pi_j}_\theta} \left (  \partial_\theta \, {L_j}_\theta + \frac12 \,  {L_j}_\theta^2 \right ) \,  \sqrt{{\pi_j}_\theta} \,  \right \} ,
\end{split}
\end{equation}
which immediately implies~$\mathcal{C}
=\frac14 \sqrt{{\pi_j}_\theta} \left (  \partial_\theta \, {L_j}_\theta
+ \frac12 \,  {L_j}_\theta^2 \right ) \,  \sqrt{{\pi_j}_\theta}$.

\step{5}
By direct computation, we have
$$\Tr(\mathcal{C})
= \frac14 \Tr\left ( {\pi_j}_\theta
\left (  \partial_\theta \, {L_j}_\theta + \frac12 \,  {L_j}_\theta^2 \right ) \right ) \, ,$$
whereas
from Eq.~\eqref{eq:secder},
$$\Tr \left ( \partial^2_\theta \, {\pi_j}_\theta\right ) = \Tr\left ( {\pi_j}_\theta
\left (  \partial_\theta \, {L_j}_\theta + \frac12 \,  {L_j}_\theta^2 \right )
+ \frac12{L_j}_\theta \, {\pi_j}_\theta \, {L_j}_\theta  \right ) \, ,$$
so that~$\Tr \left ( \partial^2_\theta \, {\pi_j}_\theta\right )
= 4 \Tr(\mathcal{C})
+ \frac12 \Tr( {L_j}_\theta \, {\pi_j}_\theta \, {L_j}_\theta )$.
Moreover, using the Taylor expansion of~${\pi_j}_\theta$, we have
$$\Tr({\pi_{\theta+\delta \theta}}_j)
=   \Tr( {\pi_j}_\theta )
+  \Tr \left (\partial_\theta {\pi_j}_\theta \right) \delta \theta +  \Tr \left (\partial^2_\theta {\pi_j}_\theta \right ) \frac{\delta \theta^2}{2}  =
\Tr( {\pi_j}_\theta )  + 2 \Tr(\mathcal{B}) \delta \theta + \left ( 2 \Tr(\mathcal{C}) + \frac14 \Tr ( {L_j}_\theta \, {\pi_j}_\theta \, {L_j}_\theta ) \right ) \delta \theta^2. $$

\step{6}
Using~$\mathcal{B}$ and~$\mathcal{C}$
to rewrite the infinitesimal fidelity
%solutions to~$\mathcal{B}$ \&~$\mathcal{C}$,
%which also imply
%\begin{equation}
%    \Tr( \mathcal{B} ) = \frac12 \Tr ( {\pi_j}_\theta \, {L_j}_\theta ) = \frac12 \Tr ( \partial_\theta \, {\pi_j}_\theta \, ) \, \, ,
%\end{equation}
%and
%\begin{equation}
%    \Tr( \mathcal{C}) = \frac14 \Tr \left ( {\pi_j}_\theta \, \left ( \partial_\theta \, {L_j}_\theta \,  \right ) \right )
%    + \frac18 \Tr \left ( {L_j}_\theta \, {\pi_j}_\theta \, {L_j}_\theta \right )  \, ,
%\end{equation}
%we get
\begin{equation}
%\begin{split}
    \mathcal{F}({\pi_j}_\theta, {\pi_{\theta+\delta\theta}}_j) = \Tr(\sqrt{X}) = \Tr({\pi_j}_\theta) \, + \Tr(\mathcal{B}) \,  \delta \theta
    + \Tr(\mathcal{C}) \, \delta \theta^2  \, ,
%    &=  \Tr({\pi_j}_\theta) + \frac12 \Tr ( \partial_\theta \, {\pi_j}_\theta \, ) \delta \theta + \frac14 \left ( \Tr \left ( {\pi_j}_\theta \, \left ( \partial_\theta \, {L_j}_\theta \,  \right ) \right )
%    +  \frac12 \Tr \left ( {L_j}_\theta \, {\pi_j}_\theta \, {L_j}_\theta \right ) \right ) \delta \theta^2 \\
%    &=  \frac12 \Tr({\pi_j}_\theta) + \frac12 \Tr({\pi_j}_{\theta+\delta \theta})  - \frac18  \Tr \left ( {L_j}_\theta \, {\pi_j}_\theta \, {L_j}_\theta \right )  \delta \theta^2 \, \, .
%\end{split}
\end{equation}
a straight-forward calculation
of~$\mathcal{D}_\mathrm{Bures}({\pi_j}_\theta \, , \, {\pi_j}_{\theta + \delta \theta} \,)$
leads to
\begin{equation}
\label{eq:distancelement}
    \mathcal{D}_\mathrm{Bures}({\pi_j}_\theta \, , \, {\pi_j}_{\theta + \delta \theta}) \, = \Tr({\pi_j}_\theta) + \Tr({\pi_j}_{\theta + \delta \theta}) - 2 \mathcal{F}({\pi_j}_\theta, {\pi_{\theta+\delta\theta}}_j) = \frac14 \Tr \left ( {L_j}_\theta \, {\pi_j}_\theta \, {L_j}_\theta \right )  \delta \theta^2 \, .
\end{equation}

\step{7}
Summing over~$j\in [m]$
in Eq.~\eqref{eq:distancelement}
then results in
\begin{equation}
\label{eq:totaldistance}
     \mathcal{D}_\mathrm{total}\left(\{{\pi_j}_\theta\}_{j\in[m]}, \{{\pi_{\theta+\delta\theta}}_j\}_{j\in[m]} \right ) = \frac14 \, \QFItr \,  \delta \theta^2 \, ,
\end{equation}
as required. \qed
%for~$J_\theta$ the trace QFI,~$J_\theta = \Tr[\sum_{j\in[m]} {L_j}_\theta {\pi_j}_\theta {L_j}_\theta]$.

\begin{lemma}
\label{lemma:pos}
    The following solutions for~$\mathcal{B}$ and~$\mathcal{C}$
    \begin{align*}
        \mathcal{B} &= \frac12 \sqrt{{\pi_j}_\theta} {L_j}_\theta \sqrt{{\pi_j}_\theta} \, , \\
        \mathcal{C} &=\frac14 \sqrt{{\pi_j}_\theta} \left (  \partial_\theta \, {L_j}_\theta
+ \frac12 \,  {L_j}_\theta^2 \right ) \,  \sqrt{{\pi_j}_\theta}
    \end{align*}
    guarantee that the form in
    Eq.~\eqref{eq:Xguess},
    i.e.,~$\sqrt{X} = {\pi_j}_\theta + \mathcal{B} \delta \theta + \mathcal{C} \delta\theta^2$,
    is Hermitian and positive semi-definite.
\end{lemma}
\begin{proof}
    Clearly,~$\mathcal{B}$ and~$\mathcal{C}$
    are Hermitian, thus implying that~$\sqrt{X}$
    is also Hermitian.
    For positivity,
    note that we can rewrite~$\sqrt{X}$
    as
    \begin{equation}
        \sqrt{X} = \sqrt{{\pi_j}_\theta} \left ( \mathds{1}_d + \frac12 {L_j}_\theta
 \delta \theta + \frac14 \left ( \partial_\theta {L_j}_\theta \right ) \delta \theta^2 + \frac18 {L_j}_\theta^2 \delta \theta^2 \right ) \sqrt{{\pi_j}_\theta} \, .
    \end{equation}
    Here,~$\sqrt{{\pi_j}_\theta}$ is Hermitian and positive semi-definite,
    and for small enough~$\delta \theta$,
    the same is true of~$ \mathds{1}_d + \frac12 {L_j}_\theta
    \delta \theta + \frac14 \left ( \partial_\theta {L_j}_\theta \right ) \delta \theta^2 + \frac18 {L_j}_\theta^2 \delta \theta^2$.
    As the product of positive semi-definite matrices
    is also positive semi-definite,
    this proves our claim.
\end{proof}

\begin{corollary}
\label{corr:CorrQFIdecreasing}
    The DQFI~$\QFItr$ is non-increasing under any
    trace-preserving completely positive (TPCP) map applied
    to each POVM element.
\end{corollary}
\begin{proof}
    This follows from the Bures metric
    on positive operators being
    non-increasing under any TPCP map~\cite{NC10,Wilde13}.
    Specifically, the distance element~$\mathcal{D}_\mathrm{Bures}$
    in Eq.~\eqref{eq:distancelement}
    is non-increasing under
    any TPCP map applied to both
    ${\pi_j}_\theta$ and~${\pi_j}_{\theta + \delta \theta}$,
    because such maps cannot enhance
    the distinguishability of~${\pi_j}_\theta$
    from~${\pi_j}_{\theta + \delta \theta}$.
    As the distance element~$\mathcal{D}_\mathrm{Bures}$
    is non-increasing for each POVM element~${\pi_j}_\theta$
    ($j \in [m]$),
    the total distance~$\mathcal{D}_\mathrm{total}$ in Eq.~\eqref{eq:totaldistance}
    is also non-increasing,
    as is the trace DQFI~$\QFItr$.
\end{proof}

\section{Convexity of DQFI}
\label{supp:Convex}

In this section,
we prove that the DQFI~$\QFItr$ is convex
in its argument.
So, for two~$m$-outcome POVMs~$\Pi_\theta\equiv\{{\pi_j}_\theta\}_{j\in[m]}$
and~$\Pi'_\theta\equiv\{{\pi'_j}_\theta\}_{j\in[m]}$,
and for~$0\leq p \leq 1$,
\begin{equation}
    \QFItr\left ( (1-p) \Pi_\theta + p \Pi'_\theta \right ) \leq (1-p) \QFItr \left ( \Pi_\theta \right ) + p \QFItr \left ( \Pi'_\theta \right ) \, .
\end{equation}

\begin{proof}
Consider the POVM~$\tilde{\Pi}_\theta \equiv \{ {{}\tilde{\pi}_j}_\theta \}_{j\in[m]}$ defined as
\begin{equation}
    \tilde{\Pi}_\theta \coloneqq (1-p) \Pi_\theta \oplus p \Pi'_\theta = \begin{bmatrix}
        (1-p)  \Pi_\theta & 0 \\ 0 & p \Pi'_\theta
    \end{bmatrix} \, .
\end{equation}
Clearly, $\tilde{\Pi}_\theta$ is a valid~$m$-outcome
POVM acting on a $2d$ dimensional Hilbert space~$\mathcal{H}_{2d}$,
which can be decomposed as~$\mathcal{H} \oplus \mathcal{H} = \mathbb{C}^2\otimes \mathcal{H}$.
Physically,~$\tilde{\Pi}_\theta$
represents a controlled measurement
on~$\mathcal{H}_d$, controlled
by a classical bit in
the state~0 with probability~$1-p$,
and the state~1 with probability~$p$
(see Fig.~\ref{fig:convexPOVM}).

\begin{figure}[hbtp]
    \centering
    \includegraphics[width=0.32\columnwidth]{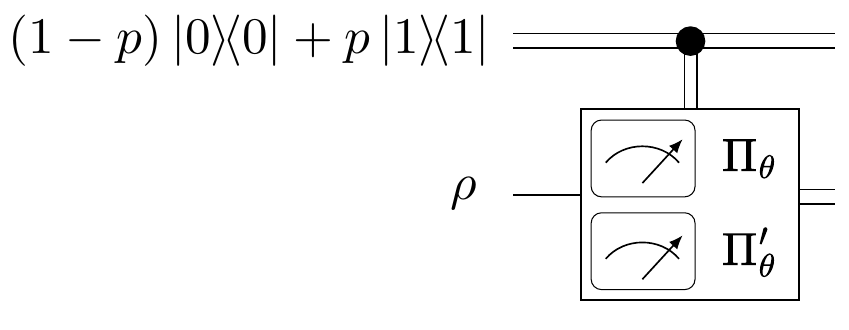}
    \caption{Physical interpretation of composite
    POVM~$\tilde{\Pi}_\theta$.}
    \label{fig:convexPOVM}
\end{figure}

The block-diagonal structure of~$\tilde{\Pi}_\theta$
implies that its detector SLD operators~$\tilde{L}_j$
are also block-diagonal,
i.e.,~$\tilde{L}_j = L_j \oplus L'_j$,
where~$L_j$ \&~$L'_j$ are the detector SLD operators of~$\Pi_\theta$
and~$\Pi'_\theta$, respectively.
The DQFI~$\QFItr$ is thus additive under direct sum,
\begin{equation}
    \QFItr \left [ \tilde{\Pi}_\theta \right ] = (1-p) \QFItr \left [ \Pi_\theta \right ] + p \QFItr \left [ \Pi'_\theta \right ]  \, .
\end{equation}
On the other hand,
if we ignore the first sub-system
and compute the DQFI
of only the second subsystem in Fig.~\ref{fig:convexPOVM},
which corresponds to partial-tracing
out the first sub-system,
the effective measurement becomes~$\Tr_{\mathbb{C}^2} \tilde{\Pi}_\theta = (1-p) \Pi_\theta + p \Pi'_\theta$.
Partial trace being a TPCP map
cannot increase DQFI~$\QFItr$ due
to Corollary~\ref{corr:CorrQFIdecreasing},
thus proving
\begin{equation}
    \QFItr \left [ (1-p) \Pi_\theta + p \Pi'_\theta \right ] = \QFItr \left [ \Tr_{\mathbb{C}^2} \tilde{\Pi}_\theta \right ]
    \leq \QFItr \left [ \tilde{\Pi}_\theta \right ] = (1-p) \QFItr \left [ \Pi_\theta \right ] + p \QFItr \left [ \Pi'_\theta \right ]  \, .
\end{equation}
\end{proof}

As the set of POVMs
is convex and compact,
the convexity of the DQFI~$\QFItr$
implies that extremal POVMs,
namely rank-one projection-valued measurements
maximise the DQFI~$\QFItr$.
Such measurements,
for example projective SIC POVMs
in certain dimensions,
are therefore the most
parameter-informative
and, thus, easiest to estimate parameters from.

\section{Multi-parameter detector estimation}
\label{supp:MultiPara}

Let~$\Pi_\theta\equiv\{{\pi_j}_\theta\}_{j\in[m]}$
be an~$m$-outcome POVM parametrised by~$n$
parameters~$\theta \equiv (\theta_1, \dots, \theta_n)$.
Say we probe this detector
with a single state~$\rho_\mathrm{in}$.
The CFI matrix (CFIM)~$\CFI_\theta[\rho_\mathrm{in}, \Pi_\theta] \in \mathbb{R}^{n\times  n}$
now
has elements
\begin{equation}
\label{eq:multiparaCFI}
    {(\CFI_\theta)}_{jk} = \sum_{l\in[m]} \frac1{p_\theta(l)} \frac{\partial p_\theta(l)}{\partial \theta_j} \frac{\partial p_\theta(l)}{\partial \theta_k} \, ,
\end{equation}
where~$p_\theta(l) = \Tr \left (\rho_\mathrm{in} \, {\pi_l}_\theta \right )$.
Now, define the multi-parameter
detector SLD operators~$\{L_j^{\theta_k}\}_{j\in[m], k\in[n]}$
for the~$j^\text{th}$ measurement outcome and the~$k^\text{th}$ parameter
to be
\begin{equation}
     L_j^{\theta_k} {\pi_j}_\theta + {\pi_j}_\theta L_j^{\theta_k} = 2 \partial_{\theta_k} {\pi_j}_\theta \, .
\end{equation}
Following steps similar to
main-text Eq.~\eqref{eq:shortproofEq1},
we can rewrite the CFIM elements as
\begin{equation}
\label{eq:CFImultidef}
    (\CFI_\theta)_{jk} = \sum_{l\in[m]} \frac{\Re \Tr \left ( \rho_\mathrm{in} {\pi_l}_\theta L_l^{\theta_j} \right ) \Re \Tr \left ( \rho_\mathrm{in} {\pi_l}_\theta L_l^{\theta_k} \right )}{\Tr \left ( \rho_\mathrm{in} {\pi_l}_\theta \right )} \, .
\end{equation}
We first introduce
an operator upper bound
to the CFIM
in the following theorem.

\begin{theorem}
\label{th:multipara}
For any probe state~$\rho_\mathrm{in}\in\mathcal{D}(\mathcal{H}_d)$,
the operator~$\tilde{Q}_\theta [\rho_\mathrm{in}, \Pi_\theta] \in \mathbb{R}^{n \times n}$
with elements defined as
\begin{equation}
\label{eq:Qthetadef}
    \left ({\tilde{Q}_\theta}\right)_{jk} \coloneqq \frac{1}{2} \Tr \left [ \sum_{l \in [m]} \left ( L_l^{\theta_j} {\pi_l}_\theta L_l^{\theta_k}  +  L_l^{\theta_k} {\pi_l}_\theta L_l^{\theta_j} \right ) \rho_\mathrm{in} \right ] \,
\end{equation}
satisfies~$z \, \CFI_\theta[\rho_\mathrm{in}, \Pi_\theta]  \, z^\top \leq z\, \tilde{Q}_\theta[\rho_\mathrm{in}, \Pi_\theta] \,  z^\top$
for any~$z \coloneqq [ z_1, \dots, z_n] \in \mathbb{R}^n$,
i.e.,
\begin{equation}
    \CFI_\theta[\rho_\mathrm{in}, \Pi_\theta] \preccurlyeq \tilde{Q}_\theta[\rho_\mathrm{in}, \Pi_\theta] \, .
\end{equation}
\end{theorem}
\begin{proof}
By straightforward computation
(suppressing the functional dependence
of~$\CFI_\theta$ and~$\tilde{Q}_\theta$
on~$\rho_\mathrm{in}$ and~$\Pi_\theta$
for brevity),
\begin{equation}
    z\, \CFI_\theta \,  z^\top  = \sum_{j,k\in[n]} z_j   \left ( \CFI_\theta \right )_{jk} z_k
    = \sum_{l\in[m]} \frac{\Re \Tr \left ( \rho_\mathrm{in} {\pi_l}_\theta \mathcal{L}_l^{(z)} \right ) \Re \Tr \left ( \rho_\mathrm{in} {\pi_l}_\theta \mathcal{L}_l^{(z)} \right )}{\Tr \left ( \rho_\mathrm{in} {\pi_l}_\theta \right )}  \, ,
\end{equation}
where we have defined Hermitian operators~$\mathcal{L}_l^{(z)} \coloneqq \sum_{i\in[n]} z_i L^{\theta_i}_l$.
Extending the approach in Eq.~\eqref{eq:proofEq1}
from Methods,
we get
\begin{equation}
\begin{aligned}
    z\, \CFI_\theta \,  z^\top &= \sum_{l\in[m]} \frac{\left ( \Re \Tr \left ( \rho_\mathrm{in} {\pi_l}_\theta \mathcal{L}_l^{(z)} \right ) \right )^2}{\Tr \left ( \rho_\mathrm{in} {\pi_l}_\theta \right )}
    \leq \sum_{l\in[m]} \left \vert \frac{ \Tr \left ( \rho_\mathrm{in} {\pi_l}_\theta \mathcal{L}_l^{(z)} \right )}{\sqrt{\Tr \left ( \rho_\mathrm{in} {\pi_l}_\theta \right )}} \right \vert^2\\
    &= \sum_{l\in[m]} \left \vert \Tr \left (
    %\underbrace{
    \frac{\sqrt{\rho_\mathrm{in}} \sqrt{{\pi_l}_\theta}}{\sqrt{\Tr(\rho_\mathrm{in} \, {\pi_l}_\theta)}} %}_{\mathcal{A}}
    %\underbrace{
    \sqrt{{\pi_l}_\theta} \mathcal{L}^{(z)}_l \sqrt{\rho_\mathrm{in}} %}_{\mathcal{B}}
    \right ) \right \vert^2\\
    &\leq \sum_{l\in[m]}  \Tr \left (
    %\underbrace{
    \frac{\sqrt{\rho_\mathrm{in}} \sqrt{{\pi_l}_\theta} \sqrt{{\pi_l}_\theta} \sqrt{\rho_\mathrm{in}}}{\Tr(\rho_\mathrm{in} \, {\pi_l}_\theta)} \right )
    \Tr \left (
    %}_{\mathcal{A}}
    %\underbrace{
    \sqrt{{\pi_l}_\theta} \mathcal{L}^{(z)}_l \sqrt{\rho_\mathrm{in}} \sqrt{\rho_\mathrm{in}} \mathcal{L}^{(z)}_l \sqrt{{\pi_l}_\theta}  %}_{\mathcal{B}}
    \right ) \\
    &= \sum_{l\in[m]} \Tr \left (  \mathcal{L}^{(z)}_l {\pi_l}_\theta \mathcal{L}^{(z)}_l \rho_\mathrm{in} \right )
    = \sum_{j,k\in[n]}  z_j   \frac12 \Tr \left (  \sum_{l\in[m]} \left ( L_l^{\theta_j} {\pi_l}_\theta L_l^{\theta_k}
    + L_l^{\theta_k} {\pi_l}_\theta L_l^{\theta_j} \right ) \rho_\mathrm{in} \right ) z_k = z \, \tilde{Q}_\theta \, z^\top \, .
\end{aligned}
\end{equation}
In the second inequality above,
we have used the operator Cauchy-Schwarz,~$\vert \Tr (A^\dagger B) \vert^2 \leq \Tr(A^\dagger A) \Tr(B^\dagger B)$.
\end{proof}

\begin{corollary}[Corollary to Theorem~\ref{th:multipara}]
\label{corr:multipara}
    For an ensemble of probes~$\{q_k, \rho_k\}_{k=1}^p$,
    the effective CFI is upper-bounded by~$\tilde{Q}_\theta$
    of the ensemble average state, i.e.,
    \begin{equation}
        \CFI_\theta\big[\{q_k, \rho_k\}_{k=1}^p, \Pi_\theta\big] \preccurlyeq \tilde{Q}_\theta\big[  \big( \sum_{k=1}^p q_k \rho_k  \big) , \Pi_\theta\big] \, .
    \end{equation}
\end{corollary}
\begin{proof}
    The proof follows directly from
    the convexity of the CFIM and
    the linearity of~$Q_\theta[\rho_\mathrm{in}, \Pi_\theta]$
    (defined in Eq.~\eqref{eq:Qthetadef})
    in its state argument~$\rho_\mathrm{in}$.
    In particular, for each probe state~$\rho_k$,
    it follows from Theorem~\ref{th:multipara} that
    \begin{equation}
        \CFI_\theta [\rho_k, \Pi_\theta] \preccurlyeq \tilde{Q}_\theta [\rho_k, \Pi_\theta] \, ,
    \end{equation}
    so that for the CFIM of the ensemble, we have
    \begin{equation}
        \CFI_\theta\big[\{q_k, \rho_k\}_{k=1}^p, \Pi_\theta\big]
        = \sum_{k=1}^p q_k \, \CFI_\theta [\rho_k, \Pi_\theta]
        \preccurlyeq  \sum_{k=1}^p q_k \tilde{Q}_\theta [\rho_k, \Pi_\theta] = \tilde{Q}_\theta\big[  \big( \sum_{k=1}^p q_k \rho_k  \big) , \Pi_\theta\big] \, ,
    \end{equation}
    where the first equality holds because the probabilities~$\{q_k\}$
    are independent of the parameters a priori~\cite{Pezze18}.
\end{proof}

\subsection{Multi-parameter DQFI Matrix}

Theorem~\ref{th:multipara}
and its corollary prove that
for estimating~$n$ parameters
using any ensemble,
the effective CFIM is
upper-bounded as
$\CFI_\theta \preccurlyeq \tilde{Q}_\theta$,
where~$\tilde{Q}_\theta$
is computed for the ensemble average state.
By simply noting that for any quantum state~$\rho$,
it holds that~$\rho \preccurlyeq \mathds{1}_d$,
we obtain the multi-parameter trace DQFI matrix,
\begin{equation}
    \tilde{Q}_\theta(\rho, \Pi_\theta) \preccurlyeq \tilde{Q}_\theta(\mathds{1}_d, \Pi_\theta) \eqqcolon \QFItr \left [ \Pi_\theta \right ] \, .
\end{equation}

\begin{definition}
Define the multi-parameter
trace DQFI~$\QFItr \in \mathbb{R}^{n \times n}$
of a POVM~$\Pi_\theta\equiv\{{\pi_j}_\theta\}_{j\in[m]}$
with respect to parameters~$\theta \equiv (\theta_1, \dots, \theta_n)$
as
\begin{equation}
    \left ( \QFItr \left [ \Pi_\theta \right ] \right )_{jk} \coloneqq
     \frac{1}{2} \Tr \left [ \sum_{l \in [m]} \left ( L_l^{\theta_j} {\pi_l}_\theta L_l^{\theta_k}  +  L_l^{\theta_k} {\pi_l}_\theta L_l^{\theta_j} \right )\right ]  \, .
\end{equation}
\label{def:QFImulti}
\end{definition}
\noindent
Clearly,~$\QFItr \left [ \Pi_\theta \right ]$
reduces to main-text Def.~\ref{def:QFI1}
%trace detector QFI
%definition,~$\QFItr\coloneqq \sum_{l\in [m]} \Tr({\pi_\theta}_l \,  {L_\theta}_l^2)$
in the single-parameter case.
On the other hand,
if we try to extend
the single-parameter spectral DQFI~$\QFIsp$
(main-text Def.~\ref{def:QFI2})
to multi-parameters by
defining~$\mathcal{Q}_\theta \in\mathbb{R}^{n \times n}$
as
\begin{equation}
        {\mathcal{Q}_\theta}_{jk} \coloneqq \frac{1}{2} \left \Vert \sum_{l \in [m]} \left ( L^{l}_{\theta_j} \pi_l L^{l}_{\theta_k}  +  L^l_{\theta_k} \pi_l L^l_{\theta_j} \right ) \right \Vert_{\mathrm{sp}}  \, ,
\end{equation}
then this implies
the element-wise
inequality~$\tilde{Q}_\theta
\leq_{\mathrm{elem}} \mathcal{Q}_\theta$
but does not guarantee~$ \tilde{Q}_\theta\preccurlyeq \mathcal{Q}_\theta$.

\subsection{Multi-parameter QCRB}

The multi-parameter trace DQFI
from Def.~\ref{def:QFImulti}
inherits useful properties
like convexity and connection to a distance metric
from its single-parameter counterpart in Def.~\ref{def:QFI1}
(we don't prove these mathematically),
but also shares the non-attainability
discussed below Def.~\ref{def:QFI1}.
A tighter or more attainable
precision bound can be constructed
by instead applying the QCRB approach~\cite{Helstrom1967,Helstrom1968,Helstrom1969, Helstrom1974}
directly to the operator~$\tilde{Q}_\theta$.
In this setting,
the lowest attainable
weighted sum of variances and covariances
is given by the CCRB minimised over probing strategies,
\begin{equation}
\label{eq:QCRBdef}
   \mathcal{C}^{\mathrm{CCRB}}_*[\Pi_\theta, W] \coloneqq \min_{\substack{V=V^\top \in \mathbb{R}^{n \times n} \, , \\
    \rho_k = \rho_k^\dagger  \in \mathbb{C}^{d \times d}, \, \Tr(\rho_k) = 1 \, ,\\ \rho_k \succcurlyeq 0 , \,  \sum q_k = 1}} \Tr (W V) \, {\big \vert} \,  V \succcurlyeq  \CFI_\theta\big(\{q_k, \rho_k\}_{k=1}^p \big)^{-1}  \, ,
\end{equation}
where~$W \in \mathbb{R}^{n \times n}$
is the symmetric, positive semi-definite weight matrix
and~$V$ is the covariance matrix
of parameter estimates.
Despite the minimisation
in Eq.~\eqref{eq:QCRBdef}
not being a semi-definite program,
we can leverage Theorem~\ref{th:multipara}
to obtain a lower bound to~$\mathcal{C}^{\mathrm{CCRB}}_*$
that is tighter than
the trace approach.
Note that the weighted trace QCRB
resulting from Def.~\ref{def:QFImulti},
denoted~$\mathcal{C}_{\Tr}^{\mathrm{QCRB}}$,
is
\begin{equation}
\label{eq:TrQCRBWtSupp}
    \mathcal{C}_{\Tr}^{\mathrm{QCRB}}[\Pi_\theta, W] \coloneqq \Tr\left ( W \mathcal{J}_{\Tr,\theta}^{-1} \right ) \, .
\end{equation}

From Theorem~\ref{th:multipara}
and Corollary~\ref{corr:multipara},
we have~$\CFI_\theta\big[\{q_k, \rho_k\}_{k=1}^p, \Pi_\theta\big] \preccurlyeq \tilde{Q}_\theta\big[  \big( \sum_{k=1}^p q_k \rho_k  \big) , \Pi_\theta\big]$,
leading to
\begin{equation}
    \tilde{Q}_\theta\big[  \big( \sum_{k=1}^p q_k \rho_k  \big) , \Pi_\theta\big]^{-1} \preccurlyeq \CFI_\theta\big[\{q_k, \rho_k\}_{k=1}^p, \Pi_\theta\big]^{-1}  \, .
\end{equation}
We can therefore define the weighted
spectral QCRB as
\begin{equation}
\label{eq:MultiQCRBdef}
    \mathcal{C}^{\mathrm{QCRB}}_\Vert[\Pi_\theta, W] \coloneqq \min_{\substack{V=V^\top \in \mathbb{R}^{n \times n} \, , \\
    \rho = \rho^\dagger  \in \mathbb{C}^{d \times d}, \, \Tr(\rho) = 1 \, ,\\ \rho \succcurlyeq 0}} \Tr (W V) \, \vert \,  V \succcurlyeq \left ( \tilde{Q}_\theta[\rho, \Pi_\theta] \right )^{-1}  \, .
\end{equation}
Clearly, any candidate~$V$ that is feasible for
the minimisation in Eq.~\eqref{eq:QCRBdef}
is also feasible for the minimisation
in Eq.~\eqref{eq:MultiQCRBdef},
because~$V \succcurlyeq  \CFI_\theta^{-1}
\succcurlyeq (\tilde{Q}_\theta)^{-1}$.
This proves that
\begin{equation}
     \mathcal{C}^{\mathrm{QCRB}}_\Vert[\Pi_\theta, W] \leq \mathcal{C}^{\mathrm{CCRB}}_*[\Pi_\theta, W]
\end{equation}
and that Eq.~\eqref{eq:MultiQCRBdef}
presents a valid lower bound
to the minimum attainable uncertainties of estimates.
Finally,
due to the linearity of~$\tilde{Q}_\theta$,
the minimisation
in Eq.~\eqref{eq:MultiQCRBdef}
can be recast into the semi-definite program,
\begin{equation}
\label{eq:wtspectralSDPSupp}
    \begin{array}{cl}
        \mathrm{minimise} \quad \quad &\Tr(W V) \\
        {\substack{V=V^\top \in \mathbb{R}^{n \times n} \, , \\
    \rho = \rho^\dagger  \in \mathbb{C}^{d \times d}, \, \Tr(\rho) = 1 \, ,\\ \rho \succcurlyeq 0}}\\
    \mathrm{subject\; to} \quad \quad &\begin{pmatrix}
        V & \mathds{1}_n \\
        \mathds{1}_n & \tilde{Q}_\theta[\rho, \Pi_\theta]
    \end{pmatrix} \succcurlyeq 0 \, \,  ,
    \end{array}
\end{equation}
which can be solved efficiently
using standard numerical solvers
like YALMIP or CVX~\cite{Lorcan21}.
In the single-parameter case,
the minimisation
in Eq.~\eqref{eq:MultiQCRBdef} reduces to
solving
\begin{equation}
    \min_{\rho} \frac{1}{\tilde{Q}_\theta[\rho, \Pi_\theta]} =
    \min_{\rho} \frac{1}{\Tr \left [ \sum_{l \in [m]} {L_l}_\theta {\pi_l}_\theta {L_l}_\theta  \rho \right ]} \,
    = \frac{1}{\max_{\rho} \Tr \big [ Q_\theta  \rho \big ]}
    = \frac1{\QFIsp\left [ \Pi_\theta\right ]} \, ,
\end{equation}
as expected.
Therefore, the multi-parameter QCRB~$\mathcal{C}^{\mathrm{QCRB}}_\Vert$
constitutes an extension of
the spectral DQFI~$\QFIsp$
to multiple parameters.

\begin{example}
\label{eg:bitphaseflip}
    Consider an imperfect projective measurement along
    polar angle~$\theta = \pi/4$
    and azimuthal angle~$\phi=0$ that is subject
    to independent bit-flip errors (rate~$p_1$)
    and phase-flip errors (rate~$p_2$).
    The noiseless measurement at~$p_1=p_2=0$
    is a projective~$(X+Z)/\sqrt{2}$ measurement.
    The POVM~$\Pi_p \equiv \{{\pi_1}_p, {\pi_2}_p\}$ parametrised
    by~$p \equiv (p_1, p_2)$ corresponding to this measurement is
    \begin{equation}
    \begin{split}
        {\pi_1}_p &= \frac{1}{2\sqrt{2}} \begin{pmatrix}
              \sqrt{2} + 1 - 2  p_1  & 1-2 p_2 \\
              1-2 p_2 & \sqrt{2} - 1 + 2 p_1
        \end{pmatrix} \\
        \text{and} \quad  {\pi_2}_p &= \frac{1}{2\sqrt{2}}  \begin{pmatrix}
              \sqrt{2} - 1 + 2 p_1  & -1+2 p_2 \\
              -1+2 p_2 &  \sqrt{2} + 1 - 2  p_1
        \end{pmatrix} \, .
    \end{split}
    \end{equation}
    The trace QCRB and the spectral QCRB for this problem are
    \begin{equation}
        \mathcal{C}^{\mathrm{QCRB}}_{\Tr} = \frac14 + \frac{(1-p_1)p_1}{2} + \frac{(1-p_2)(p_2)}{2} \quad \text{and} \quad \mathcal{C}^{\mathrm{QCRB}}_\Vert = \frac12 + (1-p_1)p_1 + (1-p_2)(p_2) = 2 \mathcal{C}^{\mathrm{QCRB}}_{\Tr} \,
    \end{equation}
    and are shown in main-text Fig.~\ref{fig:tightQFIexamples}(d)
    as blue and green surfaces, respectively.
    Main-text Fig.~\ref{fig:tightQFIexamples}(d) also depicts
    the true minimum of total MSE (golden surface),
    i.e., the tight bound,
    while demonstrating the
    trace and spectral QCRB to be unattainable
    (except at the extreme values of~$p_1$ and~$p_2$).
    The Gill-Massar QCRB~$\Tr\left (\sqrt{\mathcal{J}_{\Tr, p}^{-1}}\right)^2$
    for this problem is also unattainable,
    but the tight bound equals twice the Gill-Massar QCRB,
    i.e.,
    \begin{equation}
        \mathcal{C}^\mathrm{CCRB}_* = 2 \Tr\left (\sqrt{\mathcal{J}_{\Tr, p}^{-1}}\right)^2 \, .
    \end{equation}
    The simple strategy of probing with an
    (optimally weighed) ensemble of~$\ket{0}$
    and~$\ket{1}$ (grey surface) is highly effective,
    though sub-optimal.
\end{example}

\section{Comparison with Total QFI and Probe Incompatibility Effect}
\label{sec:SupMatMultiComparison}

In this section, we compare our multi-parameter QCRBs
to existing techniques
for multi-parameter channel estimation~\cite{Albarelli2022}.
The comparison reveals that the
detector-based approach is often more informative
than generic channel techniques, but also highlights
features of multi-parameter detector estimation
that require further exploration to be
fully understood.

In Example~\ref{eg:example2multi} of the main text,
the optimal probe states~($\ket{0},\ket{1}$)
for different parameters~($\theta_1,\theta_2$) were orthogonal,
so that probing with~$\ket{0}$
extracted no~$\theta_2$-information
and with~$\ket{1}$
extracted no~$\theta_1$-information.
More generally, the optimal probes for
different parameters~$\theta_j$
and~$\theta_k$ could be different,
leading to an incompatibility---termed probe
incompatibility~\cite{Albarelli2022}---in the multi-parameter case.
This suggests a straight-forward
but generally sub-optimal estimation strategy:
probing with an ensemble,~$\{p_j, \rho_j^*\}$,
of the single-parameter-optimal probes~$\rho_j^*$
with mixing probabilities~$\{p_j\}$.
More simply, we may split the total number
of detector uses into~$n$ and address the~$n$ single-parameter
problems separately, estimating only~$\theta_i$
from a fraction~$p_i$ of the samples
using optimal probe~$\rho_i^*$.
Let us call this the sequential scheme,
shortened to~$\mathrm{seq}$ below.
For this scheme,~$\mathrm{MSE}^{\mathrm{seq}}[\theta_j] \geq 1/(p_j \mathcal{J}_{j})$,
where~$\mathcal{J}_{j}$ denotes
the single-parameter DQFI for~$\theta_j$,
so the total MSE is bounded from below
by
\begin{equation}
    \sum_{j\in[n]} \mathrm{MSE}^\mathrm{seq}[\theta_j] \geq \sum_{j\in[n]} 1/(p_j \mathcal{J}_{j}) \, .
\end{equation}
If the QFI used is tight for
all single-parameter problems,
this lower bound is attainable.
Naturally, this lower bound
upper-bounds the total MSE of
the optimal simultaneous strategy,
\begin{equation}
\label{eq:simultsequentialcomp}
    \sum_{j\in[n]} \frac{1}{p_j \mathcal{J}_j} \geq
    \min \sum_{j\in [n]} \mathrm{MSE}[\theta_j]  \, .
\end{equation}
Therefore,
a comparison of the minimum
sequentially-achievable MSE
with the minimum simultaneously-achievable MSE
reveals the extent of probe incompatibility
in the problem~\cite{Albarelli2022}.

Interestingly, a generalisation of the quantity~$\sum_j p_j \mathcal{J}_j$ from probabilities~$p_j$
to positive weights~$w_j>0$ (not necessarily
normalised to~1) leads to the total QFI~$\mathcal{J}^\mathrm{tot} \coloneqq \sum_j w_j \mathcal{J}_j$~\cite{Albarelli2022},
which can yield a lower bound to the minimum
simultaneously-achievable MSE in
Eq.~\eqref{eq:simultsequentialcomp}. In Ref.~\cite{Albarelli2022},
the lower bound~$n^2/\mathcal{J}^\mathrm{tot}$
was introduced for $n$-parameter channel estimation
with an SDP solution (Appendix~F in~\cite{Albarelli2022}).
By adopting the channel representation of measurements
(discussed in Supplemental Material~\ref{sec:tightboundsdp}),
this bound, which we refer to as the total QFI QCRB,
may be applied to multi-parameter detector estimation.
For simplicity, here we restrict our comparisons to
total MSE, for which~$w_j=1$.
We compare the total QFI QCRB to our trace
and spectral detector QCRBs for 10,000 randomly generated
two-parameter, two-outcome, qubit detector models in Figs.~\ref{fig:multiparacomb}
\&~\ref{fig:multiparacomp}.
The comparison in Fig.~\ref{fig:multiparacomp} reveals
the trace and spectral QCRBs to often
be tighter than the channel-based total QFI QCRB~\cite{Albarelli2022},
though no absolute hierarchy exists.
The fact that the trace QCRB can outperform
the total QFI QCRB is also surprising,
especially given that the former is
analytically-solved in closed-form
whereas the latter requires numerical optimisation.

\begin{figure}[htb]
    \centering
    \includegraphics[width=0.92\linewidth]{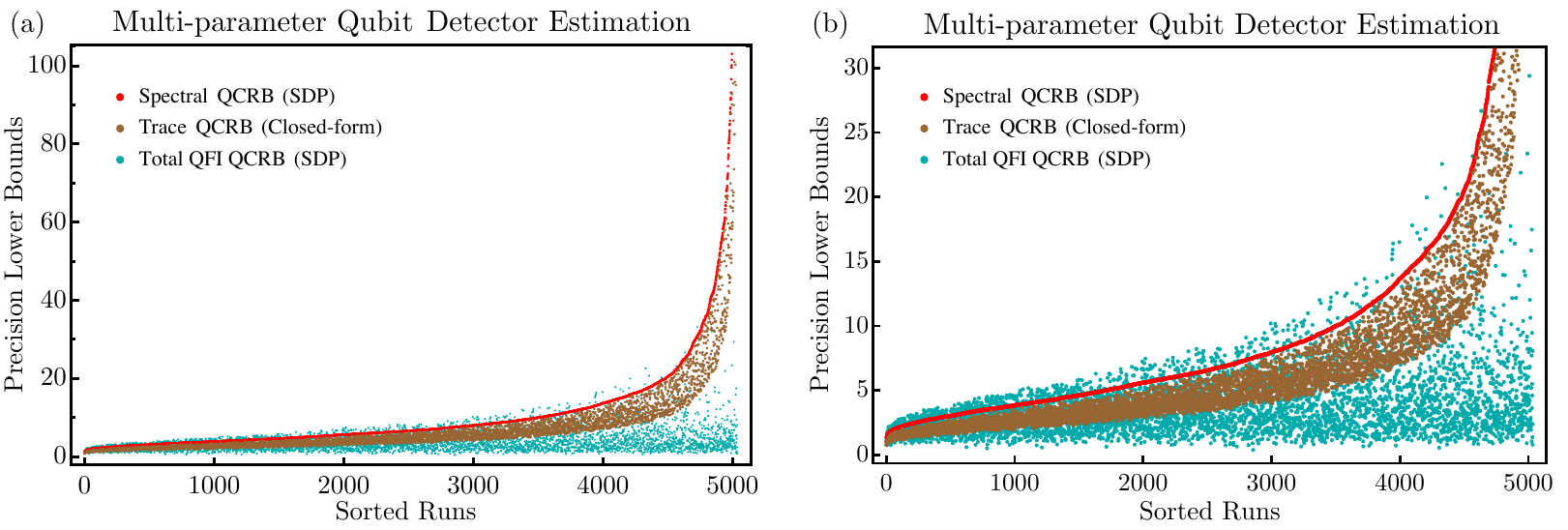}
    \caption{Comparison of
    the total QFI QCRB and the DQFI QCRBs
    for two-parameter, on-off,
    qubit detector estimation.
    (a) For two-parameter estimation across~10,000
    randomly-generated qubit measurements,
    we find that
    the spectral QCRB (red) is typically tighter than
    the total QFI QCRB (light blue), but not always.
    (b) Zoomed-in version of (a).
    The scatter points are sorted in the increasing order of
    the spectral QCRB.}
    \label{fig:multiparacomb}
\end{figure}
\begin{figure}[htb]
    \centering
    \includegraphics[width=0.92\linewidth]{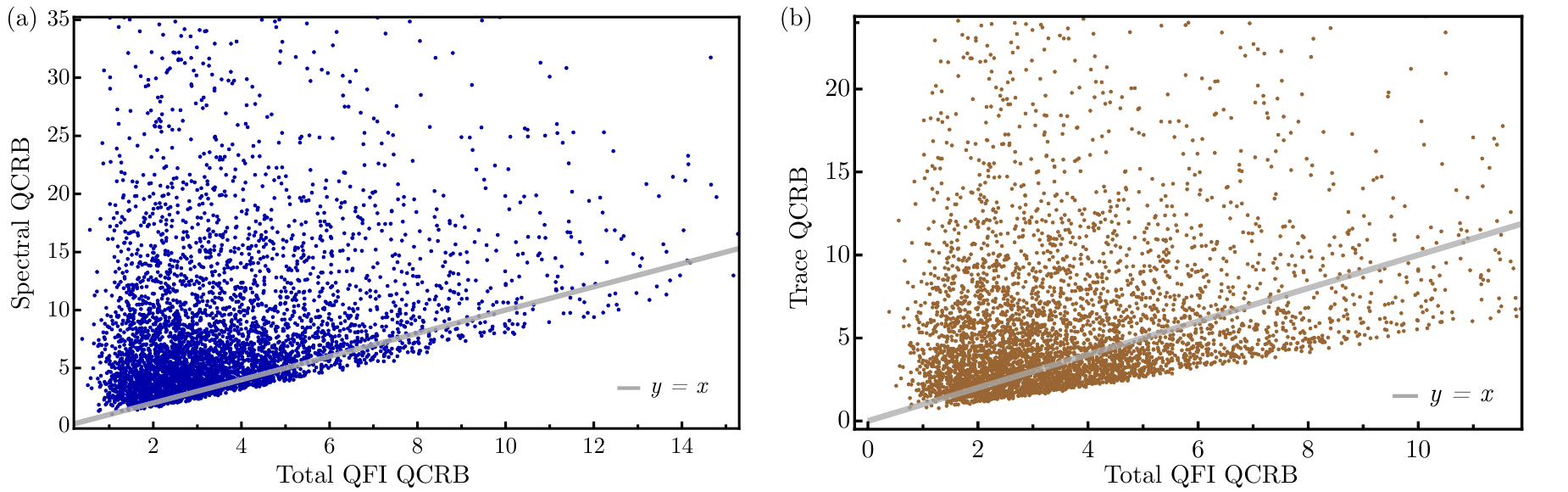}
    \caption{One-to-one comparison of
    the spectral QCRB (a)
    and the trace QCRB (b)
    versus the total QFI QCRB
    for two-parameter estimation from~10,000
    random, two-outcome, qubit detectors.
    (a) The spectral QCRB is often tighter
    than the total QFI QCRB, evidenced by a majority
    of the points lying above the~$y=x$ line.
    (b) The trace QCRB (closed-form) can also be
    more informative than
    the total QFI QCRB (SDP).}
    \label{fig:multiparacomp}
\end{figure}

The comparison with the total QFI QCRB
also sheds light on quantum aspects of detector estimation,
like probe incompatibility and measurement incompatibility.
These effects, which are well-studied in state~\cite{Candeloro2024}
and channel literature~\cite{Albarelli2022}, need to be
fully understood to assess the
extent and practicality of quantum
enhancement in detector estimation.
Notably, in multi-parameter detector estimation,
incompatibility may arise from a number of sources,
not only probe incompatibility. For instance,
the embedding of parameters across different measurement
outcomes can itself be a source of incompatibility,
even if the optimal probes are the same;
this is explored in the following
example.

\begin{example}
\label{eg:ProbeIncompat}
    Consider the following four-outcome qubit measurement
    parametrised by~$(\theta_1, \theta_2)$ (for~$0 \leq \theta_i \leq \nicefrac{1}{2}$):
    \begin{equation}
    \pi_1 = \begin{pmatrix} \theta_1 & 0 \\ 0 & \frac12-\theta_1\end{pmatrix} \, , \, \,
    \pi_2 = \begin{pmatrix} \frac12-\theta_1 & 0 \\ 0 & \theta_1\end{pmatrix} \, , \, \,
    \pi_3 = \begin{pmatrix} \theta_2 & 0 \\ 0 & \theta_2\end{pmatrix} \, , \, \,
    \pi_4 = \begin{pmatrix} \frac12-\theta_2 & 0 \\ 0 & \frac12-\theta_2\end{pmatrix} \, .
    \end{equation}
    The spectral QCRB for this problem
    equals~$\mathcal{C}^{\mathrm{QCRB}}_\Vert = \sum_j \theta_j(1-2\theta_j) = 2 \, \mathcal{C}^{\mathrm{QCRB}}_{\Tr}$.
    The spectral bound is tight and is attained
    by optimal probes~$\ket{0}$ or~$\ket{1}$.
    However, note that~$\pi_1$ and~$\pi_2$ ($\pi_3$ and~$\pi_4$) depend solely on~$\theta_1$ ($\theta_2$),
    and~$\pi_1+\pi_2 = \pi_3 + \pi_4 = \nicefrac{1}{2} \mathds{1}_2$. This means that
    regardless of the input probe,
    outcomes~1 \&~2 click with probability half,
    as do outcomes~3 \&~4.
    Therefore, any sequential estimation strategy
    that uses a fixed fraction of samples to
    estimate~$\theta_1$ and~$\theta_2$ separately
    will be suboptimal, because it wastes half the number
    of samples. On the other hand,
    consider
    the two single-parameter problems
    obtained by fixing one of the parameters:
    they
    are both equivalent to main-text Example~\ref{eg:bitflipdet}
    and share the same
    single-parameter optimal probes,
    either~$\ket{0}$ or~$\ket{1}$, for both parameters.
    Therefore, there is no probe incompatibility in this problem.
    As expected, the channel-based SDP, which accounts for
    probe incompatibility, is generally
    less tight than the spectral QCRB for this example,
    as shown in Fig.~\ref{fig:multiparacombExample} below.
\end{example}

In fact, Fig.~\ref{fig:multiparacombExample}
reveals two interesting features worth further investigation.
First, in Fig.~\ref{fig:multiparacombExample}(a),
the total QFI QCRB can be smaller than the trace QCRB,
meaning it is more than a factor of 2 away from the
tight bound. In contrast, in Ref.~\cite{Albarelli2022},
it was reported that
measurement incompatibility in channel estimation
can at most double the attainable MSE
predicted by probe incompatibility alone.
This prompts a
deeper analysis of various incompatibility sources
in detector estimation.
Second, Fig.~\ref{fig:multiparacombExample}(b) seems to
imply that the total QFI QCRB can be zero, which typically
indicates that
estimation is not feasible, whereas the simultaneous optimal
strategy can still estimate both parameters.
This is understood as follows:
whenever one of the parameters (say~$\theta_j$)
is close (but not equal) to
either~0 or~1/2,
the corresponding single-parameter DQFI~$\QFI_{\Vert,\,j} = 1/(\theta_j(1-2\theta_j))$
can be arbitrarily large,
making the total QFI~$\mathcal{J}^\mathrm{tot} = \QFI_{\Vert,\,1}
+\QFI_{\Vert,\,2}$ arbitrarily large
and the precision bound~$4/\mathcal{J}^\mathrm{tot}$
arbitrarily small. However, for simultaneous
multi-parameter estimation,
the bound~$1/\QFI_{\Vert,\,1}+1/\QFI_{\Vert,\,2}$
is more relevant
and tighter (Eq.~(6)~in Ref.~\cite{Albarelli2022}), and this bound
can be non-vanishing even if one of the parameters
is close to~0 or~1/2 as long as the other
parameter is not.

The above observations highlight
the need for further exploration of
the fundamental aspects of multi-parameter detector estimation.
In this work, however, we do not go into
further depth
on the various incompatibility effects in
detector estimation.
Nonetheless, we believe that this topic
should be studied in future work,
as it furthers our understanding of the merits
and limitations of simultaneous
estimation for multi-parameter detector models.

\begin{figure}[htb]
    \centering
    \includegraphics[width=0.91\linewidth]{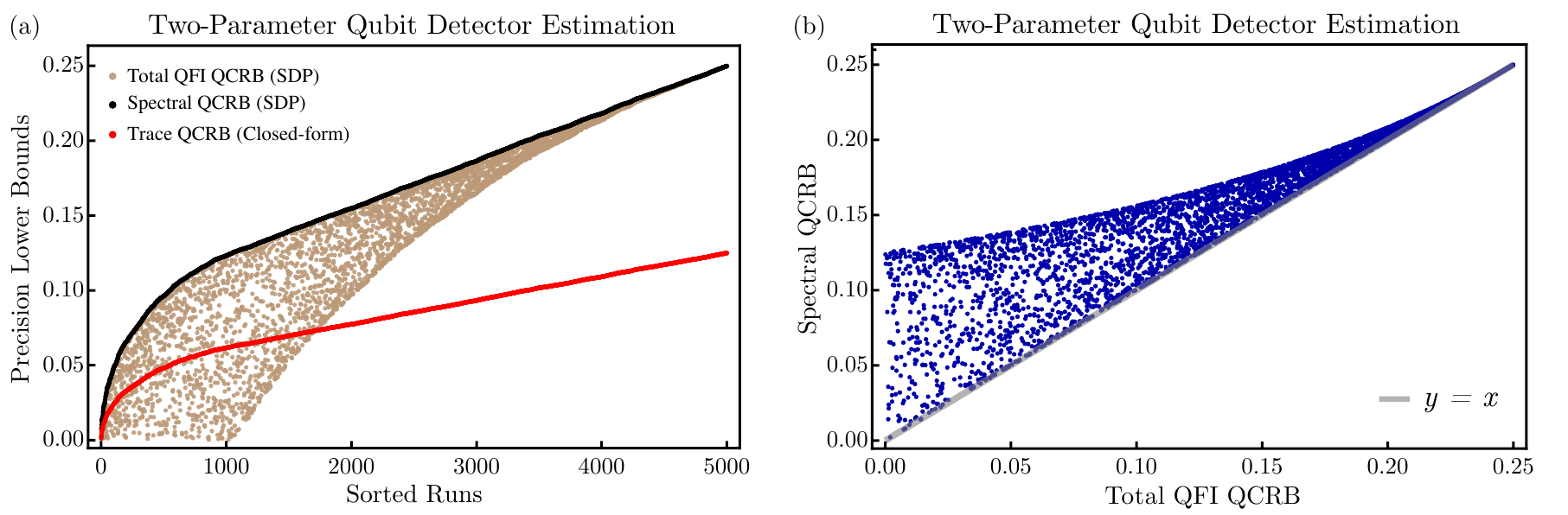}
    \caption{Comparison of
    the total QFI QCRB and the DQFI QCRBs
    for a two-parameter detector model
    without probe incompatibility
    (Example~\ref{eg:ProbeIncompat}).
    (a)~For estimating two parameters~$0 < \theta_j < 1/2$,
    we find that the spectral QCRB (black) is always tight
    and equals twice the trace QCRB (red). The total QFI QCRB
    (light brown) is generally less tight
    than the spectral QCRB, and sometimes less tight than the trace QCRB as well.
    (b)~A direct comparison of the spectral QCRB versus
    the total QFI QCRB shows the former to
    be generally tighter (all the points lie
    above the~$y=x$ line), whereas
    the latter could be close to zero
    even if simultaneous estimation is feasible.
    In~(a), the scatter points are sorted
    in increasing order of the spectral QCRB.}
    \label{fig:multiparacombExample}
\end{figure}

\section{Non-additivity of DQFI \& scaling with copies}
\label{supp:nonadditive}

Suppose we have two identical copies
of the same detector.
In Eq.~\eqref{eq:additiveCFI} in Methods,
we showed that the CFI of detector outcomes
is additive under tensoring,
from which it follows
\begin{equation}
    \CFI_\theta[ \Pi \otimes \Pi \vert \rho_1 \otimes \rho_2] = \CFI_\theta[\Pi \vert \rho_1 ] + \CFI_\theta[\Pi \vert \rho_2 ] \, ,
\end{equation}
implying~$\max_{\rho_1, \,  \rho_2} \CFI_\theta [\Pi^{\otimes 2} \vert \rho_1 \otimes \rho_2]= 2 {\CFI_\theta}_\mathrm{max}[\Pi]$.
However,~${\CFI_\theta}_\mathrm{max}[\Pi^{\otimes 2}]$
is not additive and is greater
than~$2 {\CFI_\theta}_\mathrm{max}[\Pi]$
in general.
And, to attain the maximum two-copy
CFI, the optimal probe state
is in general
a bipartite entangled state.
Importantly,
it is not necessary to have two identical
copies of the detector---the same detector
can be used twice by time-delaying one arm
of the entangled state.

For~$m$-outcome measurements,
the two-copy measurement
can be written as an~$m^2$-outcome measurement
\begin{equation}
     \Pi^{\otimes 2} \equiv \left \{ \pi_j \otimes \pi_k \right \}_{j,k\in[m]} \, ,
\end{equation}
and, clearly,
the SLD operators for~$\pi_j\otimes \pi_k$
are~$\tilde{L}_{jk} = L_j \otimes \mathds{1}_d + \mathds{1}_d \otimes L_k$.
The two-copy DQFIs~${\mathcal{J}_{\mathrm{Tr},\theta}^{(2)}}$ \&~${\mathcal{J}_{\vert \vert, \theta}^{(2)}}$
are computed by
first
evaluating~$\sum_{j,k} \tilde{L}_{jk} \left ( \pi_j \otimes \pi_k \right ) \tilde{L}_{jk}$.
By explicit calculation,
\begin{equation}
\label{eq:Qmatrix2copy}
    Q^{(2)} \coloneqq \sum_{j,k} \tilde{L}_{jk} \left ( \pi_j \otimes \pi_k \right ) \tilde{L}_{jk}
    = Q \otimes \mathds{1}_d + \mathds{1}_d \otimes Q + A \otimes A^\dagger + A^\dagger \otimes A   \, ,
\end{equation}
where~$Q = \sum_{j} L_j \pi_j L_j$
and~$A = \sum_{j} L_j \pi_j$.
The two DQFIs,~${\mathcal{J}_{\mathrm{Tr},\theta}^{(2)}}$ \&~${\mathcal{J}_{\vert \vert, \theta}^{(2)}}$,
for the two-copy measurement
are then~$\Tr[Q^{(2)}]$
and~$\Vert Q^{(2)} \Vert_\mathrm{sp}^2$,
respectively.
By straight-forward
computation,
we find
\begin{equation}
\label{eq:twocopytrQFI}
    {\mathcal{J}_{\mathrm{Tr},\theta}^{(2)}} = 2 d \, \Tr[Q] =2 d \, \QFItr
\end{equation}
because~$\Tr[A \otimes A^\dagger + A^\dagger \otimes A] = 2 \vert \Tr[A]\vert^2$
whereas~$\Tr[A] = \sum_j \Tr[\partial_\theta \pi_j]= \partial_\theta \Tr[\sum_j \pi_j] = 0$.
As~$d\geq 2$, it is clear that
the trace DQFI
satisfies~${\mathcal{J}_{\mathrm{Tr},\theta}^{(2)}} \geq 2 \QFItr$
and is thus not additive.
The same relationship holds for
the spectral DQFI,
as we prove in the lemma below.
Thus the spectral DQFI is not additive either,
representing the superiority
of an entangled probing process.

\begin{lemma}
\label{lemm:twoQFInonadd}
For the two-copy spectral DQFI~${\mathcal{J}_{\vert \vert, \theta}^{(2)}}$
defined above,
\begin{equation}
    {\mathcal{J}_{\vert \vert, \theta}^{(2)}} \geq 2 \, \QFIsp  \, .
\end{equation}
\end{lemma}
\begin{proof}
    Call a normalised eigenvector of~$Q$
    corresponding to its
    largest eigenvalue~$\ket{\psi}$.
    For this pure state,~$\bra{\psi} Q \ket{\psi} = \QFIsp$ by definition.
    Now, by computing~$\bra{\psi}^{\otimes 2} Q^{(2)} \ket{\psi}^{\otimes 2}$,
    we find
    \begin{equation}
        \begin{split}
        \bra{\psi}^{\otimes 2} Q^{(2)} \ket{\psi}^{\otimes 2} &= 2 \QFIsp + \bra{\psi}^{\otimes 2} (A \otimes A^\dagger + A^\dagger \otimes A) \ket{\psi}^{\otimes 2}\\
        &= 2 \QFIsp + 2 \bra{\psi} A \ket{\psi} \bra{\psi} A^\dagger \ket{\psi}\\
        &= 2 \QFIsp + 2 \vert \bra{\psi} A \ket{\psi}\vert^2 \geq 2 \QFIsp   \, .
        \end{split}
    \end{equation}
    And, by definition,
    \begin{equation}
        {\mathcal{J}_{\vert \vert, \theta}^{(2)}} = \max_{\substack{\ket{\phi'}\in\mathbb{C}^4,\\
        \vert\braket{\phi'}\vert^2=1}}
        \bra{\phi'} Q^{(2)} \ket{\phi'} \geq \bra{\psi}^{\otimes 2} Q^{(2)} \ket{\psi}^{\otimes 2} \, ,
    \end{equation}
    thus proving~${\mathcal{J}_{\vert \vert, \theta}^{(2)}} \geq 2 \QFIsp$
    as required.
\end{proof}

We also provide an upper
bound to the non-additivity
of the spectral DQFI
through the following
sequence of inequalities:
\begin{equation}
\label{eq:twocopyspecQFI}
\begin{split}
    {\mathcal{J}_{\vert \vert, \theta}^{(2)}} &= \max \mathrm{eig} \left [ Q \otimes \mathds{1}_d + \mathds{1}_d \otimes Q + A \otimes A^\dagger + A^\dagger \otimes A \right ]  \\
    & \leq  \max \mathrm{eig} \left [ Q \otimes \mathds{1}_d + \mathds{1}_d \otimes Q \right ] + \max \mathrm{eig} \left [  A \otimes A^\dagger + A^\dagger \otimes A \right ] \\
    & \leq 2 \max \mathrm{eig} \left [ Q \right ] + 2 \max \mathrm{eig} \left [ A^\dagger A \right ] \\
    &= 2 \,  \QFIsp + 2 \max \mathrm{eig} [ \, \sum_{j,k} \pi_k L_k L_j \pi_j  ] \, .
\end{split}
\end{equation}
We can thus sandwich
the two-copy spectral DQFI as
\begin{equation}
\label{eq:QFIspsandwich}
     2  \, \QFIsp   \leq {\mathcal{J}_{\vert \vert, \theta}^{(2)}} \leq 2 \,  \QFIsp + 2 \max \mathrm{eig}  [ \, \sum_{j,k} \pi_k L_k L_j \pi_j  ] \, .
\end{equation}
Extending main-text Eq.~\eqref{eq:ordering},
we find the ordering relation
between the two-copy DQFIs
to be
%holds for the two-copy
%QFIs~$\QFItr^{(2)}$ and~$\QFIsp^{(2)}$
%as well,
%simply from definition:
%specifically,
\begin{equation}
\label{eq:orderingtwocopy}
    \frac{1}{d^2} \, {\mathcal{J}_{\mathrm{Tr},\theta}^{(2)}} =
    \frac{2}{d} \, \QFItr
    \leq {\mathcal{J}_{\vert \vert, \theta}^{(2)}} \leq 2 d \QFItr = {\mathcal{J}_{\mathrm{Tr},\theta}^{(2)}} \leq d^2 {\mathcal{J}_{\vert \vert, \theta}^{(2)}} \, .
\end{equation}
%simply from definition.
%Substituting the two-copy
%trace QFI from Eq.~\eqref{eq:twocopytrQFI}
%into
%the first inequality from
%Eq.~\eqref{eq:orderingtwocopy}
%then yields
%\begin{equation}
%\label{eq:orderingtwocopy2}
%    2  \, \QFItr \leq \QFIsp^{(2)} \leq \QFItr^{(2)} \, .
%\end{equation}
%Finally,
%inserting~$\QFIsp\leq \QFItr$
%into Eq.~\eqref{eq:orderingtwocopy2}
%gives us
%\begin{equation}
%\label{eq:orderingtwocopyfinal}
%     2  \, \QFIsp \leq 2  \, \QFItr \leq \QFIsp^{(2)} \leq \QFItr^{(2)} \, .
%\end{equation}
%This proves that~$  \QFIsp^{(2)} \geq 2  \, \QFIsp $,

Although the general
non-additivity
of the DQFI reflects the
added utility entanglement can contribute
to the estimation process,
there are specific cases where
entangled probes are no more
precise than separable ones.
In particular, for phase-insensitive
detectors corresponding to
diagonal measurement operators,
there is no advantage to entangled probes.
Physically, this is because the DQFI here
is additive with respect to number of copies,
and the following lemma
establishes this formally.

\begin{lemma}
\label{lemma:nocollectiveadvantagediagonalmeasurement}
    For estimating a phase-insensitive measurement,
    the two-copy spectral DQFI is additive, i.e.,
    \begin{equation}
        {\mathcal{J}_{\vert \vert, \theta}^{(2)}} = 2 \,  \QFIsp \, ,
    \end{equation}
    and entangled probes offer no precision advantage.
\end{lemma}
\begin{proof}
Diagonal measurement operators~$\{\pi_j\}_{j\in[m]}$
lead to diagonal SLD operators~$\{L_j\}_{j\in[m]}$,
so that the measurement and its derivative commute,
\begin{equation}
    \pi_j L_j = L_j \pi_j = \partial_\theta \pi_j \, .
\end{equation}
The operator~$A = \sum_{j} L_j \pi_j$ thus
equals~$\sum_{j} \partial_\theta \pi_j = \partial_\theta \sum_{j}  \pi_j = 0$.
As~$A=0$, its eigenvalues are~0,
and the inequality in Eq.~\eqref{eq:twocopyspecQFI}
becomes~$2 \QFIsp \leq {\mathcal{J}_{\vert \vert, \theta}^{(2)}} \leq 2 \QFIsp$,
from which the lemma follows.
\end{proof}

The above results
for two-copy estimation
%this supplemental note
can be directly generalised
to simultaneous~$n$-copy estimation.
For the measurement~$\Pi^{\otimes n}$
with POVM elements~$\pi_{j_1} \otimes
\pi_{j_2} \otimes \dots \otimes \pi_{j_n}$,
the SLD operators
are~$L_{j_1} \otimes \mathds{1}_d^{\otimes n-1} + \mathds{1}_d \otimes L_{j_2} \otimes  \mathds{1}_d^{\otimes n-2} + \dots + \mathds{1}_d^{\otimes n-1} L_{j_n}$
where~$\{j_1, \dots, j_n\} \in [m]^{n}$.
And for~$n>2$,~$Q^{(n)}$ is given by
\begin{equation}
\label{eq:Qmatrixncopy}
\begin{split}
    Q^{(n)} &= \tilde{Q}^{(n)} + \tilde{A}^{(n)} +  \tilde{A}^{(n)\dagger} \\
    \tilde{Q}^{(n)} &= \left ( Q \otimes \mathds{1}_d^{\otimes (n-1)} + \mathds{1}_d \otimes Q \otimes \mathds{1}_d^{\otimes (n-2)} + \dots + \mathds{1}_d^{\otimes (n-1)} \otimes Q \right ) \\
    \tilde{A}^{(n)} &= A \otimes A^\dagger \otimes \mathds{1}_d^{\otimes (n-2)} + A \otimes \mathds{1}_d \otimes A^\dagger  \otimes \mathds{1}_d^{\otimes (n-3)}
    + \dots + A \otimes \mathds{1}_d^{\otimes (n-2)} \otimes A^\dagger \\
    &+ \mathds{1}_d \otimes A \otimes A^\dagger \otimes \mathds{1}_d^{\otimes (n-3)} + \dots + \mathds{1}_d^{\otimes (n-2)} A \otimes A^\dagger \, .
\end{split}
\end{equation}
Here~$\tilde{A}^{(n)} + {{\tilde{A}}^{(n)\dagger}}$
corresponds to a sum over all possible terms of the form~$\bigotimes_{j\in[n]} X_j$
with exactly one~$X_j = A$,
exactly one~$X_j = A^\dagger$
and all other~$X_j = \mathds{1}_d$.
There are~$2! \binom{n}{2}$ such terms.
Then~${\mathcal{J}_{\mathrm{Tr},\theta}^{(n)}} = \Tr[Q^{(n)}] =  n d^{n-1} \QFItr$
and~${\mathcal{J}_{\vert \vert, \theta}^{(n)}} = \Vert Q^{(n)} \Vert^2_{\mathrm{sp}}$.
A straightforward
extension of Lemma~\ref{lemm:twoQFInonadd}
and Eq.~\eqref{eq:twocopyspecQFI}
then yields
the~$n$-copy version of Eq.~\eqref{eq:QFIspsandwich},
\begin{equation}
\label{eq:QFIncopySandwich}
    n \QFIsp \leq {\mathcal{J}_{\vert \vert, \theta}^{(n)}} \leq n \QFIsp + n(n-1) \Vert A^\dagger A \Vert_\mathrm{sp}^2 = n \QFIsp + n(n-1) \bigg \Vert  \, \sum_{j,k} \pi_k L_k L_j \pi_j   \bigg \Vert_\mathrm{sp}^2 \, ,
\end{equation}
whereas the ordering relation from Eq.~\eqref{eq:orderingtwocopy}
becomes
\begin{equation}
\label{eq:orderingncopy}
    \frac{1}{d^n} \, {\mathcal{J}_{\mathrm{Tr},\theta}^{(n)}} =
    \frac{n}{d} \, \QFItr
    \leq {\mathcal{J}_{\vert \vert, \theta}^{(n)}} \leq n d^{n-1} \QFItr = {\mathcal{J}_{\mathrm{Tr},\theta}^{(n)}} \leq d^n {\mathcal{J}_{\vert \vert, \theta}^{(n)}} \, .
\end{equation}

From Eqs.~\eqref{eq:orderingtwocopy}
and~\eqref{eq:orderingncopy},
it is evident that the maximum mismatch
between the spectral DQFI
and the trace DQFI for
two-copy detector estimation
is a factor of~$d^2$,
and, more generally, a factor of~$d^n$
for~$n$-copy simultaneous estimation.
This large disagreement
can be attributed to the poor scaling
of the trace DQFI~$\QFItr$ with number of copies~$n$.
This poor scaling is, in turn, explained
by the fact that~$\Pi$ is not
a unit-trace operator, unlike quantum states,
and~$\Tr(\Pi^{\otimes n}) = d^n$ contributes to
the exponential-in-$n$ scaling of
the trace DQFI.
On the other hand,
the spectral DQFI scales
reasonably with~$n$
and is attainable in
a wider range of cases.
Furthermore,
Eqs.~\eqref{eq:orderingtwocopy}
and~\eqref{eq:orderingncopy}
clearly indicate
the spectral DQFI
to be a tighter bound
than the trace DQFI regardless
of the number of copies
of the detector being probed
simultaneously.

\begin{example}
\label{eg:heisenbergscaling}
    Consider estimating parameter~$p$
    (where~$0 \leq p \leq \frac{k^2}{k^2+1}$
    and~$k>0$) from the POVM~$\Pi_p \equiv \{ \pi_{1,p}, \pi_{2,p}\}$
    with elements
    \begin{equation}
        \pi_{1,p} = \begin{pmatrix}
        p & \nicefrac{p}{k} \\ \nicefrac{p}{k} & 1-p
        \end{pmatrix}
        \quad \& \quad
        \pi_{2,p} = \begin{pmatrix}
        1-p & -\nicefrac{p}{k} \\ -\nicefrac{p}{k} & p
        \end{pmatrix}\, .
    \end{equation}
    The two one-copy DQFIs are
    \begin{equation}
        \mathcal{J}_{\Vert, p} = \frac{\mathcal{J}_{\Tr, p}}{2} = \frac{k^2}{p(k^2(1-p)-p)} \, ,
    \end{equation}
    whereas from Eqs.~\eqref{eq:Qmatrix2copy} \&~\eqref{eq:Qmatrixncopy}, the~$n$-copy spectral DQFI
    is
    \begin{equation}
        \mathcal{J}_{\Vert, p}^{(n)} = n \left ( \frac{k^2}{p(k^2(1-p)-p)} -\frac{4}{k^2}\right ) + n^2 \frac{4}{k^2} = n \, \mathcal{J}_{\Vert, p}  + n(n-1) \Vert A^\dagger A\Vert_\mathrm{sp}^2 \, ,
    \end{equation}
    thus saturating the upper-bound in Eq.~\eqref{eq:QFIncopySandwich}
    and achieving a quadratic (Heisenberg) scaling in
    number of copies~$n$.
\end{example}

As for attainability
of the multi-copy DQFIs,
note that whenever~$\Pi$ is diagonal (in some basis),
so is~$\Pi^{\otimes 2}$,
and whenever~$\{L_j \}_{j\in[m]}$ share a common
eigenvector (say~$\ket{\psi}$),
so do~$\{\tilde{L}_{jk} \}_{j,k\in[m]}$
(given by~$\ket{\psi}^{\otimes 2}$).
Thus, whenever the single-copy DQFI~$\QFIsp$
is attainable,
so is the two-copy DQFI~${\mathcal{J}_{\vert \vert, \theta}^{(2)}}$,
and more generally,
the multi-copy DQFI~${\mathcal{J}_{\vert \vert, \theta}^{(n)}}$.
In these cases, the non-additivity
of~$\QFIsp$ mimics that of~${\CFI_\theta}_\mathrm{max}$.
And despite
the relation~${\mathcal{J}_{\vert \vert, \theta}^{(n)}} \leq {\mathcal{J}_{\mathrm{Tr},\theta}^{(n)}}$
holding for any~$n$,
the exponential term~$d^{n-1}$
in~${\mathcal{J}_{\mathrm{Tr},\theta}^{(n)}}$
leads to a much worse
scaling for the trace DQFI
compared to the spectral DQFI.
In conclusion,
the spectral DQFI~$\QFIsp$
%(and the expression in Eq.~\eqref{eq:twocopyspecQFI})
provides the truly attainable
maximum $n$-copy Fisher information
and through its non-additivity,
directly reflects the added utility
of entangled probe states,
in other words,
a collective quantum advantage~\cite{Lorcan21,LorcanGap,Conlon2023,Conlon2023b,Conlon2023c,Das2024}.
Below, in Table~\ref{table:SEvsDE},
we summarise the similarities and differences
that have emerged between single-parameter detector estimation
and single- and multi-parameter state estimation.

\renewcommand{\arraystretch}{1.4}
\begin{table}[h!]
\centering
\begin{tabular}{|l|c|c|c|}
\hline
\textbf{} & \textbf{Single-parameter SE} & \textbf{Single-parameter DE} & \textbf{Multi-parameter SE} \\ \hline
\textbf{Information Measure} & SQFI~$\sQFI$ (scalar) & DQFI~$\QFIsp$ (scalar) & SQFI~$\sQFI$ (operator) \\ \hline
\textbf{Attainable} & Always & Subject to compatibility & Subject to compatibility \\ \hline
\textbf{Multi-copy Scaling} & Additive & Non-Additive & Non-Additive \\ \hline
\textbf{Collective Advantage} & No & Yes & Yes \\ \hline
\textbf{Heisenberg Scaling} & Yes & Yes & Yes \\ \hline
\end{tabular}
\caption{Summary comparing State Estimation (SE) and Detector Estimation (DE)
in the local estimation setting.}
\label{table:SEvsDE}
\end{table}
\renewcommand{\arraystretch}{1}

\section{General measurements producing classical and quantum outputs}
\label{supp:QPT}

%\paragraph{About:}
In this section,
we consider measurements
that produce a classical outcome
as well as a quantum state
in each trial of characterisation.
This includes
mid-circuit measurements~\cite{RRG+22},
weak measurements~\cite{LSP+11,Aharonov1988,Ritchie1991,Pryde2005,Hosten2008,KBR+11} and
non-demolition measurements.
Whereas the POVM formalism
(that is central to our approach)
cannot capture post-measurement
states, the process representation
is valid for such measurements,
now mapping input states to a composite output
space containing both classical outcomes and corresponding
post-measurement states.
The combined larger-dimensional state
can be analysed using state estimation tools
whereas the smaller classical outcome space,
if treated on its own, reduces to the DQFI.
However, the DQFI disregards crucial parameter information
contained in the output state,
and therefore underestimates the true information content.
Nonetheless, the total amount of information
here can be bounded by using a combination
of the SQFI and the DQFI.

Below we first summarise the process
representation of a quantum measurement,
and refer to the optimal channel QFI---the
SQFI of the channel-output state
maximised over separable channel inputs---as
the process QFI~$\QFI_\mathrm{QPT}$.
Then we calculate the process QFI
by following the QPT approach.
Next, we prove that for such measurements,
the process approach is more informative
than the detector approach,
because it extracts parameter information
from not just the measurement outcomes
but also the post-measurement states.
However, we also prove that for some measurements
that do not imprint any
\textit{extra} information
on the post-measurement states
(beyond the information
present in the measurement outcome
distribution),
the DQFI~$\QFItr$
coincides with the process QFI~$\QFI_\mathrm{QPT}$.
Lastly,
we show that the information content
of generalised measurements
can only be accurately characterised by considering
the classical and the quantum outputs
as a whole, and present upper and lower bounds
to this information content.

\subsection{Channel or Process Representation of Measurements}
Our analysis proceeds by treating
the measurement described by~$\Pi_\theta \equiv \{{\pi_j}_\theta\}_{j\in [m]}$
as a quantum channel~$\mathcal{N}^{\Pi_\theta}$,
\begin{equation}
    \mathcal{N}^{\Pi_\theta} \colon \mathcal{H}_d \mapsto \mathbb{C}^m \otimes \mathcal{H}_d  \, , \quad
    \mathcal{N}^{\Pi_\theta} \left ( \rho_\mathrm{in} \right ) =  \sum_{j\in[m]} \ketbra{j}  \otimes \sqrt{{\pi_j}_\theta} \, \rho_\mathrm{in} \, \sqrt{{\pi_j}_\theta} \,
    \equiv   \rho_{\mathrm{out},\theta} \, .
\end{equation}
This channel,~$\mathcal{N}^{\Pi_\theta}$,
maps input state~$\rho_\mathrm{in}\in\mathcal{H}_d$
to~$\rho_{\mathrm{out},\theta} \in \mathbb{C}^m \otimes  \mathcal{H}_d$
with~$\mathbb{C}^m$ the~$m$-dimensional
classical space
representing measurement outcomes.
Note that
the choice of~$\sqrt{{\pi_j}_\theta} \, \rho_\mathrm{in} \, \sqrt{{\pi_j}_\theta} \,$
as the post-measurement state
is not unique---it is
specified by the dilated projection-valued measurement (PVM)
realising the measurement,
rather than the POVM itself.
Also,
the post-measurement
states~$\sqrt{{\pi_j}_\theta} \, \rho_\mathrm{in} \, \sqrt{{\pi_j}_\theta} \,$
are trace sub-normalised to~$p_j = \Tr(\rho_\mathrm{in} {\pi_j}_\theta)$.
By defining their trace-normalised counterparts
as
\begin{equation}
    \rho_{\mathrm{out},\theta}^{(j)} \coloneqq \left ( \sqrt{{\pi_j}_\theta} \, \rho_\mathrm{in} \, \sqrt{{\pi_j}_\theta} \right ) / p_j \, ,
\end{equation}
we can rewrite the channel action
as
\begin{equation}
    \mathcal{N}^{\Pi_\theta} \left ( \rho_\mathrm{in} \right )
    = \rho_{\mathrm{out},\theta} \,
    = \sum_{j\in[m]} p_j \ketbra{j}  \otimes \rho_{\mathrm{out},\theta}^{(j)} \, .
\end{equation}

\subsection{Maximum QFI for Process Estimation: Process QFI}
%\paragraph{QPT QFI Definition:}
The channel output state~$\rho_{\mathrm{out},\theta}$
can be represented block-diagonally,
\begin{equation}
    \rho_{\mathrm{out},\theta} \, = \begin{bmatrix} \sqrt{{\pi_1}_\theta} \, \rho_\mathrm{in} \, \sqrt{{\pi_1}_\theta} & 0 & \cdots \\ 0 & \sqrt{{\pi_2}_\theta} \, \rho_\mathrm{in} \, \sqrt{{\pi_2}_\theta}  & \cdots \\
    \vdots & \vdots & \ddots \end{bmatrix} \,  \equiv \bigoplus_{j\in[m]} p_j {\rho_\mathrm{out}^{(j)}}_\theta \, ,
\end{equation}
and thus has
block-diagonal SLD operators,~$\mathcal{L}_\theta
= \oplus_{j\in[m]} {\mathcal{L}_j}_\theta$.
The SQFI of~$\rho_{\mathrm{out},\theta}$ is then
\begin{equation}
\label{eq:QPTQFIrho}
    \sQFI[\rho_{\mathrm{out},\theta}] = \Tr(\rho_{\mathrm{out},\theta} \mathcal{L}_\theta^2)
    = \sum_{j\in [m]} \Tr( p_j \rho_{\mathrm{out},\theta}^{(j)} {\mathcal{L}_j}_\theta^2)
    = \Tr(\rho_\mathrm{in} \sum_{j\in [m]} \,  \sqrt{{\pi_j}_\theta} {\mathcal{L}_j}_\theta^2  \sqrt{{\pi_j}_\theta} )
\end{equation}
and can be
maximised over all~$\rho_\mathrm{in}$.
Similar to main-text Eq.~\eqref{eq:shortproofEq1},
this maximum can
be upper-bounded by~$\Vert \sum_{j\in [m]} \,  \sqrt{{\pi_j}_\theta} {\mathcal{L}_j}_\theta^2  \sqrt{{\pi_j}_\theta} \Vert_\mathrm{sp}^2 = \Vert \sum_{j\in [m]} \, {\mathcal{L}_j}_\theta \,   {\pi_j}_\theta \, {\mathcal{L}_j}_\theta   \Vert_\mathrm{sp}^2 $
but this expression is still
dependent on~$\rho_\mathrm{in}$ via~${\mathcal{L}_j}_\theta$.
Thus, without a general closed form,
we need to compute the ultimate QPT precision bound using
the process QFI, which we define as
\begin{equation}
\label{eq:QPTstateestSuppNote}
    \QFI_\mathrm{QPT}\left [ \Pi_\theta \right ]  \coloneqq \max_{\rho_\mathrm{in}\in \mathcal{D}(\mathcal{H}_d)} \sQFI [\rho_{\mathrm{out},\theta}] \, .
\end{equation}

\subsection{Process QFI Larger than DQFI}
%\paragraph{QPT More Informative in General:}
In general,
the process QFI~$\QFI_\mathrm{QPT}\left [ \Pi_\theta \right ]$
is larger than our DQFIs~$\QFItr \left [ \Pi_\theta \right ]$
and~$\QFIsp \left [ \Pi_\theta \right ]$,
because access to the post-measurement states
can only be more informative than simply
estimating from the measurement outcomes.
We provide a short proof
of this claim.
Note that the reduced
state on the first sub-system
of~$\rho_{\mathrm{out},\theta}$
is the classical state
\begin{equation}
    \Tr_{\mathcal{H}_d} \left [ \rho_{\mathrm{out},\theta} \right ] = \sum_{j\in[m]} p_j \ketbra{j} \, ,
\end{equation}
where~$\Tr_{\mathcal{H}_d}$ denotes
partial-tracing out the second sub-system.
The SQFI of this classical state
is simply the CFI of the distribution~$\{p_j\}_{j\in [m]}$.
As partial-tracing is a completely-positive
trace-preserving operation, the SQFI of a state
cannot increase under partial-tracing~\cite{Wilde13}.
This proves that for any given probe state~$\rho_\mathrm{in}$,
\begin{equation}
\label{eq:CFIlessQFI}
    \CFI_\theta \left [ \{ p_j \}_{j \in [m]} \right ] \leq \, \sQFI \left [ \rho_{\mathrm{out},\theta} \right ].
\end{equation}
We now need to maximise Eq.~\eqref{eq:CFIlessQFI}
over all probe states~$\rho_\mathrm{in}\in \mathcal{D}(\mathcal{H}_d)$,
but the optimal probe state might be different for~$\CFI_\theta$
and for~$\sQFI$.
Instead, we first note that the inequality holds for
the~$\CFI_\theta$-maximising input state~$\rho^\mathrm{opt}$,
as defined in Eq.~\eqref{eq:maxCFIdef}.
Further, the output state SQFI~$\sQFI [\rho_{\mathrm{out},\theta}]$
when probing with~$\rho^\mathrm{opt}$
is at most equal to the maximum~$\sQFI$
over all states~$\rho_\mathrm{in}\in\mathcal{D}(\mathcal{H}_d)$.
Therefore, we have
$    {\CFI_\theta}_\mathrm{max}
    \leq \QFI_\mathrm{QPT}$.
For the attainable families
of measurements, this means that
the DQFI presents a tighter bound
than the process approach
if
post-measurement states
are inaccessible
for characterisation.

\subsection{Process QFI Equal to DQFI}
For a class of measurements,
the process QFI~$\QFI_\mathrm{QPT}$ and the DQFI~$\QFIsp$
agree on
the parameter information
content of measurements.
These are measurements
for which
the CFI of measurement outcomes
coincides with the SQFI of the post-measurement
state,
at least for the optimal probe state.
As a special case,
we now prove that for measurements
parametrised such that~${\pi_j}_\theta$
and~$\partial_\theta {\pi_j}_\theta$
commute,
the two approaches yield the same QFI.
The proof is broken
up into three steps for readability.

\step{1}
Note that the detector
SLD operators were defined in the main-text as
\begin{equation}
    {L_j}_\theta {\pi_j}_\theta + {\pi_j}_\theta {L_j}_\theta = 2 \partial_\theta {\pi_j}_\theta \, ,
\end{equation}
whereas the QPT SLD operators
are defined via
\begin{equation}
\label{eq:defsldQPT}
    {\mathcal{L}_j}_\theta \sqrt{{\pi_j}_\theta} {\rho_\mathrm{in}} \sqrt{{\pi_j}_\theta} + \sqrt{{\pi_j}_\theta} {\rho_\mathrm{in}} \sqrt{{\pi_j}_\theta} {\mathcal{L}_j}_\theta
    = 2 \partial_\theta  ( \sqrt{{\pi_j}_\theta} {\rho_\mathrm{in}} \sqrt{{\pi_j}_\theta} ) \, .
\end{equation}
Clearly, for~$\rho_\mathrm{in} = \mathds{1}_d/d$,
i.e., the maximally-mixed probe state,
the two definitions are identical
meaning~$\mathcal{L}_j = L_j$.

\step{2}
As~${\pi_j}_\theta$
commutes with~$\partial_\theta {\pi_j}_\theta$,
we can write
$$\partial_\theta \sqrt{{\pi_j}_\theta}
= \frac12 \, {\pi_j}_\theta^{-\frac12} \left ( \partial_\theta {\pi_j}_\theta \right)
=  \frac12  \,  \left (  \partial_\theta {\pi_j}_\theta \right) {\pi_j}_\theta^{-\frac12} \, ,$$
where~${\pi_j}_\theta^{-\frac12}$ denotes the
inverse (or pseudo-inverse) of
the unique positive, Hermitian square-root~$\sqrt{{\pi_j}_\theta}$
of~${\pi_j}_\theta$.
Thus, the right-hand-side of Eq.~\eqref{eq:defsldQPT}
becomes
\begin{equation}
    \begin{gathered}
        2 \partial_\theta (\sqrt{{\pi_j}_\theta} \,\rho_\mathrm{in} \, \sqrt{{\pi_j}_\theta})  \\
        =
         {\pi_j}_\theta^{-\frac12} \, \partial_\theta {\pi_j}_\theta \, \rho_\mathrm{in} \, {\pi_j}_\theta^{\frac12}
         + {\pi_j}_\theta^{\frac12} \, \rho_\mathrm{in}  \, \partial_\theta  {\pi_j}_\theta \, {\pi_j}_\theta^{-\frac12} \\
         =  {\pi_j}_\theta^{-\frac12} \,  \partial_\theta {\pi_j}_\theta  {\pi_j}_\theta^{-\frac12} \sqrt{{\pi_j}_\theta} {\rho_\mathrm{in}} \sqrt{{\pi_j}_\theta}
         + \sqrt{{\pi_j}_\theta} {\rho_\mathrm{in}} \sqrt{{\pi_j}_\theta} {\pi_j}_\theta^{-\frac12} \,  \partial_\theta {\pi_j}_\theta  {\pi_j}_\theta^{-\frac12} \, ,
    \end{gathered}
\end{equation}
comparing which with the left-hand-side of Eq.~\eqref{eq:defsldQPT},
we conclude~${\mathcal{L}_j}_\theta = {\pi_j}_\theta^{-\frac12} \partial_\theta {\pi_j}_\theta  {\pi_j}_\theta^{-\frac12} $.
Moreover, as~${\pi_j}_\theta$ and~$\partial_\theta {\pi_j}_\theta$ commute,
\begin{equation}
\label{eq:term4eq1}
    2 ( {\pi_j}_\theta^{-1} \partial_\theta {\pi_j}_\theta - \partial_\theta {\pi_j}_\theta {\pi_j}_\theta^{-1}) = 0 \implies
    {\pi_j}_\theta^{-1} {L_j}_\theta {\pi_j}_\theta ={\pi_j}_\theta {L_j}_\theta {\pi_j}_\theta^{-1} \, ,
\end{equation}
where~${\pi_j}_\theta^{-1}$ represents
the inverse (or pseudo-inverse)
of~${\pi_j}_\theta$.

\step{3}
We can now rewrite the term inside the trace
in the process QFI from Eq.~\eqref{eq:QPTQFIrho}
in terms of the detector SLD operators as
\begin{gather*}
\sqrt{{\pi_j}_\theta} {\mathcal{L}_j}_\theta^2  \sqrt{{\pi_j}_\theta}
= \partial_\theta {\pi_j}_\theta \, {\pi_j}_\theta^{-1} \, \partial_\theta {\pi_j}_\theta \\
= \frac{1}{4} \left ( ({L_j}_\theta {\pi_j}_\theta + {\pi_j}_\theta {L_j}_\theta) {\pi_j}_\theta^{-1} ({L_j}_\theta {\pi_j}_\theta + {\pi_j}_\theta {L_j}_\theta) \right ) \\
= \frac14 (  {\pi_j}_\theta {L_j}_\theta^2 +  {L_j}_\theta^2 {\pi_j}_\theta  + {L_j}_\theta {\pi_j}_\theta {L_j}_\theta +  {\pi_j}_\theta {L_j}_\theta  {\pi_j}_\theta^{-1} {L_j}_\theta  {\pi_j}_\theta ) \, .
\end{gather*}
Thus, the maximum of
the SLD QFI~$\Tr(\rho_{\mathrm{out},\theta} \mathcal{L}_\theta^2)$
over probe states~$\rho_\mathrm{in}$
is the largest eigenvalue
of the matrix
\begin{equation}
\label{eq:bigmatrixqpt}
    \frac14 \sum_{j \in [m]} (  {\pi_j}_\theta {L_j}_\theta^2 +  {L_j}_\theta^2 {\pi_j}_\theta  + {L_j}_\theta {\pi_j}_\theta {L_j}_\theta +  {\pi_j}_\theta {L_j}_\theta  {\pi_j}_\theta^{-1} {L_j}_\theta  {\pi_j}_\theta ) \, .
\end{equation}
Using the fact that
the eigenvalues of~$AB$ and~$BA$
are the same for any matrices~$A$ and~$B$,
we conclude that the largest eigenvalue of~$ {\pi_j}_\theta {L_j}_\theta^2$
equals that of~${L_j}_\theta^2 {\pi_j}_\theta$
and that of~${L_j}_\theta {\pi_j}_\theta {L_j}_\theta$.
For the last term in Eq.~\eqref{eq:bigmatrixqpt},
we use Eq.~\eqref{eq:term4eq1}
to write
\begin{equation}
    {\pi_j}_\theta {L_j}_\theta  {\pi_j}_\theta^{-1} {L_j}_\theta  {\pi_j}_\theta = {\pi_j}_\theta {L_j}_\theta  {\pi_j}_\theta {L_j}_\theta  {\pi_j}_\theta^{-1}
\end{equation}
and then argue that the eigenvalues
of~$ ( {\pi_j}_\theta {L_j}_\theta  ) \, ( {\pi_j}_\theta {L_j}_\theta  {\pi_j}_\theta^{-1}) $
should equal those
of~$ ( {\pi_j}_\theta {L_j}_\theta  {\pi_j}^{-1}) \, ( {\pi_j}_\theta {L_j}_\theta  )
=   {\pi_j}_\theta {L_j}_\theta {L_j}_\theta$
and hence equal those of~$ {L_j}_\theta {\pi_j}_\theta {L_j}_\theta$.
Thus, the maximum eigenvalue of
the matrix in Eq.~\eqref{eq:bigmatrixqpt}
is equal to the
maximum eigenvalue of
$\sum_{j\in [m]} {L_j}_\theta {\pi_j}_\theta {L_j}_\theta$,
which is exactly the spectral DQFI~$\QFIsp$.
This proves that for measurements
with~${\pi_j}_\theta$
and~$\partial_\theta {\pi_j}_\theta$
commuting,
\begin{equation}
        \QFI_\mathrm{QPT}\left [ \Pi_\theta \right ] = \max_{\rho_\mathrm{in}} \sQFI [\rho_{\mathrm{out},\theta}] = \QFIsp \left [ \Pi_\theta \right ] \, .
\end{equation}
The optimal probe state according
to both approaches also agree---it
is the eigenvector of~$\sum_{j\in [m]} {L_j}_\theta {\pi_j}_\theta {L_j}_\theta$
corresponding to its largest eigenvalue.

\subsection{Information Content of General Measurements}

\begin{figure}[hbtp]
    \centering
    \includegraphics[width=0.5\columnwidth]{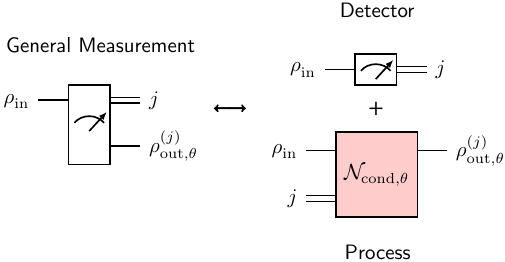}
    \caption{Information content of output from general measurements.
    The SQFI of the combined channel output state equals the sum of
    the CFI of the outcome distribution and the expected SQFI
    of the conditional output states.}
    \label{fig:generalmeasurements}
\end{figure}

The defining equation~\eqref{eq:defsldQPT} for
the process SLD operators~${\mathcal{L}_j}_\theta$
can be rewritten, to explicitly include the probabilities
of measurement and the normalised conditional output states,
as
\begin{equation}
    {\mathcal{L}_j}_\theta \rho_{\mathrm{out},\theta}^{(j)} +  \rho_{\mathrm{out},\theta}^{(j)} {\mathcal{L}_j}_\theta
    = 2 \frac{\partial_\theta  ( {p_j}_\theta \rho_{\mathrm{out},\theta}^{(j)} )}{{p_j}_\theta}
    = 2 \frac{\partial_\theta  ( {p_j}_\theta  )}{{p_j}_\theta } \rho_{\mathrm{out},\theta}^{(j)}
    + 2 \partial_\theta  (  \rho_{\mathrm{out},\theta}^{(j)} )\, .
\end{equation}
The first and the second terms on the right hand side
represent the change of the process
reflected in the change of output probabilities
and in the change of conditional output states,
respectively.
The classical part of the anti-commutator (or Lyapunov) equation
for~${\mathcal{L}_j}_\theta$
is solved by~$\nicefrac{\partial_\theta {p_j}_\theta}{{p_j}_\theta} \mathds{1}_d$.
We may therefore assume~${\mathcal{L}_j}_\theta \coloneqq \nicefrac{\partial_\theta {p_j}_\theta}{{p_j}_\theta} \mathds{1}_d + {\mathcal{L}'_j}_\theta$,
which makes~${\mathcal{L}'_j}_\theta$ the SLD operator
for~$\rho_{\mathrm{out},\theta}^{(j)}$, the quantum output state
conditioned on the~$j^\text{th}$ classical outcome.
%Say we define~${\mathcal{L}_j}_\theta \coloneqq \frac{\partial_\theta %{p_j}_\theta}{{p_j}_\theta} \mathds{1}_{d} + {\mathcal{L}'_j}_\theta$.
% Then~${\mathcal{L}'_j}_\theta$ are the usual SLD operators
% for each conditional quantum output state.
As a result, the output SQFI equals
\begin{equation}
    \sQFI[\rho_{\mathrm{out},\theta}]
    = \sum_{j\in [m]} {p_j}_\theta  \Tr( \rho_{\mathrm{out},\theta}^{(j)} {\mathcal{L}_j}_\theta^2)
    = \CFI_\theta[\{ {p_j}_\theta \}] + \sum_{j\in[m]} {p_j}_\theta \sQFI[\rho_{\mathrm{out},\theta}^{(j)}] \, .
\end{equation}
It is specifically the second term
here, the convex combination of the output state SQFIs,
with respect to the distribution of measurement outcomes,
that the detector approach ignores.
The spectral DQFI~$\QFIsp$ still upper-bounds~$\CFI_\theta[\{ {p_j}_\theta \}]$
but the second term
may be maximised for some probe state other
than the CFI-optimal one.
Nonetheless,
an upper bound to~$\QFI_\mathrm{QPT} \equiv \max_{\rho_\mathrm{in}} \sQFI[\rho_{\mathrm{out},\theta}]$
can be obtained by treating the two terms independently.
The first term is upper-bounded by~${\CFI_\theta}_\mathrm{max}$,
which itself is upper-bounded by the DQFI of
the effective POVM.
The second term can be interpreted
as the expected SQFI of the conditional
output states. Being a convex sum, this term is therefore upper-bounded
by the largest of the~$m$ different conditional-output SQFIs, i.e.,~$\max_{j\in[m]} \sQFI[\rho_{\mathrm{out},\theta}^{(j)}]$.
Computing the maximum of this quantity
over input probes~$\rho_\mathrm{in}$
is therefore equivalent
to using the standard approach to
optimal process estimation for each of the~$m$
different conditional channels
(see Fig.~\ref{fig:generalmeasurements}),
computing the process QFI for each conditional channel,
and then choosing the maximum out of these.
Efficient characterisation here would require
measuring the quantum output conditioned on the
classical outcomes, while probing with the optimal
states~\cite{CN97,Wilde13}.

%\bibliographystyle{naturemag}
%\bibliography{sample}

\begin{thebibliography}{10}
\bibitem{footnote1}
\bibinfo{title}{Except if the largest eigenvalue is degenerate; see discussion in attainability section of methods.}
\bibitem{footnote2}
\bibinfo{title}{Normalising here does not produce an upper bound to~${\CFI_\theta}_\mathrm{max}$, e.g., see example~\ref{eg:example2multi}.}
\bibitem{footnote3}
\bibinfo{title}{Transforming to a parameter-dependent basis does not generally preserve the cfi, because although the measurement probabilities~$p(j\vert\rho)$ are unchanged, their derivatives~$\partial_\theta p(j\vert \rho)$ are changed in the process.}
\end{thebibliography}

\end{document}